\documentclass[11pt]{amsart}
\usepackage{amsaddr}
\usepackage{color}
\usepackage{cite}
\usepackage{amssymb}
\usepackage{amsthm}
\usepackage{amsmath,bm}
\usepackage{mathrsfs}
\usepackage{amsfonts}
\usepackage{todonotes}
\usepackage{graphicx}
\usepackage[subrefformat=parens, labelfont=up]{subcaption} 
\usepackage{epstopdf}
\usepackage{algorithmic}
\usepackage{multirow}
\usepackage{multicol}
\usepackage{cases}
\usepackage{float}
\usepackage{ragged2e}
\usepackage{geometry}
\usepackage{booktabs}

\newtheorem{assumption}{Assumption}
\newtheorem{definition}{Definition}
\newtheorem{lemma}{Lemma}
\newtheorem{proposition}{Proposition}
\newtheorem{theorem}{Theorem}
\newtheorem{algorithm}{Algorithm}
\newtheorem{remark}{Remark}
\newtheorem{corollary}{Corollary}
\def\ud{\mathrm d}
\def\e{\mathrm e}

\title{A model of strategic sustainable investment}

\author[De Angelis]{Tiziano De Angelis}

\address{School of Management and Economics, ESOMAS, University of Turin\\ and Collegio Carlo Alberto, Turin, Italy.} \email{tiziano.deangelis@unito.it}

\author[Graciani Rodrigues]{Caio C\'esar Graciani Rodrigues}

\address{Department of Modeling and Engineering Risk and Complexity,\\ Scuola Superiore Meridionale, Naples, Italy.} 
\email{c.graciani@ssmeridionale.it}%

\author[Tankov]{Peter Tankov}%
\address{CREST, ENSAE, Institute Polytechinique de Paris, Paris, France.}
\email{peter.tankov@ensae.fr}%

\begin{document}


\begin{abstract}

We study a problem of optimal irreversible investment and emission reduction formulated as a nonzero-sum dynamic game between an investor with environmental preferences and a firm. The game is set in continuous time on an infinite-time horizon. The firm generates profits with a stochastic dynamics and may spend part of its revenues towards emission reduction (e.g., renovating the infrastructure). The firm's objective is to maximize the discounted expectation of a function of its profits. The investor participates in the profits, may decide to invest to support the firm's production capacity and uses a profit function which accounts for both financial and environmental factors. Nash equilibria of the game are obtained via a system of variational inequalities. We formulate a general verification theorem for this system in a diffusive setup and construct an explicit solution in the zero-noise limit. Our explicit results and numerical approximations show that both the investor's and the firm's optimal actions are triggered by moving boundaries that increase with the total amount of emission abatement.

\medskip




\end{abstract}

\maketitle

\section{INTRODUCTION}

As acknowledged in the 2015 Paris Agreement, making finance flows consistent with a pathway towards low greenhouse gas emissions is a key step towards decarbonizing the economy with the aim of limiting the global warming well below 2$^\circ$ C compared to pre-industrial levels. Sustainable investors, who care not only about the financial performance but also about the environmental performance of their assets, play a key role, as they finance the green companies and create incentives for brown companies to reduce their emissions. However, the interests of sustainable investors may be misaligned with those of the company management, especially if executive compensation schemes do not account for environmental aspects. The process of investment can thus be seen as a nonzero-sum game, where the investors provide capital to companies, aiming to maximize both environmental and financial performance, and company managers determine mitigation strategies, aiming to maximize financial performance, but taking into account the capital provision by the sustainable investors. 

In this paper, we develop a model for green investment and emission reduction, framing it as a dynamic game between a representative investor and a representative privately owned company. The investor continuously provides capital directly to the company, and the company determines its emission abatement strategy. The optimization criterion of the investor is concave increasing in both financial and environmental performance of the company and decreasing in the cost of capital provided to the company. Instead, the company maximizes the discounted expectation of its future financial performance. We formalize market equilibrium as the Nash equilibrium of this stochastic game, and characterize the value functions of the company and of the investor as the solution of a system of variational inequalities, via a verification theorem. 

We formulate a general verification theorem in a diffusive set-up. When the diffusion coefficient is set to zero the game becomes deterministic and we are able to produce an explicit solution of the variational problem and the unique equilibrium for the game, which is also explicit. In the general stochastic game we design a numerical algorithm to solve the problem using a variant of policy iteration. Remarkably, the structure of the equilibrium we obtain from the algorithm is qualitatively the same as the one for the equilibrium we obtain explicitly in the deterministic case. That shows that our explicit solution may be used as a proxy for the solution to the general problem.

Our results show that both the investor's and the firm's optimal actions are triggered by moving boundaries that increase with the total amount of emission abatement performed by the firm. More precisely, denoting by $X_t$ the firm's {production capacity}  at time $t\ge 0$ and by $R_t$ the total abatement performed up to time $t\ge 0$ we determine two functions $r\mapsto a(r)$ and $r\mapsto b(r)$, with $a(r)<b(r)$, that fully characterize the equilibrium strategies of the investor and of the firm, respectively. In particular, when the {production capacity} is high and thus the financial performance of the firm is very good, relatively to the current abatement level $R_t$ (i.e., for $X_t>b(R_t)$) neither the firm nor the investor act. However, when the current {production capacity} is more modest, the firm finds it convenient to face abatement costs in order to attract future investments.  This translates into the fact that when $X_t\le b(R_t)$ the firm invests in pollution abatement at the maximum allowed rate, shifting the dynamics of the pair $(X_t,R_t)$ closer to the investor's optimal investment boundary $a(R_t)$. In turn, the investor guarantees to the firm an optimal level of investment to keep the cashflow dynamics above the increasing threshold $t\mapsto a(R_t)$. This structure of the solution is new but it is in line with the qualitative findings of the literature on impact investing, which argues that sustainable investors achieve greatest impact when they allocate capital to small impactful companies with high growth potential rather than large established ones \cite{kolbel2020can}. 

When the firm's production capacity equals the investment threshold, our model shows that, in equilibrium, by abating emissions, the firm can improve its financial performance, because better environmental performance attracts new investment. This provides an alternative explanation for the positive relationship between the financial and environmental performance of firms in the presence of impact investors, discussed very recently by Green and Roth \cite{green2024allocation}, who show that firms can raise their social impact by improving their financial performance.

It is worth emphasizing that the structure of our equilibrium differs from the one typically found in the literature on nonzero-sum games of irreversible investment (see next section). In that literature, the players’ action regions are usually separated by an inaction region. For instance, in \cite{ferrari2019strategic}, the firm intervenes when profits are low, whereas the regulator acts when profits are high; between these two regions, neither player finds it optimal to act.

In contrast, our model is designed to capture the strategic interaction between impact investors and firms and relies on a different incentive structure. In particular, the investor’s objective depends on the environmental performance variable controlled by the firm. This interdependence alters the geometry of equilibrium. For sufficiently high profit levels, neither player intervenes, but for low profit values their action regions overlap, with the firm’s action set strictly containing that of the investor. As a consequence, the economic implications of our model differ from the standard predictions, highlighting mechanisms that are new -- or at least uncommon -- in the literature on nonzero-sum investment games.

Our mathematical contribution in the deterministic economy is also novel. We establish a correspondence between the unique equilibrium in our game and a solution to a system of first-order, fully degenerate, semilinear, partial differential equations in the domain $(0,\infty)^2$ (cf.\ Corollary \ref{cor:strong}). Standard PDE methods do not seem to be directly applicable in our setup to determine existence (and uniqueness) of a solution to the PDE and so we rely on different methods rooted in the theory of stochastic control. Besides, we go beyond existence of a solution by explicitly characterizing the free boundaries for both the investor and the firm (up to numerical root-finding). That enables explicit construction of optimal strategies for both players. 

\subsection{Review of literature} From the applied point of view, our paper contributes to the burgeoning financial literature on sustainable investment.  Sustainable investors are motivated either by the non-pecuniary benefits they get from holding green stocks (warm glow), as in, e.g., \cite{pastor2021sustainable,pedersen2021responsible} or by their concern for the provision of public goods (e.g., climate change mitigation) by the companies they invest in as, e.g., in \cite{oehmke2023theory,chowdhry2019investing}. This concern for public good provision may be driven by the genuine concern of investors for the environment, but also by the purely financial concern that companies with bad environmental performance are exposed to higher risk, as demonstrated, e.g., in \cite{bolton2021investors}. Our paper belongs to the second category and we therefore contribute to the impact investing literature \cite{oehmke2023theory,landier2024socially,broccardo2022exit,green2024allocation}. Indeed, our investor has the profile of a socially responsible, or impact investor, ready to trade-off financial performance against reductions in social costs and able to influence the abatement decisions of the firm \cite{oehmke2023theory}. Similarly to \cite{oehmke2023theory}, we take the corporate finance point of view and study the interaction between a privately-owned firm and its pool of investors. 
 Unlike the above reference and most other papers on the topic, and similarly to \cite{de2023climate}, we consider a dynamic, continuous-time model and frame our investment problem as a dynamic game between the firm and its investors. The dynamic formulation allows to shed light on non-trivial strategies for the firm and the investor, not previously found in the literature, namely the asynchronous sequence of actions between the investor and the firm, based on action thresholds that we identify.
Also, while the greater part of papers in the sustainable investment literature, including \cite{de2023climate}, focus on publicly owned companies we do not set up a financial market, and consider a representative investor who provides capital directly to a representative privately owned company. 

To the extent that our paper aims to determine optimal abatement trajectories, it bears some relation to the integrated assessment modeling literature, starting from the seminal DICE model of Nordhaus \cite{nordhaus1992optimal} to more recent developments \cite{crost2014optimal,jensen2014optimal,hambel2021optimal}. Beyond its microeconomic focus, our framework differs from this literature in that we do not attempt to model climate damages but measure the firm's environmental performance directly in terms of its abatement spending. Optimal abatement dynamics have also been studied in the context of emission trading schemes \cite{carmona2010market,aid2023optimal}, where the interaction typically involves a large number of emitting firms and a regulator determining the optimal supply of emissions allowances in the presence of a carbon market. By contrast, our paper examines a decentralized mechanism: we analyze how an impact investor can incentivize a firm to reduce emissions through purely financial market instruments, without relying on regulatory intervention. 

We also contribute to the literature on irreversible investment. The vast majority of papers in this domain adopt the point of view of one firm (or several firms) which looks for an optimal investment level to maximize uncertain future profits \cite{pindyck1991irreversibility,dixit1995irreversible,kogan2001equilibrium,decamps2006irreversible,guo2005irreversible}. The firm is considered as a single economic entity and it is assumed that the interests of the investors are aligned with those of the managers of the firm. However, in the presence of environmental concerns this may no longer be the case: as evidenced by the literature on shareholder activism \cite{barko2022shareholder,heeb2023investors,hart2017companies}, the investors may care for the environmental impact of the company much more than the managers. We therefore frame our irreversible investment problem as a game between a representative company and a representative investor, who pursue different objectives. 

From the more theoretical viewpoint, we contribute to the literature on nonzero-sum dynamic games with singular controls \cite{de2018stochastic,DeA23,kwon2015game}. Those games arise in, e.g., irreversible investment problems \cite{back2009open} and constitute the class of two-player games that is perhaps the most difficult to analyze, with few contributions going beyond a verification theorem or numerical illustrations \cite{aid2020nonzero,ferrari2019strategic}. The reference \cite{ferrari2019strategic}, which describes a game of pollution control between the government and a representative firm is, perhaps, closest in spirit to our work. However, in \cite{ferrari2019strategic} the emissions are directly linked to the production process, the only decision the firm takes is to expand production, and the state variable is one-dimensional. Instead we consider separately the financial and the environmental performance of the firm, leading to a state process with two components, both of which are controlled by the agents. Moreover, in our model the investor uses singular controls and the firm uses classical controls in order to guarantee well-posedness of the game (cf.\ Remarks \ref{rem:wellposed} and \ref{rem:wellposed2}). 

\subsection{Structure of the paper} 
In Section \ref{sec:model} we formulate our model and introduce the notion of Nash equilibrium in our game. Subsections \ref{sec:main1} and \ref{sec:main2} contain the statements of our main results (Theorems \ref{thm:main0} and \ref{thm:verif}). In Section \ref{sec:verif} we prove the general verification theorem (Theorem \ref{thm:verif}) with some ancillary considerations in Corollary \ref{cor:verif} and Remark \ref{rem:verif}. In Section \ref{sec:deterministic} we construct an explicit equilibrium in our game when there is no diffusive component. That proves Theorem \ref{thm:main0}. The construction relies partially upon the use of Theorem \ref{thm:verif} and partially upon the use of Corollary \ref{cor:verif} and Remark \ref{rem:verif}. In particular, the investor's optimal strategy and equilibrium payoff are obtained explicitly from concavity / first order conditions (without using PDE arguments), while the firm's payoff and optimal strategy are characterized via a Hamilton-Jacobi-Bellman equation. After we obtain both the firm's and the investor's equilibrium payoffs, we are able to show that also the investor's payoff solves an Hamilton-Jacobi-Bellman equation in a slightly relaxed form (cf.\ Section \ref{sec:strong}). In Section \ref{sec:equilibrium} we devise a numerical algorithm for the solution of the variational problem associated to our game as presented in the verification theorem. In Section \ref{sec:numerics} we illustrate our numerical findings for the general stochastic game. A summary of the paper and some conclusions are provided in Section \ref{sec:concl}. In Appendix \ref{appendix.sec} we provide practical
guidance on how each parameter of our model could be mapped to observable quantities.

\section{Model formulation {and main results}\label{sec:model}}

\subsection{{ Model formulation}} 
Let $(\Omega, \mathcal{F}, \mathbb{F}, \mathbb{P})$ with $\mathbb F=(\mathcal{F}_t)_{t\ge 0}$ be a filtered probability space. Let $(B_t)_{t \geq 0}$ be a one-dimensional Brownian motion and let $\mathbb{E}[\cdot]$ denote the expectation under $\mathbb P$.
We use the standard notations $\mathbb{R}\triangleq (-\infty,\infty)$ and $\mathbb{R}_+ \triangleq (0,\infty)$ and set $(z)^+=[z]^+ \triangleq \max(0,z)$. Finally, given a set $A\subseteq[0,\infty)\times\mathbb R_+$ we denote by $A^c$ its complement, by $\overline A$ its closure relative to $[0,\infty)\times\mathbb R_+$ and by $\partial A=\overline A\cap\overline{A^c}$ its boundary.

In our model, we consider a representative firm, which
produces a single good {and has a production capacity} denoted as $(X_t)_{t\ge 0}$. We assume that the firm generates a continuous stream of profits, which is proportional to its production capacity (this is consistent with the widely used AK-model in economics, where the output is equal to the product of capital and a constant productivity factor). Without loss of generality, we assume the proportionality factor to be equal to one, that is, we shall interchangeably refer to $X_t$ both as the production capacity and the instantaneous profit of the firm.

The firm is privately owned by a pool of investors, described in our model as a single representative investor, who can inject funds directly into the company, increasing its production capacity and hence, its profit stream.
The investor's decisions are guided by financial and environmental considerations. {The financial performance of the firm is measured by a profit function applied to the production capacity process $(X_t)_{t\ge 0}$. 
The environmental performance is measured by} the total expenditure $(R_t)_{t\ge 0}$, allocated by the firm towards emission reduction activities. While the actual emission reduction is not a linear function of abatement expenditure and typically has a concave profile determined by the marginal abatement cost (MAC) curve \cite{kesicki2012marginal}, the MAC curves are company-specific, difficult to estimate and have a number of other methodological drawbacks \cite{kesicki2012marginal}. The abatement expenditure, on the other hand, is less directly related to environmental impact of investments, but is easy to define, measure and compare to other financial quantities in the model. 

The problem is formulated as a stochastic differential game of optimal controls. In the absence of any intervention from either the firm or the investor, the dynamics of the {production capacity} process $(X_t)_{t\ge 0}$ is given by 
\begin{equation}
X^0_t=X_0+\int_0^t \mu X^0_s\mathrm{d} s+\int_0^t \sigma X^0_s \ud B_s,\quad X_0=x>0,
\end{equation}
where $\mu\in \mathbb{R}$ and $\sigma\in \mathbb{R}_+$ are model parameters. The investor picks a control from the class
\begin{equation}
\mathcal A_I\triangleq\left\{\nu:
\begin{array}{c}
(\nu_t)_{t\ge 0}\ \text{is c\`adl\`ag, nondecreasing, $\mathbb F$-adapted,}\\
\text{with $\nu_0\ge 0$ and $\mathbb E\Big[\int_{[0,\infty)} \e^{-\rho t} \ud\nu_t\Big]<\infty$}
\end{array}
\right\}.
\end{equation}

For $\nu\in\mathcal A_I$, the random variable $\nu_t$ corresponds to  the total amount of investment that the investor has provided to the firm over the time interval $[0,t]$, where $\nu_0$ corresponds to the initial investment.
We notice that investment may arrive with lump-sum payments, i.e., it may be $\Delta \nu_t\triangleq \nu_t-\nu_{t-}>0$ for $t>0$, and there is no cap on the investment rate (this corresponds to the situation of so-called {\em singular controls}). Any admissible control $\nu\in\mathcal A_I$ admits a decomposition $\nu_t=\nu^c_t+\sum_{s\le t}\Delta\nu_s$ in a continuous part plus a sum of jumps. 

The dynamics for the total expenditure allocated by the firm towards emission reduction
reads:
\begin{equation}\label{eqR}
R^\eta_t=R_0+\int_0^t \eta_s\ud s,\quad R_0=r\ge 0,
\end{equation}
where the process $(\eta_t)_{t\ge 0}$ is the control chosen by the firm and it belongs to the class, for some $\eta_{\mathrm{max}}>0$,
\begin{equation}
\begin{aligned}
\mathcal A_F\triangleq\big\{\eta:&\,(\eta_t)_{t\ge 0}\ \text{is progressively measurable, with $0\le \eta_t\le \eta_{\mathrm{max}}$}\big\}.
\end{aligned}
\end{equation}

Any pair $(\nu_t,\eta_t)_{t\ge 0}$ describes the investment and emission reduction policies of the two agents. For a fixed choice of $(\nu_t,\eta_t)_{t\ge 0}$, the dynamics of the production capacity reads:
\begin{equation}\label{eqX}
X_t^{\nu,\eta} = X_0+\int_0^t \mu X_s^{\nu,\eta} \ud s + \int_0^t \sigma X_s^{\nu,\eta} \ud B_s + \nu_t - \int_0^t\eta_s \ud s, 
\end{equation}
with initial condition $X_0 = x$ before a possible lump-sum investment $\nu_0$ at time zero. Investment increases the production capacity of the company whereas abatement decreases it: to reduce emissions, the company switches to a greener but costlier technology or modifies its energy mix by using more expensive clean energy. Note that for simplicity and without loss of generality, both investment and abatement are measured in units of production capacity.
By an application of It\^o's formula it is readily verified that the dynamics of $X^{\nu,\eta}$ can be written more explicitly as 
\begin{equation}\label{eq:Xexpl}
X_t^{\nu,\eta}=X^0_t\Big(1+\int_{[0,t]} \frac{1}{X^0_s}\big(\ud \nu_s-\eta_s\ud s\big)\Big).
\end{equation}

The firm's optimization criterion is given in terms of expected discounted future { profits}. Mathematically we express it as
\begin{equation}\label{eqvaluew}
\mathcal{J}^F_{r,x}(\eta,\nu) \triangleq \mathbb{E}_{r,x} \left[ \int_0^\infty \e^{- \bar\rho t} \pi(X_t^{\nu,\eta})\ud t \right],
\end{equation}
where $\bar\rho>0$ is the discount rate of the firm, $\mathbb{E}_{r,x}$ stands for the conditional expectation given $(R_0,X_{0})=(r,x)$ and $\pi:[0,\infty)\to [0,\infty)$ is a continuous profit function,  which we interpret as the compensation plan of the company's management. Its specific form will be further detailed below.

The investor's optimization criterion is also expressed in terms of discounted future profits, but these are computed differently. To take into account the investor's preference for green assets, we assume that their profit function depends on both the firm's production capacity and its environmental performance: the profits from a green firm are valued higher than the profits from a brown firm. 
This is in line with the literature on sustainable investment, where the utility of investors also depends both on their wealth and on the environmental performance of the firms they invest in \cite{pastor2021sustainable}.  We write the investor's optimization functional as follows:
\begin{equation}\label{eqvaluev}
\mathcal{J}^I_{r,x}(\eta,\nu) \triangleq \mathbb{E}_{r,x} \left[ \int_0^\infty \e^{- \rho t} \Pi(R^\eta_t,X_t^{\nu,\eta})\ud t- \alpha \int_{[0,\infty)} \e^{- \rho t} \ud \nu_t \right],
\end{equation}
where $\rho>0$ is the investor's discount rate, $\Pi:[0,\infty)^2\to [0,\infty)$ is a continuous function which is further specified below and $\alpha>0$ is a weighting factor, which quantifies the cost for the investor of increasing the firm's production capacity by one unit. The function $\Pi$ measures, in monetary units, the profits of the investor adjusted for the environmental performance of the firm. The parameter $\alpha$ can also quantify the investment frictions such as the fund manager's fees if the investment is managed through a fund.

For a fixed investment strategy $\nu\in\mathcal A_I$ selected by the investor, the firm aims to maximize its profits. Therefore, the firm's problem is formulated as
\begin{equation}
    \bar w(r,x;\nu) \triangleq \sup_{\eta\in\mathcal A_F} \mathcal{J}^F_{r,x}(\eta,\nu).
\end{equation}
Conversely, when the firm selects an emission reduction strategy $\eta\in\mathcal A_F$, the investor's objective is to maximize the inter-temporal optimization functional \eqref{eqvaluev}. Therefore the investor's problem reads
\begin{equation}\label{eq:equilv}
    \bar v(r,x;\eta) \triangleq \sup_{\nu\in\mathcal A_I} \mathcal{J}^I_{r,x}(\eta,\nu).
\end{equation}
We are interested in obtaining a pair of strategies $(\eta^{\ast}, \nu^{\ast})$, that constitutes a Nash equilibrium for the game according to the following definition:
\begin{definition}\label{defNash}
A pair $(\eta^{\ast}, \nu^{\ast})\in\mathcal A_F\times\mathcal A_I$ is a Nash equilibrium for the game starting at $(r,x)$ if
\begin{equation*}
\begin{aligned}
\mathcal{J}^F_{r,x}(\eta^{\ast},\nu^{\ast})\ge \mathcal{J}^F_{r,x}(\eta,\nu^{\ast})\quad\text{and}\quad
\mathcal{J}^I_{r,x}(\eta^{\ast},\nu^{\ast})\ge \mathcal{J}^I_{r,x}(\eta^{\ast},\nu),
\end{aligned}
\end{equation*}
for any other pair $(\eta,\nu)\in \mathcal A_F\times\mathcal A_I$. Then we say that $w(r,x)\triangleq\bar w(r,x;\nu^{\ast})$ and $v(r,x)\triangleq\bar v(r,x;\eta^{\ast})$ are the equilibrium payoffs (or values) for the firm and the investor, respectively.

We say that an equilibrium pair $(\eta^*,\nu^*)$ is unique\footnote{Equivalently we may also say that the equilibrium is unique.} if for any other equilibrium pair $(\nu,\eta)$ we have $\nu_t=\nu^*_t$ for all $t\ge 0$ and $\eta_t=\eta^*_t$ for a.e.\ $t\ge 0$.
\end{definition}

\begin{remark} To simplify the model, we assumed that the firm's manager only cares about the financial performance, while the investor takes both the financial and environmental performance into account. In practice, to align the interests of investors with those of company management, executive compensation packages increasingly involve ESG (environmental, social and governance) performance metrics. In a study based on a large international sample of companies \cite{cohen2023executive}, Cohen et al.~find that the proportion of companies using ESG metrics in executive compensation grew from 1\% to 38\% between 2011 and 2021. However, only a small fraction of firms use carbon emissions as performance criterion: in 2020 this was the case for 13\% of all firms using ESG metrics; the most widely used metrics were related to safety/security and employee satisfaction (both were used by 51\% of all firms using ESG metrics). We therefore consider that emission reduction goals of investors often remain misaligned with those of the management.

On the other hand, not all investors are green and firms can also be financed by brown capital, which does not require them to decarbonize their production processes. The proportion of green investors is difficult to estimate and is a controversial issue in sustainable finance literature: Berk and Van Binsbergen \cite{berk2025impact} estimate this share to 1\% while Cheng et al.~\cite{cheng2024impact} use the estimate of 20\% in their calibration exercise. In any case, it is clear that the presence of brown investors in the market hinders the impact of green investors on the cost of capital and the environmental performance of public companies  \cite{cheng2024impact}. The empirical evidence for investors' impact is strongest for young financially constrained companies in underdeveloped markets, which are usually privately owned \cite{heeb2020}. For this reason, we focus on a privately-owned company in this paper.
\end{remark}

\begin{remark}\label{rem:wellposed}
We notice that the admissible control classes for the investor and the firm are different. The firm's maximum investment rate is capped by $\eta_{\mathrm{max}}>0$ whereas the investor's one is uncapped. There are two reasons for this choice. From a modeling perspective, $(R^\eta_t)_{t\ge 0}$  may be interpreted as a continuous flow of spending towards emission reduction/compensation, for example, by purchasing bio-fuel instead of fossil-fuel, green energy certificates or carbon offsets.  Of course one could also interpret $(R^\eta_t)_{t\ge 0}$ as spending on large emission reduction projects, such as refurbishing a steel-making plant to use the electric arc furnace technology instead of a more carbon-intensive blast furnace. In the latter case, perhaps, discontinuous $(R^\eta_t)_{t\ge 0}$ would make more sense. However, from a mathematical perspective, it is not clear that our stochastic game is well-posed when both the firm and the investor are allowed to use singular controls, even under rather innocuous specifications of functions $\pi$ and $\Pi$. We are going to illustrate this issue with an example during the solution of the zero-noise limit of our game (cf.\ Remark \ref{rem:wellposed2}). 
\end{remark}

\begin{remark}
The definition of equilibrium is formulated on the product space $\mathcal A_F\times\mathcal A_I$ of pairs $(\eta,\nu)$. A note of caution is necessary here, because our game is dynamic and we expect equilibrium controls $(\eta^*,\nu^*)$ in feedback form (i.e., as functionals of the path of the controlled dynamics). It is not obvious that any choice of $(\eta,\nu)$ in feedback form would yield a well-posed dynamics of the pair $(X^{\nu,\eta},R^\eta)$ (i.e., a unique strong solution of the resulting SDE). This is a common feature in stochastic games in continuous time and we tacitly adopt the convention that a pair $(\eta,\nu)$ which does not yield a well-posed dynamics is associated with payoffs $\mathcal{J}^I_{r,x}(\eta,\nu)=\mathcal{J}^F_{r,x}(\eta,\nu)=-\infty$. Then, additionally to the conditions in Definition \ref{defNash}, a pair $(\eta,\nu)$ is an equilibrium if it also yields a well-posed dynamics of the system. For the purpose of this paper we do not need to dig deeper in this direction but we point the interested reader to \cite[Sec.\ 2]{possamai2020zero} for an extended discussion in the context of zero-sum games and to \cite[Sec.\ 3]{DeA23} for rigorous game-theoretic formulations of nonzero-sum stochastic games of singular control (which our setting is a special case of).
\end{remark}

{We proceed to illustrate our main results in the next subsection. We first state the form of an explicit equilibrium in the zero-noise limit and then we state a general verification theorem in the fully diffusive setup. The proofs of both statements will be presented in the subsequent sections.}

\subsection{{Main results, part 1: zero noise limit}}\label{sec:main1}
In the special case of a deterministic dynamics of the firm's profits (i.e., in the zero-noise limit $\sigma=0$), with decreasing profit stream (i.e., {$\mu\le 0$}) we are able to obtain the unique equilibrium in explicit form. The assumption $\mu\le 0$ is in line with standard economic modeling, where the profitability of a firm, in the absence of investment, must decrease over time (e.g., due to ageing of manufacturing machines, etc.). Instead the assumption $\sigma=0$ is harder to justify from an economic perspective but we can think of equilibria in this setting as a zero-order approximation of equilibria in stochastic setups with small noise. Indeed, we will observe numerically in the subsequent sections that equilibria in the stochastic framework maintain the same qualitative structure as the one we obtain here. 

The benchmark example that we have in mind is when firm's profit function is linear (i.e., the firm is risk-neutral) and the investor adopts a Cobb-Douglas profit function. In that case we have 
\begin{equation}\label{fucfu}
\pi(x) = x \quad \mbox{and} \quad \Pi(r,x) = x^\beta r^\gamma,\quad \text{for $(r,x)\in[0,\infty)\times\mathbb{R}_+$},
\end{equation}
with $\beta,\gamma \in (0,1)$ and $\gamma+\beta\geq 1$. The latter condition ensures that item (iii) in the assumption below holds, because in this case the function $a(r)$ specified therein is proportional to $r^{\frac{\gamma}{1-\beta}}$.
However, in the interest of mathematical generality we solve the problem under the following assumption, which is enforced throughout the section.
\begin{assumption}\label{ass:profit}
We have $\sigma=0$, $\mu\le 0$ and the following properties hold:
\begin{itemize}
\item[(i)] The function $\Pi$ is non-decreasing in both variables; $\Pi(r,\cdot)$ is strictly concave, $\Pi_x$ and $\Pi_{xx}$ exist and are continuous on $(0,\infty)^2$,
with derivative $\Pi_x(r,\cdot)$ satisfying the Inada conditions
$$
\lim_{x\to 0}\Pi_x(r,x) = +\infty,\qquad \lim_{x\to +\infty}\Pi_x(r,x)=0. 
$$
Therefore $\Pi_x(r,\cdot)$ admits an inverse $G(r,z)\triangleq(\Pi_x(r,\cdot))^{-1}(z)$. For $\delta\triangleq \rho-\mu\ge 0$ we assume that the mapping $r\mapsto a(r)\triangleq G(r,\alpha\delta)$ is strictly increasing and continuously differentiable on $(0,\infty)$, with 
\[
\int_0^\infty \e^{-\rho t}a(\eta_{\mathrm{max}}t)\ud t<\infty.
\]

\item[(ii)] The function $\pi$ is continuously differentiable and strictly increasing on $(0,\infty)$; the derivative $\dot \pi$ is non-decreasing on $(0,\infty)$; we extend $\pi$ to $(-\infty,0]$ as $\pi(x)=\pi([x]^+)$; finally,
\[
\int_0^\infty \e^{-\bar \rho t} \pi(a(\eta_{\max} t))\ud t<\infty.
\]

\item[(iii)] When $\mu<0$, the mapping $r\mapsto \dot\pi\big(a(r)\big) \dot a(r)$ is non-decreasing.
\end{itemize}
\end{assumption}

\begin{remark}
Assumption \ref{ass:profit}-(ii) implies, in particular, that the company's profit function $\pi$ is convex. This aligns with the interpretation of $\pi$ as the compensation package of the company's management, which may include convex features such as stock options. Research on CEO compensation packages \cite{ross2004compensation,hemmer1999introducing} suggests that convex compensation schemes may mitigate excessive risk-avoiding behavior of firm's managers. Furthermore, the authors of 
\cite{dittmann2010sticks} find that actual CEO contracts tend to be convex and develop a principal-agent model demonstrating that convex compensation schemes are optimal for medium and high performance outcomes. 
\end{remark}

For the statement of the next theorem we introduce some basic notation. For $\eta\in\mathcal A_F$ we let $\tau(\eta)=\inf\{t\ge 0: X^{0,\eta}_t\le a(R^\eta_t)\}$ and $\nu^a\in\mathcal A_I$ be defined as $\nu^a_t=0$ for $t\in[0,\tau(\eta)]$ and 
\begin{equation}\label{eq:nua0}
\nu^a_t\triangleq a(R^\eta_t)+R^\eta_t-a(R^\eta_{\tau(\eta)})+R^\eta_{\tau(\eta)}-\int_{\tau(\eta)}^t \mu a(R^\eta_s) ds ,\quad \text{for $t\ge \tau(\eta)$},
\end{equation}
where $a(r)=G(r,\alpha\delta)$ is from (i) in Assumption \ref{ass:profit}. Notice that $\nu^a$ also depends on the control $\eta$, via $\tau(\eta)$ and the dynamics of $R^\eta$. In the statement of the theorem below, $\nu^a$ is intended as in \eqref{eq:nua0} with $\tau(\eta)$ and $R^\eta$ arising from the control chosen by the firm in equilibrium. The required integrability of the process $\nu^a$ will be verified in the proof of the theorem. 

\begin{theorem}[{\bf Equilibrium in the zero noise limit}]\label{thm:main0} Let Assumption \ref{ass:profit} hold.
\begin{itemize}
\item[(i)] If $\mu=0$, the unique equilibrium pair is given by
\begin{equation}
(\eta^*,\nu^*)=\left\{
\begin{array}{ll}
(0,0),& \text{for $x>b(r)$},\\
(\eta_{\mathrm{max}},\nu^a), & \text{for $a(r)\le x\le b(r)$},
\end{array}
\right.
\end{equation}
where $b\in C([0,\infty))\cap C^1(\mathbb{R}_+)$ can be uniquely determined (cf.\ Lemma \ref{bdef.lm}), it is increasing and such that $b(r)>a(r)$ for all $r\in[0,\infty)$. 

\item[(ii)] If $\mu<0$, the unique equilibrium pair is given by 
\[
(\eta^*,\nu^*)=(\eta_{\mathrm{max}}1_{\{t\ge \tau_b\}},\nu^a)
\]
with $\tau_b=\inf\{t\ge 0:x\e^{\mu t}\le b(r)\}$, where $b\in C([0,\infty))\cap C^1(\mathbb{R}_+)$ can be uniquely determined (cf.\ Proposition \ref{Proposition_funcb}), it is increasing and such that $b(r)>a(r)$ for all $r\in[0,\infty)$.  
\end{itemize}
\end{theorem}
It is possible to determine explicit expressions for the equilibrium payoffs of both players and this will be done in the detailed analysis that leads to the proof of the theorem. The main difference between (i) and (ii) in the theorem above is that when $\mu=0$, if $X_0>b(R_0)$ there is going to be no dynamics and the process $(R_t,X_t)$ remains constant in equilibrium, i.e., no emission abatement and no investment. Instead, when $\mu<0$, there is always going to be a dynamics for the production capacity and, in equilibrium, both investment and emission abatement are performed.
The equilibrium strategies of Theorem \ref{thm:main0}-(ii) are illustrated in Figure \ref{illustration.fig} in the situations when the initial values $(X_0,R_0)$ are in the no action region for both the company and the investor, that is $X_0>b(R_0)$. The starting point is the point A in the left graph, corresponding to $t=0$. In the first period, the company and the investor do nothing, and the production capacity of the firm decreases exponentially until it hits the value $b(R_0)$ (point B in the left graph). At this time ($\tau_b$), the firm starts to invest in pollution abatement, but the investor takes no action. As a result, during the second period, the production capacity declines faster than in the first period, and the abatement investment grows at the maximum rate. The second period continues until the production capacity $X_t$ hits the moving boundary $a(R_t)$ (point C in the left graph). At this time ($\tau_M$) the investor starts to invest to keep the production capacity at the moving boundary.

\begin{figure}[htbp]
\centerline{\includegraphics[width=.5\textwidth]{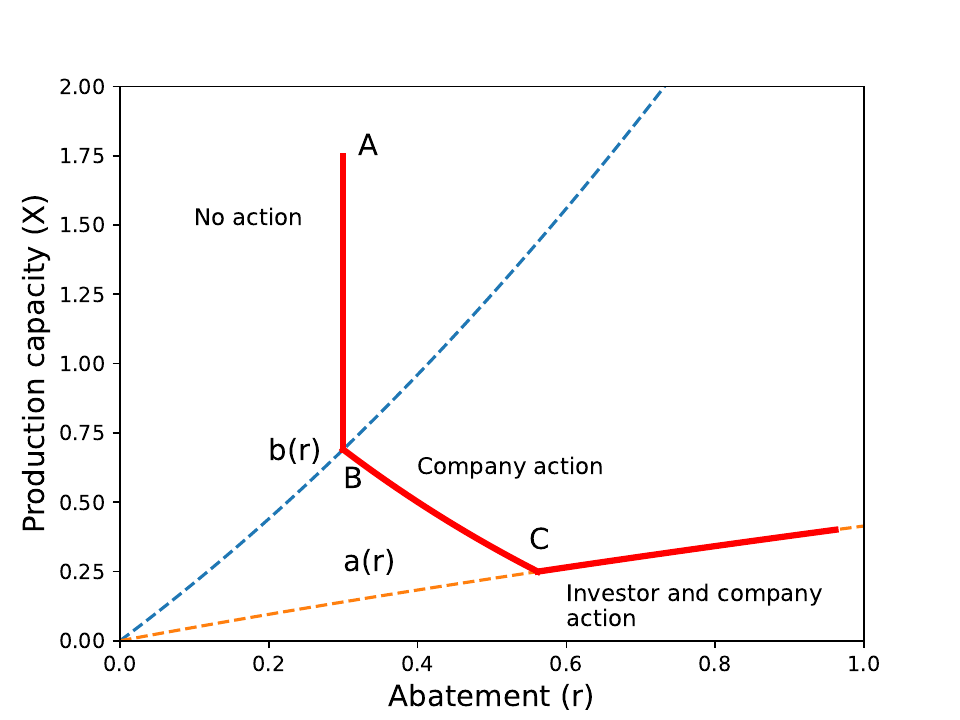}\includegraphics[width=.5\textwidth]{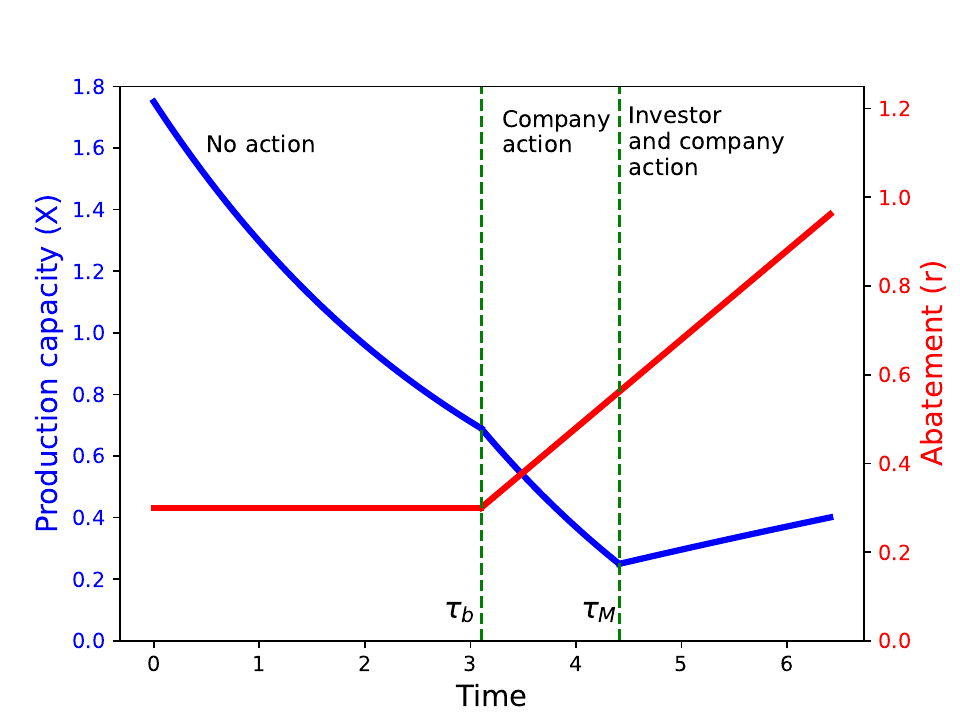}}
\caption{Illustration of the optimal strategies in the deterministic case with $\mu<0$. Left: $X$ as function of $R$. Right: $X$ (in blue) and $R$ (in red) as function of time. See text for a detailed discussion.}
\label{illustration.fig}
\end{figure}

\subsection{{Main results,  part 2: verification theorem}}\label{sec:main2}

By standard considerations based on dynamic programming, we expect that a pair of equilibrium payoffs should be obtained by solving a system of variational inequalities. First, we state the variational problem and then we formally connect it with the game via a so-called {\em verification theorem} (cf.\ Theorem \ref{thm:verif}). In what follows, given a function $\varphi$ we denote by $\varphi_x$, $\varphi_{xx}$ and $\varphi_r$ its first and second order derivatives with respect to $x$ and its first order derivative with respect to $r$, respectively. Sometimes we equivalently use $\partial_x \varphi$, $\partial_{xx} \varphi$ and $\partial_r \varphi$.

The infinitesimal generator of the uncontrolled process $X^0$ is defined via its action on sufficiently smooth functions $\varphi$ as
\[
(\mathcal{L} \varphi)(r,x) \triangleq \frac{\sigma^2 x^2}{2} \varphi_{xx}(r,x) + \mu x \varphi_{x}(r,x).
\]
The Hamiltonian of the firm's problem is defined as 
\begin{equation}\label{eq:Ham}
\mathcal{H}(r,x;\varphi) \triangleq \sup_{0 \leq \eta \leq \eta_{\mathrm{max}}}  \big(\varphi_r(r,x) - \varphi_x(r,x)\big)\eta ,
\end{equation}
for smooth $\varphi$. 
Then, we expect a pair $(w,v)$ of equilibrium payoffs (as in Definition \ref{defNash}) to be solution, in a sense to be specified later, of the following system: let
\[
\mathcal M\triangleq\{(r,x): v_x(r,x) = \alpha\}, \quad {\mathcal I}\triangleq\mathcal{M}^c=\big([0,\infty)\times\mathbb{R}_+\big)\setminus\mathcal M\ \text{ and }\ \partial\mathcal{M}=\overline{\mathcal M}\cap\overline{{\mathcal I}};
\] 
the function $w$ solves 
\begin{equation}\label{eqw1}
\left\{
\begin{array}{ll}
(\mathcal{L} w - \bar\rho w) (r,x) + \mathcal{H}(r,x;w) + \pi(x) =0, &\qquad (r,x)\in\mathcal I,\\[+4pt]
w_x(r,x) = 0, &\qquad (r,x)\in\partial\mathcal M; 
\end{array}
\right.
\end{equation}
letting $\eta^*(r,x)\triangleq \eta_{\mathrm{max}} 1_{\{w_r>w_x\}}(r,x)$, the function $v$ solves
\begin{equation}\label{eqvariationalv1}
\max \left\{ (\mathcal{L} v- \rho v) (r,x) + \left( v_{r}(r,x) - v_{x}(r,x) \right) \eta^*(r,x)  + \Pi(r,x), v_x(r,x) - \alpha \right\} = 0,
\end{equation}
for $(r,x)\in[0,\infty)\times\mathbb R_+$. 

For now we may assume that all solutions are understood in the classical sense (i.e., with continuous derivatives). Later, in Section \ref{sec:strong}, we will adopt a notion of {\em strong} solution. Suitable growth conditions should also be specified. These will be encoded in the so-called transversality conditions of our verification theorem.
It is immediate to check that 
\[
\eta^*(r,x)=\mathrm{argmax}_{0\le \eta\le \eta_{\mathrm{max}}}\big\{\big(w_r(r,x) - w_x(x,r)\big)\eta\big\}.
\]

In the next theorem, we show that if a sufficiently smooth solution pair $(v,w)$ of the above problem exists, then indeed it corresponds to a pair of equilibrium payoffs for the game.
It is convenient to also define the following compact notation: for $q\ge 0$
\begin{equation}\label{eqG}
\mathcal G^q[\varphi,\eta](r,x)\triangleq \left(\mathcal{L} \varphi- q \varphi + ( \varphi_{r} - \varphi_{x}) \eta\right)(r,x),
\end{equation}
for any pair of sufficiently smooth functions $(\varphi,\eta)$. In the diffusive case $\sigma>0$ the verification theorem is stated under classical regularity assumptions, which are standard in one-dimensional singular control settings and are consistent with the smooth-fit behavior observed in our numerical solutions.

We emphasize that the next theorem {\em does not} require Assumption \ref{ass:profit}. 

\begin{theorem}[\bf Verification] \label{thm:verif}
Assume there is a pair $(v,w)$ of non-negative, continuous functions on $[0,\infty)^2$ that solves \eqref{eqw1}--\eqref{eqvariationalv1} with $w\in C^{1,2}(\overline{\mathcal I})$ and  $v\in C^{1,2}((0,\infty)^2)$. Assume further that there is a control $\nu^*\in\mathcal A_I$ such that for any $\eta\in\mathcal A_F$, the pair $(R^\eta_t,X^{\nu^*,\eta}_t)_{t\ge 0}$ is well-posed and such that $\mathbb P_{r,x}$-a.s.
\begin{equation}\label{eq:SC}
\begin{aligned}
&(R^\eta_t,X^{\nu^*,\eta}_t)\in\overline{\mathcal I},\quad \text{for $t\ge 0$},\\
&1_{\{(r,x)\in\mathcal I\}}\, \nu^*_0+\int_{(0,T]} 1_{\{(R^\eta_{t},X^{\nu^*,\eta}_{t-})\in\mathcal I\}}\ud \nu^*_t=0,\quad \text{for any $T>0$},\\
&(R^\eta_t,X^{\nu^*,\eta}_{t-}+\Delta\nu^*_t)\in\partial\mathcal{M},\quad \text{for any $t> 0$}.
\end{aligned}
\end{equation}
Finally, assume that the process $(R^{\eta^*}_t,X^{\nu^*,\eta^*}_t)_{t\ge 0}$ with $\eta^*_t=\eta^*(R^{\eta^*}_t,X^{\nu^*,\eta^*}_t)$ is well-defined and the transversality conditions
\begin{equation}\label{eq:transv.0}
\lim_{n\to\infty}\mathbb E_{r,x}\Big[\e^{-\bar\rho \theta^*_n}w\big(R^{\eta^*}_{\theta^*_n},X^{\nu^*,\eta^*}_{\theta^*_n}\big)\Big]=0\quad\text{and}\quad
\lim_{n\to\infty}\mathbb E_{r,x}\Big[\e^{-\rho \theta^*_n}v\big(R^{\eta^*}_{\theta^*_n},X^{\nu^*,\eta^*}_{\theta^*_n}\big)\Big]=0,
\end{equation}
hold with $\theta^*_n\triangleq\inf\{t\ge 0:(R^{\eta^*}_t,X^{\nu^*,\eta^*}_t)\notin[0,n)^2\}$.
Then $(v,w)$ are a pair of equilibrium payoffs as in Definition \ref{defNash} and the pair $(\eta^*_t,\nu^*_t)_{t\ge 0}$ is a pair of optimal strategies.
\end{theorem}

Various ramifications of the above theorem are discussed in Section \ref{sec:verif}, including versions of the theorem under more relaxed regularity requirements for the solution pair $(v,w)$. It is worth observing that the verification theorem, in one of its relaxations, is indeed used to prove rigorously Theorem \ref{thm:main0} in the zero noise limit. 

Figure \ref{illustration2.fig} shows an optimal trajectory obtained by numerical solution of the system (\ref{eqw1}--\ref{eqvariationalv1}) using the algorithm described in Section \ref{sec:equilibrium} with parameter values given in Section \ref{simul:sthoc}. Similarly to the deterministic setting, the solution is characterized through abatement boundary $b$ (higher boundary in the top graph) and investment boundary $a$ (lower boundary in the top graph). The investor (whose optimal control $\nu^*$ is shown with thick blue curve in the bottom graph) acts to keep the production capacity of the firm above the investment boundary. Notice that the blue curve is flat when the production capacity curve in the top graph is above the boundary and only grows when the production capacity touches the boundary. The company (whose abatement investment $R$ is shown with thin red curve in the bottom graph) acts with maximum force when the production capacity curve is between the two boundaries, and the red curve is flat when the production capacity is above the abatement boundary. Unlike the deterministic setting of Figure \ref{illustration.fig}, where the production capacity process remains at the investment boundary after touching it once, and the abatement never stops once it has started, in the stochastic setting, random fluctuations of the production capacity process may push it away from the investment boundary and even above the abatement boundary, halting both investment and abatement, which restart at a later date when the respective boundaries are attained.

\begin{figure}[htbp]
\centerline{\includegraphics[width=.6\textwidth]{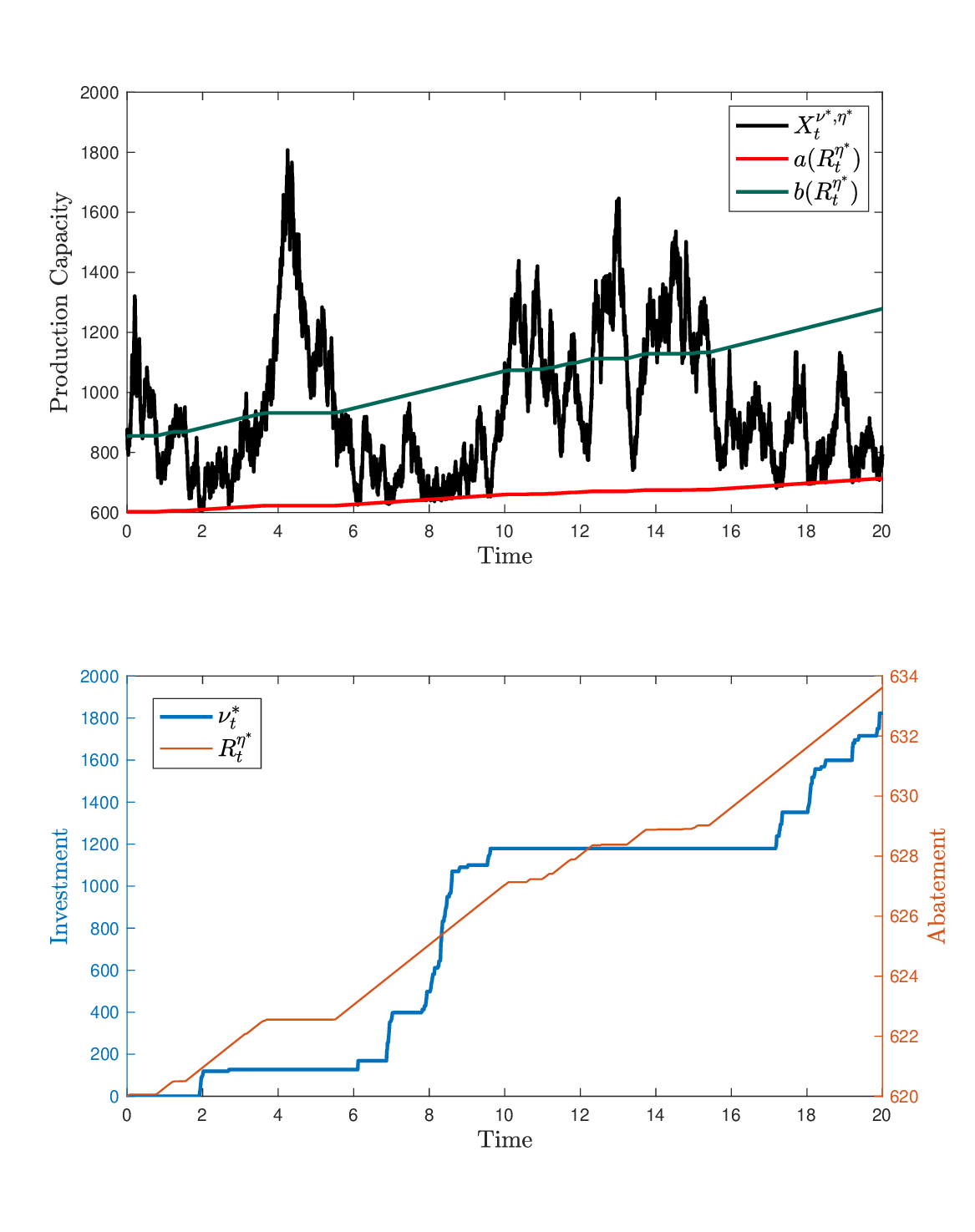}}
\caption{Illustration of the optimal strategies in the stochastic case. Top graph: production capacity ($X$) and the two boundaries as function of time. Bottom graph: optimal investment $\nu^*$ (thick blue line with left scale) and abatement $R$ (thin red line, right scale) as function of time.}
\label{illustration2.fig}
\end{figure}

\section{Proof of verification theorem}\label{sec:verif}

We begin our study by offering the relatively easy proof of the verification theorem (Theorem \ref{thm:verif}).
\begin{proof}[{\bf Proof of Theorem \ref{thm:verif}}]
The proof follows familiar arguments from stochastic control theory and so we only outline it here.
For $(r,x)\in\overline{\mathcal{I}}$ and any $\eta\in\mathcal A_F$, setting $\theta_n\triangleq\inf\{t\ge 0:(R^{\eta}_t,X^{\nu^*,\eta}_t)\notin[0,n)^2\}$, by an application of It\^o's formula we obtain
\begin{equation*}
\begin{aligned}
&\e^{-\bar\rho (t\wedge\theta_n)}w\big(R^\eta_{t\wedge\theta_n},X^{\nu^*,\eta}_{t\wedge\theta_n}\big) \\
&=w(r,x+\nu^*_0)+\int_0^{t\wedge\theta_n}\e^{-\bar\rho s}\mathcal G^{\bar \rho}[w,\eta_s]\big(R^\eta_s,X^{\nu^*,\eta}_s\big)\ud s\\
&\quad+\int_0^{t\wedge\theta_n}\e^{-\bar\rho s}w_x\big(R^\eta_s,X^{\nu^*,\eta}_s\big)\ud \nu^{*,c}_s+\sum_{0<s\le t\wedge\theta_n}\e^{-\bar\rho s}\big[w\big(R^\eta_s,X^{\nu^*,\eta}_s\big)-w\big(R^\eta_s,X^{\nu^*,\eta}_{s-}\big)\big]\\
&\quad+\int_0^{t\wedge\theta_n}\e^{-\bar\rho s}w_x\big(R^\eta_s,X^{\nu^*,\eta}_s\big)\sigma X^{\nu^*,\eta}_s\ud B_s\\
&\le w(r,x)+\int_0^{t\wedge\theta_n}\e^{-\bar\rho s}\big(\mathcal L w-\bar\rho w+\mathcal H(\cdot,w)\big)\big(R^\eta_s,X^{\nu^*,\eta}_s\big)\ud s\\
&\quad+\int_0^{t\wedge\theta_n}\e^{-\bar\rho s}w_x\big(R^\eta_s,X^{\nu^*,\eta}_s\big)\sigma X^{\nu^*,\eta}_s\ud B_s\\
&= w(r,x)-\int_0^{t\wedge\theta_n}\e^{-\bar\rho s}\pi\big(X^{\nu^*,\eta}_s\big)\ud s+\int_0^t\e^{-\bar\rho s}w_x\big(R^\eta_s,X^{\nu^*,\eta}_s\big)\sigma X^{\nu^*,\eta}_s\ud B_s,
\end{aligned}
\end{equation*}
where the inequality is by definition of the Hamiltonian \eqref{eq:Ham} and we use $w_x(R^\eta_s,X^{\nu^*,\eta}_s)\ud \nu^{*,c}_s=0$, $w\big(R^\eta_s,X^{\nu^*,\eta}_s\big)-w\big(R^\eta_s,X^{\nu^*,\eta}_{s-}\big)=0$ and $w(r,x+\nu^*_0)-w(r,x)=0$ by \eqref{eq:SC} and the second equation in \eqref{eqw1}. Now, taking expectations and rearranging terms we obtain
\begin{equation*}
\begin{aligned}
w(r,x)&\ge \mathbb E_{r,x}\Big[\int_0^{t\wedge\theta_n}\e^{-\bar\rho s}\pi\big(X^{\nu^*,\eta}_s\big)\ud s+\e^{-\bar\rho (t\wedge\theta_n)}w\big(R^\eta_{t\wedge\theta_n},X^{\nu^*,\eta}_{t\wedge\theta_n}\big)\Big]\\
&\ge \mathbb E_{r,x}\Big[\int_0^{t\wedge\theta_n}\e^{-\bar\rho s}\pi\big(X^{\nu^*,\eta}_s\big)\ud s\Big],
\end{aligned}
\end{equation*}
where the inequality holds because $w\ge 0$. Since also $\pi\ge 0$, letting $t\to\infty$ and $n\to\infty$ and using Monotone Convergence yields 
\[
w(r,x)\ge \mathbb E_{r,x}\Big[\int_0^\infty\e^{-\bar\rho s}\pi\big(X^{\nu^*,\eta}_s\big)\ud s\Big]=\mathcal J^F_{r,x}(\nu^*,\eta).
\]
Repeating the arguments above with $\eta=\eta^*$ yields 
\begin{equation*}
\begin{aligned}
w(r,x)&= \mathbb E_{r,x}\Big[\int_0^{t\wedge\theta^*_n}\e^{-\bar\rho s}\pi\big(X^{\nu^*,\eta^*}_s\big)\ud s+\e^{-\bar\rho (t\wedge\theta^*_n)}w\big(R^{\eta^*}_{t\wedge\theta^*_n},X^{\nu^*,\eta^*}_{t\wedge\theta^*_n}\big)\Big].
\end{aligned}
\end{equation*}
First we let $t\to\infty$ and use dominated convergence. Then we let $n\to\infty$ and use Monotone Convergence and the transversality condition to obtain  $w(r,x)=\mathcal J^F_{r,x}(\nu^*,\eta^*)$. This shows that $\eta^*$ is a best response to $\nu^*$.

Now, let us look at the investor's payoff. Take $(r,x)\in[0,\infty)\times\mathbb R_+$ and arbitrary $\nu\in\mathcal A_I$. Assume that the dynamics $(R^{\eta^*}_t,X^{\nu,\eta^*}_t)$ is well-posed and redefine $\theta_n\triangleq\inf\{t\ge 0:(R^{\eta^*}_t,X^{\nu,\eta^*}_t)\notin[0,n)^2\}$ with a slight abuse of notation. By It\^o's formula we have 
\begin{equation*}
\begin{aligned}
&\e^{-\rho (t\wedge\theta_n)}v\big(R^{\eta^*}_{t\wedge\theta_n},X^{\nu,\eta^*}_{t\wedge\theta_n}\big)\\
&=v(r,x+\nu_0)+\int_0^{t\wedge\theta_n}\e^{-\rho s}\mathcal G^\rho[v,\eta^*_s]\big(R^{\eta^*}_s,X^{\nu,\eta^*}_s\big)\ud s+\int_0^{t\wedge\theta_n}\e^{-\rho s}v_x\big(R^{\eta^*}_s,X^{\nu,\eta^*}_s\big)\ud \nu^{c}_s\\
&\quad+\sum_{0< s\le t\wedge\theta_n}\e^{-\rho s}\big[v\big(R^{\eta^*}_s,X^{\nu,\eta^*}_s\big)-v\big(R^{\eta^*}_s,X^{\nu,\eta^*}_{s-}\big)\big]+\int_0^{t\wedge\theta_n}\e^{-\rho s}v_x\big(R^{\eta^*}_s,X^{\nu,\eta^*}_s\big)\sigma X^{\nu,\eta^*}_s\ud B_s\\
&\le v(r,x)-\int_0^{t\wedge\theta_n}\e^{-\rho s} \Pi\big(R^{\eta^*}_s,X^{\nu,\eta^*}_s\big)\ud s+\int_{[0,t\wedge\theta_n]}\e^{-\rho s}\alpha\ud \nu_s\\
&\quad+\int_0^{t\wedge\theta_n}\e^{-\rho s}v_x\big(R^{\eta^*}_s,X^{\nu,\eta^*}_s\big)\sigma X^{\nu,\eta^*}_s\ud B_s,
\end{aligned}
\end{equation*}
where the inequality holds by \eqref{eqvariationalv1}. Taking expectations, rearranging terms and passing to the limit in $t$ and $n$ yields $v(r,x)\ge \mathcal J^I_{r,x}(\eta^*,\nu)$, upon using also that $v\ge 0$.  Repeating the argument with $\nu=\nu^*$ the inequality becomes equality and we obtain 
\begin{equation*}
\begin{aligned}
v(r,x)&= \mathbb E_{r,x}\Big[\int_0^{t\wedge\theta^*_n}\e^{-\rho s}\Pi\big(R^{\eta^*}_s,X^{\nu^*,\eta^*}_s\big)\ud s-\int_{[0,t\wedge\theta_n]}\e^{-\rho s}\alpha\ud \nu^{*}_s+\e^{-\rho (t\wedge\theta_n)}v\big(R^{\eta^*}_{t\wedge\theta^*_n},X^{\nu^*,\eta^*}_{t\wedge\theta^*_n}\big)\Big].
\end{aligned}
\end{equation*}
Letting $t\to\infty$ and $n\to\infty$ and using dominated and monotone convergence along with the transversality condition we arrive at $v(r,x)=\mathcal J^I_{r,x}(\eta^*,\nu^*)$. This shows that $\nu^*$ is the best response to $\eta^*$ and it concludes the proof.
\end{proof}

There is a straightforward corollary of the verification theorem which can be interpreted as a verification theorem only for the firm. This assumes that the optimal strategy for the investor is known and it is essentially characterized by a set $\mathcal{M}$. However, the corollary does not assume that the investor's equilibrium payoff be a smooth solution of the variational problem \eqref{eqvariationalv1} nor that the set $\mathcal{M}$ be of the form specified before \eqref{eqw1}. The proof is omitted because it is a repetition, line by line, of the first part of the proof of Theorem \ref{thm:verif}.
\begin{corollary}\label{cor:verif}
Assume that there is a set $\mathcal M\subset[0,\infty)\times\mathbb{R}_+$ whose complement we denote by $\mathcal{I}\triangleq\mathcal{M}^c$ and a continuous function $w$ such that the following conditions hold:
\begin{itemize}
\item[(i)] The function $w \in C^{1,2}(\overline{\mathcal I})$ solves \eqref{eqw1};
\item[(ii)] For any $\eta\in\mathcal{A}_F$  there is $\nu^*=\nu^*(\eta)\in\mathcal{A}_I$ which is optimal for the investor in the sense that $\bar v(r,x;\eta)=\mathcal{J}_{r,x}^I(\eta,\nu^*(\eta))$ for all $(r,x)\in[0,\infty)\times\mathbb{R}_+$;
\item[(iii)] For any $\eta\in\mathcal{A}_F$, the dynamics $(R^\eta_t,X^{\nu^*,\eta}_t)_{t\ge 0}$  with $\nu^*(\eta)$ as in (ii) is well-posed and it satisfies \eqref{eq:SC}; 
\item[(iv)] The process $(R^{\eta^*}_t,X^{\nu^*,\eta^*}_t)_{t\ge 0}$ with $\eta^*_t=\eta^*(R^{\eta^*}_t,X^{\nu^*,\eta^*}_t)$ is well-posed and the transversality condition for the function $w$ in \eqref{eq:transv.0} holds.
\end{itemize}
Then, the pair $(\eta^*_t,\nu^*_t)_{t\ge 0}$ is a Nash equilibrium for the game with equilibrium payoffs $v$ and $w$, where $v(r,x)\triangleq \bar v(r,x;\eta^*)$ as in \eqref{eq:equilv}.
\end{corollary}

\begin{remark}[{\bf Verification with deterministic dynamics}]\label{rem:verif}
In the special case of no diffusive component, i.e., $\sigma=0$, the problem becomes deterministic. In this case the infinitesimal generator simplifies to $(\mathcal L\varphi)(r,x)=\mu x\varphi_x(r,x)$. The verification theorem above continues to hold under the less restrictive assumptions that $v$ and $w$ be continuous on $[0,\infty)^2$ with $w\in C^1(\overline{\mathcal I})$ and $v\in C^1((0,\infty)^2)$ (for Corollary \ref{cor:verif} only $w\in C^1(\overline{\mathcal I})$ is needed). 
\end{remark}
\begin{remark}
It seems unlikely that one could construct smooth solutions to the system of variational inequalities \eqref{eqw1}--\eqref{eqvariationalv1}, either with $\sigma>0$ or $\sigma=0$. Therefore weaker notions of solution may be needed. In particular, we will see in the next section (especially in subsection \ref{sec:strong}) that a notion of {\em strong} solution is needed for the investor's payoff when $\sigma=0$.
\end{remark}
It is clear that Corollary \ref{cor:verif} imposes very restrictive assumptions, but we are able to use it in order to obtain an equilibrium in the deterministic game in the next section.

\section{Equilibrium in the zero-noise limit}\label{sec:deterministic}

In this section we address the construction of the equilibrium for the problem in the zero noise limit. In particular, at the end of our analysis we will have proven Theorem \ref{thm:main0} and we will have collected a number of auxiliary results concerning the boundary function $b$ and the equilibrium payoffs $v$ and $w$ for both players.

Setting $(R_0,X_0)=(r,x)$, we consider the following controlled dynamics: 
\begin{equation}\label{eqdet}
\left\{
\begin{array}{ll}
X_t^{\nu,\eta} =X_0+\int_0^t \mu X_s^{\nu,\eta} \ud s + \nu_t -\int_0^t\eta_s \ud s, \\[+5pt]
R_t^{\eta} = R_0+\int_0^t \eta_s\ud s,
\end{array}
\right.
\end{equation}
for $t\in[0,\infty)$. Then, the explicit dynamics of $X^{\nu,\eta}$ reads (cf.\ \eqref{eq:Xexpl})
\begin{equation}
X_t^{\nu,\eta} = Y_t \left( x + \widehat{\nu}_t - \widehat{\Lambda}_t  \right),
\end{equation}
where
\begin{gather}
 Y_t \triangleq \e^{\mu t},\qquad
 \widehat{\nu}_t \triangleq \int_{[0,t]} Y_s^{-1} \ud \nu_s,\qquad
 \widehat{\Lambda}_t \triangleq \int_0^t Y_s^{-1} \eta_s \ud s,
\end{gather}
for all $t \in [0,\infty)$. In this context the filtration $\mathbb F$ is trivial, i.e., all processes are deterministic, including all admissible controls for the firm and the investor.

\begin{theorem}\label{thm:nustar}
Let the controlled dynamics be as in \eqref{eqdet} with $\mu \leq 0$ and let the profit function of the investor $\Pi(r,x)$ satisfy Assumption \ref{ass:profit}. Then, for any $\eta\in\mathcal A_F$ the investor's optimal investment policy is unique\footnote{Because of right-continuity of processes in the class $\mathcal{A}_I$, uniqueness is up to indistinguishability. Here, because the problem is deterministic, it means that for any other optimal control $\nu\in\mathcal{A}_I$ it holds $\nu_t=\nu^*_t$, for all $t\ge 0$.} and it is given by 
\begin{equation}\label{eq:nua}
\nu^*_t\triangleq\int_{[0,t]}Y_s\ud \widehat\nu^*_s,\qquad\text{for $t\ge 0$},
\end{equation}
where
\begin{equation}\label{eq:nua2}
\widehat\nu^*_t\triangleq  \big[Y^{-1}_t a(R^\eta_t)-x+  \widehat \Lambda^\eta_t \big]^+\quad\text{with}\quad a(r)= G(r,\alpha\delta). 
\end{equation}

Moreover, the pair $(R^\eta_t,X^{\nu^*,\eta}_t)_{t\ge 0}$ satisfies \eqref{eq:SC} with $\mathcal I\triangleq\{(r,x)\in[0,\infty)^2: x>a(r)\}$ and $\mathcal{M}\triangleq\{(r,x):0<x\le a(r)\}$.
\end{theorem}
\begin{proof}
For a fixed $\eta\in\mathcal A_F$ the dynamics of $R^\eta$ is determined and so is the function $\widehat\Lambda=\widehat \Lambda^\eta$. For the ease of notation, using that $\eta$ is fixed throughout, we simply denote $(R^\eta,\widehat\Lambda^\eta)=(R,\widehat \Lambda)$. The investor's problem is to maximize over $\nu \in \mathcal A_I$ (equivalently over $\widehat \nu$) the quantity:
\begin{equation*}
\mathcal{J}^I_{r,x}(\eta,\nu) 
= \int_0^\infty \e^{-\rho t} \Pi \big(R_t,Y_t[x+\widehat \nu_t-\widehat\Lambda_t]\big) \ud t- \alpha \int_{[0,\infty)} \e^{-\rho t} Y_t\ud \widehat \nu_t. 
\end{equation*}
The integrability condition for $\nu \in \mathcal A_I$ implies  
$\lim_{t\to\infty}\e^{-\rho t}Y_t\widehat{\nu}_t = 0$, because it is easy to verify by integration by parts that
\[
0\le \int^\infty_0\e^{-\rho s}Y_s\widehat \nu_s\ud s\le \rho^{-1}\widehat \nu_0+\rho^{-1}\int_{[0,\infty)}\e^{-\rho s}Y_s\ud \widehat \nu_s=\rho^{-1}\nu_0+\rho^{-1}\int_{[0,\infty)}\e^{-\rho s}\ud \nu_s<\infty.
\]
Integration by parts then yields (recalling that $\delta= \rho-\mu$):
\begin{equation}\label{eq:barJ}
\begin{aligned}
\mathcal{J}^I_{r,x}(\eta,\nu)&= \int_0^\infty \e^{-\rho t}\Big( \Pi \big(R_t,Y_t[x+\widehat \nu_t-\widehat\Lambda_t]\big)-\alpha\delta Y_t\widehat \nu_t \Big)\ud t-\alpha \lim_{t\to\infty}\e^{-\rho t}Y_t\widehat{\nu}_t  \\
&= \int_0^\infty \e^{-\rho t} \Big( \Pi \big(R_t,Y_t[x+\widehat \nu_t-\widehat\Lambda_t]\big)-\alpha\delta Y_t\widehat \nu_t \Big)\ud t\triangleq \bar{\mathcal{J}}^I_{r,x}(\eta,\widehat \nu).
\end{aligned}
\end{equation}

For every $t\geq0$, the functional 
$$
\widehat\nu \mapsto \Pi \big(R_t,Y_t[x+\widehat \nu-\widehat\Lambda_t]\big)-\alpha\delta Y_t\widehat \nu,
$$
is strictly concave and in view of Assumption \ref{ass:profit}-(i) it admits a unique (pointwise) maximizer $\widehat{\nu}^*$ on $[0,\infty)$. If such maximizer turns out to be an admissible control
then $\sup_{\nu\in\mathcal{A}_I}\mathcal{J}_{r,x}^I(\eta,\nu)=\bar{\mathcal{J}}_{r,x}^I(\eta,\widehat{\nu}^*)$ and $\widehat{\nu}^*$ must be the unique optimal investment policy for the original problem.

By first order conditions we look for $\widehat\nu$ such that 
\begin{equation}
  \Pi_x \big(R_t,Y_t[x+\widehat \nu_t-\widehat\Lambda_t]\big)Y_t - \alpha \delta Y_t = 0,\quad \text{for all $t\ge 0$}.
\end{equation}
Since $\widehat \nu_t$ must be also positive, and recalling the inverse function $G(r,z)=(\Pi_x(r,\cdot))^{-1}(z)$, we have
\begin{equation*}
\widehat \nu^*_t=\big[G(R_t,\alpha\delta)Y^{-1}_t-x+\widehat\Lambda_t\big]^+ = \big[a(R_t)Y^{-1}_t-x+\widehat\Lambda_t\big]^+.
\end{equation*}
Since $\mu\le 0$, then $t \mapsto Y_t^{-1}$ is non-decreasing. So are also $t \mapsto R_t$ and $t \mapsto \widehat{\Lambda}_t$. Then $\widehat \nu^*$ is non-decreasing thanks to Assumption \ref{ass:profit}-(i) and it is immediate to see that it is also continuous except for a possible lump-sum investment at time zero of size
\begin{equation}\label{eq:nu0}
\widehat \nu^*_0=\big[a(r)-x\big]^+.
\end{equation}
The investment strategy associated to $\widehat \nu^*$ is obtained by setting
\begin{equation}\label{eq:nustar}
\nu^*_t\triangleq\int_0^tY_s\ud \widehat\nu^*_s,\qquad \text{for $t\ge 0$}.
\end{equation}

To prove that $\nu^*$ is admissible, it is enough to verify that 
$$
\int_0^\infty \e^{-\rho t} Y_t \widehat\nu^*_t \ud t <\infty.  
$$
Recall that $r\mapsto a(r)$ is non-decreasing. Then, for any $\eta\in\mathcal A_F$, 
\begin{equation*}
\begin{aligned}
0&\le \e^{-\rho t}Y_t\widehat \nu^*_t\le \e^{-\rho t}a(r+\eta_{\mathrm{max}}t)+\e^{-\rho t}\eta_{\mathrm{max}}\int_0^t Y_tY^{-1}_s\ud s\\
&\le \e^{-\rho t}\Big(a(r+\eta_{\mathrm{max}}t)+\eta_{\mathrm{max}}t\Big),
\end{aligned}
\end{equation*}
where we used $Y_t/Y_s\le 1$ because $\mu\le 0$. The above expression is integrable in view of Assumption \ref{ass:profit}-(i), which means that $\nu^*$  is admissible and optimal for the investor.

It is easy to verify 
\begin{equation}
X_t^{\nu^\ast,\eta} = Y_t \big( x + \widehat{\nu}^*_t - \widehat{\Lambda}_t  \big) \geq a (R_t^{\eta}),\qquad \text{for all $t\ge 0$.}
\end{equation}
If $\widehat \nu^*_0>0$ then $X^{\nu^*,\eta}_0=a(r)$ (cf.\ \eqref{eq:nu0}) and, finally,
\begin{equation}
\ud \widehat{\nu}^*_t = 1_{ \{ X_t^{\nu^\ast,\eta} = a(R_t^{\eta}) \}  } \ud  \widehat{\nu}_t^*= 1_{ \{ X_t^{\nu^\ast,\eta} = a(R_t^{\eta}) \}  } \ud  \overline{\nu}_t^*,
\end{equation}
with $\overline \nu_t^* = Y_t^{-1} a(R_t^{\eta}) - x + \widehat{\Lambda}^\eta_t$, for $t\ge 0$. 
In summary, the investor's optimal strategy is to keep the dynamics $(R^\eta_t,X^{\nu^*,\eta}_t)_{t\ge 0}$ above the threshold $r\mapsto a(r)$, $r\ge 0$. Then \eqref{eq:SC} holds with $\mathcal I=\{(r,x):x>a(r)\}$ and $\mathcal{M}=\{(r,x):0<x\le a(r)\}$ as claimed.
\end{proof}
\begin{remark}\label{rem:nustar}
If $X_0>a(R_0)$, then $\nu^*_t=0$ for $t\in[0,\tau(\eta)]$ where
\begin{equation}\label{eq:taueta}
\tau(\eta)\triangleq\inf\{t\ge 0:X^{0,\eta}_t\le a(R^\eta_t)\}.
\end{equation}
Moreover 
\[
\nu^*_t=a(R^\eta_t)+R^\eta_t-a(R^\eta_{\tau(\eta)})-R^\eta_{\tau(\eta)} -\int_{\tau(\eta)}^t \mu a(R^\eta_s) ds ,\quad \text{for $t\ge \tau(\eta)$}.
\]
\end{remark}
\begin{remark}\label{rem:unique}
It is worth emphasising that the optimal investment policy $\widehat \nu^*_t$ is a Markovian control for the pair $(R^\eta_t,\widehat\Lambda^\eta_t)$ and therefore $\nu^*_t$ is a feedback map of the same pair. The uniqueness statement in the above theorem, implies that such maps are optimal against {\em any} choice of the control process $\eta\in\mathcal{A}_F$ adopted by the firm. This will allow us to establish uniqueness of the equilibrium in the current setup with zero noise.
\end{remark}

Having obtained an optimal strategy for the investor, we now turn our attention to determining the firm's optimal emission reduction policy. It is worth noticing that 
\begin{equation}
X^{\nu^*,\eta}_t=Y_t\big(x-\widehat \Lambda^\eta_t+\big[Y^{-1}_t a(r+\Lambda^\eta_t)-x+\widehat\Lambda^\eta_t\big]^+\big)=\max\big\{Y_t(x-\widehat \Lambda^\eta_t),a(r+\Lambda^\eta_t)\big\},\label{eq:xopt}
\end{equation}
where, to streamline notation, we write
$$
{\Lambda}_t \triangleq \int_0^t  \eta_s \ud s,
$$
and, recalling also $\tau(\eta)$ from \eqref{eq:taueta}, we have $X^{\nu^*,\eta}_t=a(r+\Lambda^\eta_t)$, for $t\ge \tau(\eta)$.
For any $\eta\in\mathcal A_F$ the firm's payoff reads
\begin{equation}\label{eq:JF}
\begin{aligned}
\mathcal J^F_{r,x}(\eta,\nu^*)&=\int_0^\infty\e^{-\bar\rho t} \pi(X^{\nu^*,\eta}_t)\ud t\\
&=\int_0^\infty\e^{-\bar\rho t} \pi\big(\max\big\{Y_t(x-\widehat \Lambda^\eta_t),a(r+\Lambda^\eta_t)\big\}\big)\ud t\\
&=\int_0^{\tau(\eta)}\e^{-\bar\rho t} \pi\big(Y_t[x-\widehat \Lambda^\eta_t]\big)\ud t+\int_{\tau(\eta)}^\infty\e^{-\bar\rho t}\pi\big(a(r+\Lambda^\eta_t)\big)\ud t.
\end{aligned}
\end{equation}

When $R_0=r$ and $X_0=x=a(r)$, we have $\tau(\eta)=0$ for any $\eta\in\mathcal A_F$. Then
\begin{equation}\label{eq:opteta0}
\sup_{\eta\in\mathcal A_F}\mathcal J_{r,a(r)}^F(\eta,\nu^*)=\int_{0}^\infty\e^{-\bar\rho t}\pi\big(a(r+\eta_{\mathrm{max}}t)\big)\ud t,
\end{equation}
and $\eta^*_t=\eta_{\mathrm{max}}$ for all $t\ge 0$ is optimal, because $\pi$ is increasing (Assumption \ref{ass:profit}-(ii)). We then deduce the following simple result.
\begin{proposition}\label{prop:1}
For $(R_0,X_0)=(r,a(r))$, $r\in(0,\infty)$, the unique equilibrium pair $(\eta^*,\nu^*)$ is
\begin{equation*}
\eta^*_t=\eta_{\mathrm{max}},\quad 
\nu^*_t=a(r+\eta_{\mathrm{max}}t)+\eta_{\mathrm{max}}t-a(r)-\mu\int_0^t a(r+\eta_{\mathrm{max}} s)\ud s,\quad\text{for all $t\ge 0$.}
\end{equation*}
\end{proposition}
\begin{proof}
Optimality of $\eta^*=\eta_{\mathrm{max}}$ is guaranteed by \eqref{eq:opteta0}. The optimal $\nu^*$ is given by \eqref{eq:nustar} and it is easy to calculate it explicitly thanks to Remark \ref{rem:nustar}, because $\tau(\eta^*)=0$. 

Regarding the uniqueness claim, it is obvious that changing $\eta^*$ on a set of times of zero Lebesgue measure does not change the value on the right-hand side of \eqref{eq:opteta0}. If $(\eta,\nu)$ is another equilibrium pair, then from Theorem \ref{thm:nustar} we know that for all $t\ge 0$ it must be $\nu_t=\nu^*_t(\eta)$, where $\nu^*_t(\eta)$ is the optimal feedback map given by \eqref{eq:nua2} (cf.\ also Remark \ref{rem:unique}). Then, $\eta$ must be a best response to $\nu^*_t$ and, as a result, $\eta_t=\eta_{\mathrm{max}}$ for a.e.\ $t\ge 0$.
\end{proof}
\begin{remark}\label{rem:wellposed2}
At this stage we can easily see that there may be well-posedness issues for the game if we let both players use singular controls, as anticipated in Remark \ref{rem:wellposed}. 

First of all we notice that the expression for $\widehat\nu^*_t$ in \eqref{eq:nua2} continues to hold also if we generalize the dynamics of $R^\eta$ to $R_t=R_0+\xi_t$, with $t\mapsto\xi_t$ c\`adl\`ag, nondecreasing, $\mathbb F$-adapted, with $\xi_0\ge 0$ and $\mathbb E\big[\int_{[0,\infty)} \e^{-\bar \rho t} \ud\xi_t\big]<\infty$. Moreover, the expression of the boundary $a(r)$ remains the same. Then, the firm's problem in \eqref{eq:opteta0} turns into
\begin{equation*}
\sup_{\xi}\mathcal J_{r,a(r)}^F(\xi,\nu^*)=\sup_\xi\int_{0}^\infty\e^{-\bar\rho t}\pi\big(a(r+\xi_t)\big)\ud t.
\end{equation*}
It is clear that the expression on the right-hand side above can be made arbitrarily large taking an increasing sequence of admissible controls $\{(\xi^n_t)_{t\ge 0},n\in\mathbb N\}$ with $\xi^n_0=n$. Since $\widehat\nu^*$ is optimal for the investor against any choice of the firm's abatement policy $\xi$, we deduce that there is no equilibrium in this game with finite payoff for the firm. A slight modification of this argument allows to reach the same conclusion also when $X_0=x>a(r)=a(R_0)$, so the issue is not specific to one initial state of the dynamics.
\end{remark}

Now we turn our attention to finding an optimal strategy for the firm when $x>a(r)$. For the ease of exposition, we split the analysis into the case when $\mu=0$ and when $\mu<0$.

\subsection{Equilibrium with $\mu=0$}

Let us start from $\mu=0$, in which case we have 
\[
w(r,x)=\sup_{\eta\in\mathcal A_F}\Big(\int_0^{\tau(\eta)}\e^{-\bar\rho t} \pi(x-\Lambda^\eta_t)\ud t+\int_{\tau(\eta)}^\infty\e^{-\bar\rho t}\pi\big(a(r+\Lambda^\eta_t)\big)\ud t\Big).
\]
We need to introduce some quantities of interest in order to state the form of our equilibrium. The firm's payoff with no emission reduction reads
\begin{equation*}
w_0(r,x)\triangleq\mathcal{J}^F_{r,x}(0,\nu^*)=\frac{1}{\bar\rho} \pi(x),\quad x> a(r),
\end{equation*}
because $\tau(0)=\infty$. Instead, when $\eta=\eta_{\mathrm{max}}$ we have $\tau_M\triangleq \tau(\eta_{\mathrm{max}})$ uniquely defined as the solution of
\begin{equation}\label{eq:tauM}
x-\eta_{\mathrm{max}}\tau_M=a(r+\eta_{\mathrm{max}}\tau_M).
\end{equation}
Therefore, the payoff with maximum emission abatement reads
\begin{equation*}
w_1(r,x)\triangleq\mathcal{J}^F_{r,x}(\eta_{\mathrm{max}},\nu^*)=\int_0^{\tau_M}\e^{-\bar\rho t}\pi(x-\eta_{\mathrm{max}}t)\ud t+\int_{\tau_M}^\infty\e^{-\bar\rho t}\pi\big(a(r+\eta_{\mathrm{max}}t)\big)\ud t,
\end{equation*}
which is finite by Assumption \ref{ass:profit}-(ii). 
The functions $w_0$ and $w_1$ determine a set 
\[
\mathcal B\triangleq\{(r,x)\in[0,\infty)\times\mathbb{R}_+:w_0(r,x)<w_1(r,x)\},
\]
which will be useful in finding an equilibrium in this case. For future reference, we notice that $\tau_M=\tau_M(r,x)$ is continuously differentiable in both variables with 
\[
\frac{\partial\tau_M}{\partial x}=\frac{1}{\eta_{\mathrm{max}}\big(1+\dot a(r+\eta_{\mathrm{max}}\tau_M)\big)}\quad\text{and}\quad \frac{\partial\tau_M}{\partial r}=-\frac{\dot a(r+\eta_{\mathrm{max}}\tau_M)}{\eta_{\mathrm{max}}\big(1+\dot a(r+\eta_{\mathrm{max}}\tau_M)\big)}.
\]

Next, we characterise the set $\mathcal B$.
\begin{lemma}\label{bdef.lm}
There is a unique function $r\mapsto b(r)$ such that 
\[
\mathcal B=\{(r,x)\in[0,\infty)\times\mathbb R_+: x<b(r)\}.
\]
Moreover, $b\in C([0,\infty))\cap C^1(\mathbb{R}_+)$ is increasing with $b(r)>a(r)$ for all $r\in[0,\infty)$.
\end{lemma}
\begin{proof}
Let us start by calculating $\partial_r w_1(r,x)$:
\begin{equation}\label{eq:ipp}
\begin{aligned}
\frac{\partial w_1(r,x)}{\partial r}  &= \e^{-\bar\rho\tau_M}\pi(x-\eta_{\mathrm{max}}\tau_M)\frac{\partial\tau_M}{\partial r}-\e^{-\bar\rho\tau_M}\pi\big(a(r+\eta_{\mathrm{max}}\tau_M)\big)\frac{\partial\tau_M}{\partial r}\\
&\quad+\lim_{h\to 0} \int_{\tau_M}^\infty\e^{-\bar\rho t}\frac{1}{h}\left\{\pi\big(a(r+h+\eta_{\mathrm{max}}t)-\pi\big(a(r+\eta_{\mathrm{max}}t)\big)\right\}\ud t \\
& = \lim_{h\to 0} \int_{\tau_M}^\infty\e^{-\bar\rho t}\frac{1}{h}\int_0^h \dot\pi\big(a(r+\xi+\eta_{\mathrm{max}}t)\big)\dot a(r+\xi+\eta_{\mathrm{max}}t)\ud \xi\,\ud t\\
& = \lim_{h\to 0} \frac{1}{h}\int_0^h  \Big(\int_{\tau_M}^\infty\e^{-\bar\rho t} \dot\pi\big(a(r+\xi+\eta_{\mathrm{max}}t)\big)\dot a(r+\xi+\eta_{\mathrm{max}}t)\ud t\Big)\ud\xi\\
& = \lim_{h\to 0} \frac{1}{h \eta_{\max}}\int_0^h\!\!  \Big(\bar \rho\int_{\tau_M}^\infty\!\!\e^{-\bar\rho t} \pi\big(a(r\!+\!\xi\!+\!\eta_{\mathrm{max}}t)\big)\ud t \!-\! e^{-\bar \rho \tau_M}\pi\big(a(r\!+\!\xi \!+\! \eta_{\max} \tau_M)\big)\Big)\ud\xi\\
& = \frac{1}{\eta_{\max}}\Big(\bar \rho\int_{\tau_M}^\infty\e^{-\bar\rho t} \pi\big(a(r+\eta_{\mathrm{max}}t)\big)\ud t - e^{-\bar \rho \tau_M}\pi\big(a(r + \eta_{\max} \tau_M)\big)\Big)\\
&=\int_{\tau_M}^\infty\e^{-\bar\rho t}\dot\pi\big(a(r+\eta_{\mathrm{max}}t)\big)\dot a(r+\eta_{\mathrm{max}}t)\ud t>0,
\end{aligned}
\end{equation}
where the second equality holds because of \eqref{eq:tauM}; to compute the limit we use monotonicity and continuity of $\pi$ and $a$, and Assumption \ref{ass:profit}-(ii); the strict inequality holds by strict monotonicity of $\pi$ and $a$. Then, $r\mapsto w_1(r,x)$ is strictly increasing, whereas $r\mapsto w_0(r,x)$ is constant. That shows that for each $x\in\mathbb R_+$ there is a unique $c(x)\ge 0$ such that permanent emission abatement is strictly more profitable than no action, i.e.,
\[
\mathcal B=\{(r,x)\in[0,\infty)\times\mathbb R_+: r>c(x)\}.
\]
We first notice that $c(x)$ is not constant and equal to zero: for $x\to \infty$ we have $\tau_M\to\infty$ by \eqref{eq:tauM} and Assumption \ref{ass:profit}-(i); therefore, for sufficiently large $x$ we have $w_1(r,x)<w_0(r,x)$ for every $r\ge 0$.

For the $x$-derivatives it is immediate that $\partial_x w_0(r,x)=\dot\pi(x)/\bar\rho$ and the same argument as above yields
\begin{equation}\label{eq:star}
\begin{aligned}
\frac{\partial w_1}{\partial x}(r,x)&=\int_0^{\tau_M}\e^{-\bar \rho t}\dot\pi(x-\eta_{\mathrm{max}}t)\ud t+\e^{-\bar\rho\tau_M}\Big(\pi(x-\eta_{\mathrm{max}}\tau_M)-\pi\big(a(r+\eta_{\mathrm{max}}\tau_M)\big)\frac{\partial\tau_M}{\partial x}\\
&=\int_0^{\tau_M}\e^{-\bar \rho t}\dot\pi(x-\eta_{\mathrm{max}}t)\ud t>0,
\end{aligned}
\end{equation}
where the second equality holds by \eqref{eq:tauM}.

Since $w_0$ and $w_1$ are both continuous, then $w_0(c(x),x)=w_1(c(x),x)$, for $c(x)>0$, and by the implicit function theorem
\begin{equation*}
\dot c(x)=-\frac{\partial_x(w_0-w_1)(c(x),x)}{\partial_r(w_0-w_1)(c(x),x)}=\frac{\bar\rho^{-1}\dot \pi(x)-\int_0^{\tau_M}\e^{-\bar\rho t}\dot \pi(x-\eta_{\mathrm{max}}t)\ud t}{\int_{\tau_M}^\infty\e^{-\bar\rho t}\dot \pi\big(a(r+\eta_{\mathrm{max}}t)\big)\dot a(r+\eta_{\mathrm{max}}t)\ud t}>0,
\end{equation*}
where the inequality holds because, by Assumption \ref{ass:profit}-(ii),
\[
\int_0^{\tau_M}\e^{-\bar\rho t}\dot \pi(x-\eta_{\mathrm{max}}t)\ud t\le \dot \pi(x)\bar\rho^{-1}(1-\e^{-\bar \rho \tau_M}).
\]
Then $c\in C^1(\mathbb R_+)$ and it is strictly increasing. We can also define the strictly increasing inverse boundary $b(r)=c^{-1}(r)$, so that $b\in C^1((0,\infty))$ and 
$\mathcal B=\{(r,x)\in[0,\infty)\times\mathbb R_+: x<b(r)\}$, as claimed. Because the limit $b(0):=\lim_{r\downarrow 0}b(r)$ exists and it is finite, we deduce $b\in C([0,\infty))$.

Finally, we can show that $b(r)>a(r)$ for all $r\in[0,\infty)$ arguing by contradiction. Indeed, assume there is $r_0\in[0,\infty)$ such that $a(r_0)=b(r_0)$. Then, $\tau_M(r_0,b(r_0))=0$ and 
\[
w_1(r_0,b(r_0))=\int_0^\infty\e^{-\bar\rho t}\pi\big(a(r_0+\eta_{\mathrm{max}}t)\big)\ud t>\frac{\pi\big(a(r_0)\big)}{\bar\rho}=w_0(r_0,a(r_0)).
\]
The inequality implies $b(r_0)>a(r_0)$, hence a contradiction.
\end{proof}

Now we can state our result concerning the firm's optimal strategy.
\begin{proposition}\label{prop:optfirmmu0}
Let $\mu=0$. Then, the firm's equilibrium payoff reads as
\begin{equation}\label{eq:weq0}
w(r,x)=\left\{
\begin{array}{ll}
w_0(r,x), & x>b(r),\\
w_1(r,x), & a(r)\le x\le b(r).
\end{array}
\right.
\end{equation}
The firm's optimal strategy is constant and equal to zero if $x>b(r)$ (i.e., $\eta^*_t\equiv 0$ for $t\ge 0$), whereas it is constant and equal to $\eta_{\mathrm{max}}$ if $a(r)\le x\le b(r)$ (i.e., $\eta^*_t=\eta_{\mathrm{max}}$ for $t\ge 0$). The optimal control is unique in the sense that if $\eta\in\mathcal{A}_F$ is another optimal control, then $\eta_t=\eta^*_t$ for a.e.\ $t\ge 0$.
\end{proposition}

\begin{proof}
Our strategy of proof is to check that the function $w$ from \eqref{eq:weq0} solves the HJB equation and then to apply a small adaptation of the verification theorem (Corollary \ref{cor:verif}). 
Recall that in the present case, $\mathcal L\equiv 0$ and the HJB equation takes the simple form $-\bar \rho w(r,x) + \mathcal H(r,x;w) = - \pi(x)$, for $x> a(r)$, with $\mathcal H(r,x;w)$ as in \eqref{eq:Ham}.

For $x>b(r)$ we have $(w_r-w_x)(r,x)=-\bar\rho^{-1}\dot \pi(x)<0$, thus  
$\mathcal H(r,x;w)=0$. Then it is clear that the HJB holds for $x\ge b(r)$ because $-\bar\rho w_0(r,x)=-\pi(x)$. It is immediate to verify that $w_x(r,a(r))=\partial_x w_1(r,a(r))=0$ by the explicit form of the derivative calculated in \eqref{eq:star} because $\tau_M=0$ in this case. Moreover, for $a(r)<x\le b(r)$
\begin{equation}\label{eq:grad0}
\begin{aligned}
(w_r-w_x)(r,x)&=\int_{\tau_M}^\infty\e^{-\bar\rho t}\dot\pi\big(a(r+\eta_{\mathrm{max}}t)\big)\dot a(r+\eta_{\mathrm{max}}t)\ud t-\int_0^{\tau_M}\e^{-\bar\rho t}\dot\pi(x-\eta_{\mathrm{max}}t)\ud t.
\end{aligned}
\end{equation}
It is easy to see from the first expression above that $(w_r-w_x)(r,a(r))>0$. 
Moreover, for $a(r)<x<b(r)$ we have, by integration by parts, using \eqref{eq:tauM} and Assumption \ref{ass:profit}-(ii),
\begin{equation*}
\begin{aligned}
&\bar\rho w_1(r,x)-\pi(x)\\
&=\bar\rho\int_0^{\tau_M}\e^{-\bar\rho t}\pi\big(x-\eta_{\mathrm{max}}t\big)\ud t+\bar\rho\int_{\tau_M}^\infty\e^{-\bar\rho t}\pi\big(a(r+\eta_{\mathrm{max}}t)\big)\ud t-\pi(x)\\
& = -\eta_{\max} \int_0^{\tau_M} e^{-\bar \rho t} \dot \pi(x-\eta_{\max}t) \ud t + \eta_{\max}\int_{\tau_M}^\infty e^{-\bar \rho t} \dot \pi(a(r+\eta_{\max}t)) \dot a(r+\eta_{\max}t) \ud t.
\end{aligned}
\end{equation*}
Comparing to \eqref{eq:grad0} we see that 
\begin{equation}\label{eq:HJBw1}
\eta_{\mathrm{max}}\big(\partial_r w_1(r,x)-\partial_x w_1(r,x)\big)=\bar \rho w_1(r,x)-\pi(x).
\end{equation}
The right-hand side of the above expression is positive because $a(r)<x<b(r)$ (this follows from the definition of the set $\mathcal B$ and Lemma \ref{bdef.lm}). Therefore, $\partial_r w_1(r,x)-\partial_x w_1(r,x)>0$ and
$\eta_{\mathrm{max}}\big(\partial_r w_1(r,x)-\partial_x w_1(r,x)\big)=\mathcal H(r,x;w_1)$. 
Thus, \eqref{eq:HJBw1} shows that $w_1$ solves the HJB in $a(r)<x\le b(r)$, $r\ge 0$. Notice that indeed \eqref{eq:HJBw1} also implies $\partial_r w_1(r,b(r))-\partial_x w_1(r,b(r))=0$.

We have shown that $w$ solves the HJB equation at all points $(r,x)$ with $x\ge a(r)$ and $x\neq b(r)$. We cannot directly apply Corollary \ref{cor:verif} with Remark \ref{rem:verif}, because $w$ defined as in \eqref{eq:weq0} is not continuously differentiable. However, the mapping $(r,x)\mapsto (\partial_r w-\partial_x w)(r,x)$ is well-defined as a function in $L^\infty_{\ell oc}$ on the set $\{(r,x):x\ge a(r)\}$ (and $(r,x)\mapsto \mathcal{H}(r,x;w)$ is even continuous). Since there is no dynamics, unless $\eta_t\neq 0$, such regularity is sufficient to make sense of the change of variable formula at the beginning of the proof of Theorem \ref{thm:verif}. The rest follows by the same arguments as in that proof and we omit details to avoid repetitions. To check the transversality condition \eqref{eq:transv.0} for the function $w$, notice that 
$$
w(r,x) \leq \frac{\pi(x)}{\bar \rho} + \int_0^\infty \e^{-\bar \rho t} \pi(a(r + \eta_{\max} t))\ud t = \frac{\pi(x)}{\bar \rho} + \e^{r\bar \rho/\eta_{\max}}\int_{r/\eta_{\max}}^\infty \e^{-\bar \rho t} \pi(a( \eta_{\max} t))\ud t.
$$
Then, recalling that $w$ is increasing in $r$, 
\begin{align*}
\e^{-\bar \rho t} w(R^{\eta^*}_t,X^{\nu^*,\eta^*}_t) &\leq \e^{-\bar \rho t} w(r+\eta_{\max} t,X^{\nu^*,\eta^*}_t) \\
&\leq \frac{\e^{-\bar \rho t}\pi(X^{\nu^*,\eta^*}_t)}{\bar \rho} + \e^{r\bar \rho/\eta_{\max}}\int_{r/\eta_{\max} +t}^\infty \e^{-\bar \rho t} \pi(a( \eta_{\max} t))\ud t.
\end{align*}
The second term converges to zero as $t\to \infty$ by Assumption \ref{ass:profit}-(ii). To estimate the first term, from \eqref{eq:xopt} we deduce 
$X^{\nu^*,\eta^*}_t\leq \max\{x,a(r+\eta_{\max}t)\}$.
From Assumption \ref{ass:profit}-(ii), we deduce that 
$$
\lim_{t\to \infty}\e^{-\bar \rho t} \pi(a(r+\eta_{\max}t))  = 0,
$$
which implies
$\lim_{t\to \infty}\bar \rho^{-1}\e^{-\bar \rho t}\pi(X^{\nu^*,\eta^*}_t) = 0$
as well. 

Finally, we come to the claim about uniqueness. Let $\eta\in\mathcal{A}_F$ be another optimal control and let us assume by way of contradiction that $\eta_t\neq\eta^*_t$ for $t\ge 0$ on a set of positive Lebesgue measure. More concretely, fix $x>b(r)$ and assume that $\eta_t>0$ for $t\in\mathcal I$ where $\int_0^T 1_{\{t\in\mathcal I\}}\ud t>0$ for some $T>0$ and, indeed, with no loss of generality we can assume $\mathcal I\subset[0,T]$. From the first equation in the proof of Theorem \ref{thm:verif} we obtain 
\begin{equation}
\begin{aligned}
\e^{-\bar\rho t}w\big(R^\eta_{t},X^{\nu^*,\eta}_{t}\big)
&=w(r,x)+\int_0^{t}\e^{-\bar\rho s}\big[-\rho w+\eta_s\big(w_r-w_x\big)\big]\big(R^\eta_s,X^{\nu^*,\eta}_s\big)\ud s\\
& < w(r,x)+\int_0^{t}\e^{-\bar\rho s}\big[-\bar\rho w+\mathcal H(\cdot,w)\big]\big(R^\eta_s,X^{\nu^*,\eta}_s\big)\ud s\\
& = w(r,x)-\int_0^{t}\e^{-\bar\rho s}\pi\big(X^{\nu^*,\eta}_s\big)\ud s,
\end{aligned}
\end{equation}
where the strict inequality holds for any $t>T$ because 
\[
\eta_s\big(w_r-w_x\big)\big(R^\eta_s,X^{\nu^*,\eta}_s\big)<\mathcal H(\cdot,w)\big(R^\eta_s,X^{\nu^*,\eta}_s\big),
\]
for all $s\in\mathcal{I}$ such that $X^{\nu^*,\eta}_s>b(R^\eta_s)$ (notice that the set $\{s\in\mathcal I: X^{\nu^*,\eta}_s>b(R^\eta_s)\}$ has positive measure by continuity of trajectories). The strict inequality is preserved when letting $t\to \infty$ because $\mathcal I\subset[0,T]$. Since $w\ge 0$, rearranging terms we reach a contradiction because
\[
w(r,x)>\int_0^{\infty}\e^{-\bar\rho s}\pi\big(X^{\nu^*,\eta}_s\big)\ud s=\mathcal{J}^F_{r,x}(\eta,\nu^*),
\]
and $\eta$ cannot be optimal. By analogous arguments, starting from $a(r)<x<b(r)$ and assuming $\eta_t<\eta_{\mathrm{max}}$ for $t\in\mathcal I$ we arrive at the same contradiction. This shows uniqueness of the optimal control.
\end{proof}

Combining Theorem \ref{thm:nustar} and Proposition \ref{prop:optfirmmu0} we have obtained the unique equilibrium, summarized in the next corollary. The proof of uniqueness repeats arguments similar to those in the proof of Proposition \ref{prop:1} and it is provided after the corollary for completeness. combining uniqueness of $\nu^*$ from Theorem \ref{thm:nustar} and of $\eta^*$ from Proposition \ref{prop:optfirmmu0}, and it is therefore omitted.
\begin{corollary}\label{cor:NE2}
Let the controlled dynamics be as in \eqref{eqdet} with $\mu=0$. Then, the unique equilibrium pair is given by
\begin{equation}
(\eta^*,\nu^*)=\left\{
\begin{array}{ll}
(0,0),& \text{for $x>b(r)$},\\
(\eta_{\mathrm{max}},\nu^a), & \text{for $a(r)\le x\le b(r)$},
\end{array}
\right.
\end{equation}
with $\nu^a$ as in \eqref{eq:nua}.    
\end{corollary}
\begin{proof}
We only need to prove uniqueness of the equilibrium. It is obvious that changing $\eta^*$ on a set of times of zero Lebesgue measure does not change the value function $w(r,x)=\mathcal{J}^F_{r,x}(\eta^*,\nu^*)$ with payoff given by \eqref{eq:JF}. If $(\eta,\nu)$ is another equilibrium pair, then from Theorem \ref{thm:nustar} we know that for all $t\ge 0$ it must be $\nu_t=\nu^*_t$, where $\nu^*_t=\nu^*_t(\eta)$ is given by \eqref{eq:nua2} (cf.\ also Remark \ref{rem:unique}). Then, $\eta$ must be a best response to $\nu^*_t$ and, by uniqueness of the optimal control in Proposition \ref{prop:optfirmmu0}, it must be $\eta_t=\eta^*_t$ for a.e.\ $t\ge 0$.
\end{proof}

\subsection{Equilibrium with $\mu<0$ \label{sec:detmunegative}}
Now we consider $\mu<0$. Almost the entirety of our arguments hold for a generic $\pi(x)$ satisfying Assumption \ref{ass:profit}-(ii) and we only need Assumption \ref{ass:profit}-(iii) in the proof of Lemma \ref{lem:h}. For any $\eta\in\mathcal A_F$ the firm's payoff reads
\begin{equation*}
\begin{aligned}
\mathcal J^F_{r,x}(\eta,\nu^*)=\int_0^{\tau(\eta)}\e^{-\bar\rho t} \pi\big(Y_t[x-\widehat \Lambda^\eta_t])\ud t+\int_{\tau(\eta)}^\infty\e^{-\bar\rho t}\pi\big(a(r+\Lambda^\eta_t)\big)\ud t,
\end{aligned}
\end{equation*}
which is finite by Assumption \ref{ass:profit}-(ii). 
Motivated by the calculations for $\mu=0$ we start by considering $\eta_t=\eta_{\mathrm{max}}$ and obtain
\[
\widehat \Lambda_t=\eta_{\mathrm{max}}\int_0^t\e^{|\mu|s}\ud s=\frac{\eta_{\mathrm{max}}}{|\mu|}\big(\e^{|\mu|t}-1\big)\triangleq f_M(t).
\]
For later use we notice that 
\begin{equation}\label{eq:fM}
|\mu|f_M(t)=\frac{\eta_{\mathrm{max}}}{Y_t}-\eta_{\mathrm{max}}\text{ and }\dot f_M(t)=\eta_{\mathrm{max}}\e^{|\mu|t}=\frac{\eta_{\mathrm{max}}}{Y_t}.
\end{equation}
The associated payoff is denoted by $w_1$ with a slight abuse of notation. It reads
\begin{equation}\label{eq:w1.b}
\begin{aligned}
w_1(r,x)&=\mathcal J^F_{r,x}(\eta_{\mathrm{max}},\nu^*)\\
&=\int_0^{\tau_M}\e^{-\bar\rho t}\pi\big(Y_t [x-f_M(t)]\big)\ud t+\int_{\tau_M}^\infty\e^{-\bar\rho t}\pi\big(a(r+\eta_{\mathrm{max}} t)\big)\ud t,
\end{aligned}
\end{equation}
where
$$
\tau_M=\tau_M(r,x) = \inf\{t\geq 0: x - f_M(t)\leq e^{|\mu| t} a(r + \eta_{\max}t)\}
$$
Clearly, $\tau_M = 0$ for $x\leq a(r)$ and it is easy to show that for $x>a(r)$, $\tau_M$ is the unique solution of 
\begin{equation}\label{eq:tauM2}
x-f_M(\tau_M)=\e^{|\mu| \tau_M}a\big(r+\eta_{\mathrm{max}}\tau_M\big).
\end{equation}

For later use we observe that $(r,x)\mapsto \tau_M(r,x)$ is continuously differentiable on $\{x>a(r)\}$, with derivatives that we now proceed to calculating. Taking $x$-derivative of \eqref{eq:tauM2},
\begin{equation*}
\begin{aligned}
1=\Big(\dot f_M(\tau_M)+ |\mu|\e^{|\mu|\tau_M}a\big(r+\eta_{\mathrm{max}}\tau_M\big)+\e^{|\mu|\tau_M}\dot a\big(r+\eta_{\mathrm{max}}\tau_M\big)\eta_{\mathrm{max}}\Big)\frac{\partial \tau_M}{\partial x}.
\end{aligned}
\end{equation*}
Using the form of $\dot f_M(t)$ we get
\begin{equation*}
\begin{aligned}
1=\Big[ |\mu|\e^{|\mu|\tau_M}a\big(r+\eta_{\mathrm{max}}\tau_M\big)+\eta_{\mathrm{max}}\e^{|\mu|\tau_M}\Big(1+\dot a\big(r+\eta_{\mathrm{max}}\tau_M\big)\Big)\Big]\frac{\partial \tau_M}{\partial x}.
\end{aligned}
\end{equation*}
Using \eqref{eq:tauM2} and the explicit form of $|\mu|f_M(t)$ we get
\begin{equation*}
\begin{aligned}
1=\Big[ |\mu|x-\eta_{\mathrm{max}}\e^{|\mu|\tau_M}+\eta_{\mathrm{max}}+\eta_{\mathrm{max}}\e^{|\mu|\tau_M}\Big(1+\dot a\big(r+\eta_{\mathrm{max}}\tau_M\big)\Big)\Big]\frac{\partial \tau_M}{\partial x}.
\end{aligned}
\end{equation*}
Hence,
\begin{equation}\label{eq:dtMdx}
\begin{aligned}
1=\Big[ |\mu|x+\eta_{\mathrm{max}}+\eta_{\mathrm{max}}\e^{|\mu|\tau_M}\dot a\big(r+\eta_{\mathrm{max}}\tau_M\big)\Big]\frac{\partial \tau_M}{\partial x},
\end{aligned}
\end{equation}
which also yields $\frac{\partial \tau_M}{\partial x}>0$. Taking $r$-derivative of \eqref{eq:tauM2} and arguing in a similar way as above we have
\begin{equation}\label{eq:tauMdr}
    -\dot{a}(r+\eta_{\mathrm{max}} \tau_M) = \big(\eta_{\mathrm{max}} + |\mu| a(r+\eta_{\mathrm{max}} \tau_M) + \eta_{\mathrm{max}} \dot a (r+\eta_{\mathrm{max}} \tau_M) \big) \frac{\partial \tau_M}{\partial r}.
\end{equation}
Then, $\frac{\partial \tau_M}{\partial r} < 0$.

The function $w_1$ is our candidate for the firm's value at equilibrium when $x>a(r)$ is sufficiently close to $a(r)$. We will make this statement rigorous later. For now let us state an initial result concerning $w_1$. We recall that in this section 
\begin{equation}\label{eq:MM}
\mathcal{M}=\{(r,x)\in[0,\infty)\times\mathbb{R}_+:0<x\le a(r)\} \ \text{ and }\ \mathcal I=\{(r,x)\in[0,\infty)\times\mathbb{R}_+:x>a(r)\},
\end{equation}
because of Theorem \ref{thm:nustar}

\begin{proposition}\label{prop:w1}
 For $w_1$ as in \eqref{eq:w1.b} we have $w_1\in C^1(\overline{\mathcal I})$ and   
\begin{equation}\label{eq:HJBmu0}
\begin{aligned}
\bar\rho w_1(r,x)-\pi(x)&=\mu x \partial_x w_1(r,x)+\eta_{\mathrm{max}}\big[\partial_r w_1(r,x)-\partial_x w_1(r,x)\big],\quad \text{for }(r,x)\in(0,\infty)^2.
\end{aligned}
\end{equation}
\end{proposition}
\begin{proof}
We start by computing $\partial_r w_1$ and $\partial_x w_1$. Their explicit expressions will imply $w_1\in C^1(\overline{\mathcal I})$ as claimed.
Let us first consider the derivative with respect to $r$. 

By following the same arguments as in \eqref{eq:ipp}, we obtain
\begin{equation}\label{eq:drw1}
\frac{\partial w_1}{\partial r}(r,x)=\int_{\tau_M}^\infty\e^{-\bar\rho t}\dot\pi(a(r+\eta_{\max} t))\dot a(r+\eta_{\max}t) \ud t> 0.
\end{equation}
This integral is finite because, by integration by parts, 
$$
\int_{\tau_M}^\infty\e^{-\bar\rho t}\dot\pi(a(r+\eta_{\max} t))\dot a(r+\eta_{\max}t) \ud t = \frac{\bar \rho}{\eta_{\max}} \int_{\tau_M}^\infty \e^{-\bar \rho t} \pi\big(a(r+\eta_{\max}t)\big)\ud t - \frac{1}{\eta_{\max}} \e^{-\bar \rho\tau_M} \pi\big(a(r+\eta_{\max}\tau_M)\big), 
$$
and the final expression is finite by Assumption \ref{ass:profit}-(ii). 

Turning now to the derivative with respect to $x$, using \eqref{eq:tauM2} and similar calculations to \eqref{eq:star}, we obtain:
$$
\frac{\partial w_1}{\partial x}(r,x) = \lim_{h\to 0}  \frac{1}{h} \int_0^{\tau_M}\e^{-\bar \rho t} \int_0^h \dot \pi\big(Y_t[x+\xi - f_M(t)]\big)Y_t \ud \xi\,\ud t. 
$$
Using the monotonicity of $\dot \pi$ and the monotone convergence theorem, it is then easy to conclude that 
\begin{equation}\label{eq:dxw1}
\begin{aligned}
\frac{\partial w_1}{\partial x}(r,x)&=
\int_0^{\tau_M}\e^{-\bar\rho t}\dot \pi\big(Y_t[x-f_M(t)]\big)Y_t\ud t> 0.
\end{aligned}
\end{equation}
The above integral is finite because it is computed over a bounded interval and the integrand is continuous.
Combining the expressions for the two derivatives, we obtain
\begin{equation}\label{eq:drdx}
\begin{aligned}
(\partial_r w_1-\partial_x w_1)(r,x)&=\int_{\tau_M}^\infty\e^{-\bar\rho t}\dot \pi\big(a(r+\eta_{\mathrm{max}} t)\big)\dot a(r+\eta_{\mathrm{max}} t)\ud t\\
&\quad-\int_0^{\tau_M}\e^{-\bar\rho t}\dot \pi\big(Y_t[x-f_M(t)]\big)Y_t\ud t.
\end{aligned}
\end{equation}
Moreover, $\mu x\partial_x w_1(r,x)<0$ for $x>a(r)$ with
\begin{equation*}
\mu x\partial_x w_1(r,x)=
\mu x \int_0^{\tau_M}\e^{-\bar\rho t}\dot \pi\big(Y_t[x-f_M(t)]\big)Y_t\ud t,
\end{equation*}
by \eqref{eq:dxw1} and $\mu<0$.

Now let us calculate $\bar\rho w_1(r,x)-\pi(x)$ and compare it to $\mu x \partial_x w_1(r,x)+\eta_{\mathrm{max}}(\partial_r w_1-\partial_x w_1)(r,x)$. Rewriting $\bar\rho\e^{-\bar\rho t}\ud t=-\ud \e^{-\bar\rho t}$ and integrating by parts, we have
\begin{equation*}
\begin{aligned}
&\bar\rho w_1(r,x)-\pi(x)\\
&=-\pi(x)-\int_0^{\tau_M}\pi\big(Y_t [x-f_M(t)]\big)\ud \e^{-\bar\rho t}-\int_{\tau_M}^\infty\pi\big(a(r+\eta_{\mathrm{max}} t)\big)\ud \e^{-\bar\rho t} \\
&=\int_0^{\tau_M}\e^{-\bar \rho t}\dot \pi\big(Y_t[x-f_M(t)]\big)Y_t\big(\mu x-\mu f_M(t)-\dot f_M(t)\big)\ud t\\
&\quad+\eta_{\mathrm{max}}\int_{\tau_M}^\infty\e^{-\bar \rho t}\dot \pi\big(a(r+\eta_{\mathrm{max}}t)\big)\dot a(r+\eta_{\mathrm{max}}t)\ud t,
\end{aligned}
\end{equation*}
where we also used \eqref{eq:tauM2} to cancel a term resulting from integration by parts. Notice that $-\mu f_M(t)-\dot f_M(t)=-\eta_{\mathrm{max}}$ due to \eqref{eq:fM} and therefore, it is easy to verify that \eqref{eq:HJBmu0} holds.
\end{proof}

It is clear from the form of \eqref{eq:drdx} that for sufficiently large values of $x$ it must be $(\partial_r w_1-\partial_x w_1)(r,x)<0$ because $\lim_{x\to\infty}\tau_M(r,x)=\infty$ and $\dot \pi$ is non-decreasing (cf.\ Assumption \ref{ass:profit}). Then, $w_1$ cannot be a solution of the HJB equation in the whole space but at most on the set where $(\partial_r w_1-\partial_x w_1)(r,x)\ge 0$. This motivates the next part of the analysis.

By definition of $\tau(0)=\tau_{r,x}(0)$ (cf.\ \eqref{eq:taueta}) we have
\[
x\e^{\mu \tau_{r,x}(0)}=a(r)\iff \tau_{r,x}(0)=\frac{1}{|\mu|}\ln\frac{x}{a(r)}.
\]
We want to consider a class of firm's strategies of the
form $\eta^\tau_t\triangleq\eta_{\mathrm{max}}1_{\{t\ge \tau\}}$, for a generic $\tau\ge 0$. These are strategies according to which the firm is idle until time $\tau$ and then it starts exerting emission abatement at the maximum rate. The associated payoff reads  
\begin{equation*}
\begin{aligned}
\mathcal J^F_{r,x}(\eta^\tau,\nu^*)&=\int_0^{\tau\wedge \tau_{r,x}(0)}\e^{-\bar\rho t}\pi\big(xY_t\big)\ud t+\e^{-\bar\rho(\tau\wedge \tau_{r,x}(0))}w_1(r,xY_{\tau\wedge \tau_{r,x}(0)})\triangleq \mathcal{J}_{r,x}(\tau),
\end{aligned}
\end{equation*}
where, at time $\tau\wedge \tau(0)$ the firm receives the payoff associated to exerting maximum control. 
Then, with a slight abuse of notation, we set 
\begin{equation}\label{eq:OSP}
w_0(r,x)\triangleq \sup_{\tau\ge 0}\mathcal{J}_{r,x}(\tau).
\end{equation}

The problem in \eqref{eq:OSP} is a deterministic optimal stopping problem. An explicit solution seems difficult to obtain but we can still rely on the general optimal stopping theory to find a characterization of the optimal stopping rule. Clearly $w_0(r,x)\ge\mathcal{J}_{r,x}(0)= w_1(r,x)$ and next we establish continuity of $w_0$. 
\begin{lemma}\label{lem:loclip}
For any compact $K\subset\overline{\mathcal{I}}$ there is $c_K>0$ such that 
\[
\big|w_0(r_1,x_1)-w_0(r_2,x_2)\big|\le c_K\big(|x_2-x_1|+|r_2-r_1|\big),\quad \text{for } (r,x)\in K.
\]
\end{lemma}
\begin{proof}
For $r_1,r_2\in[0,\infty)$ and $x_1,x_2\in[0,\infty)$, set $\tau_1=\tau_{r_1,x_1}(0)$ and $\tau_2=\tau_{r_2,x_2}(0)$ for notational simplicity. Then, 
\begin{equation*}
\begin{aligned}
&\big|w_0(r_1,x_1)-w_0(r_2,x_2)\big|\\
&\le\int_0^\infty\e^{-\bar\rho t}|\pi(x_2Y_t)-\pi(x_1Y_t)|\ud t+\int_{\tau_1\wedge\tau_2}^{\tau_1\vee\tau_2}\e^{-\bar\rho t}\pi(x_1 Y_t)\vee\pi(x_2Y_t)\ud t\\
&\quad +\bar \rho|\tau_2-\tau_1| w_1(r_1,x_1) + \sup_\tau\big|w_1(r_1,x_1 Y_{\tau\wedge\tau_1}) - w_1\big(r_1,x_1 Y_{\tau\wedge\tau_2}\vee a(r_1)\big)\big|\\
&\quad +\sup_\tau\big|w_1\big(r_1,x_1 Y_{\tau\wedge\tau_2}\vee a(r_1)\big)-w_1(r_2,x_2 Y_{\tau\wedge\tau_2})\big|,
\end{aligned}
\end{equation*}
where we used $w_1(r,xY_t)\le w_1(r,x)$ by monotonicity of $w_1$ and $Y_t\le 1$. Since $(r,x)\mapsto \tau_{r,x}(0)$ and $(r,x)\mapsto w_1(r,x)$ are locally Lipschitz continuous in $\overline{\mathcal I}$, and $r\mapsto a(r)$ is locally Lipschitz in $\mathbb{R}_+$, then 
it is not hard to see that $(r,x)\mapsto w_0(r,x)$ is also locally Lipschitz continuous in $\overline{\mathcal I}$, as claimed. 
\end{proof}
Thanks to continuity of $w_0$, standard optimal stopping theory guarantees that it is optimal to stop at 
\begin{equation}\label{eq:tau*}
\tau_*=\inf\{t\ge 0: w_0(r,xY_t)=w_1(r,xY_t)\}\wedge \tau_{r,x}(0).
\end{equation}
Moreover, 
\begin{equation}\label{eq:Mp}
s\mapsto\int_0^{s\wedge \tau_*}\e^{-\bar\rho t}\pi(xY_t)\ud t+\e^{-\bar\rho(s\wedge \tau_*)}w_0(r,x Y_{s\wedge \tau_*}),
\end{equation}
must be constant (a deterministic martingale).
Now, it is natural to split the state space into the so-called {\em continuation} and {\em stopping} sets, defined respectively as
\[
\mathcal C\triangleq\{(r,x)\in[0,\infty)\times\mathbb{R}_+:w_0(r,x)>w_1(r,x)\}\quad\text{and}\quad\mathcal S=\mathcal C^c\triangleq \big([0,\infty)\times\mathbb{R}_+\big)\setminus\mathcal C.
\]
By construction $(r,a(r))\in\mathcal S$ for any $r\in[0,\infty)$. Moreover, writing
\[
\e^{-\bar\rho(\tau\wedge\tau_{r,x}(0))}w_1(r,xY_{\tau\wedge\tau_{r,x}(0)})=w_1(r,x)+\int_0^{\tau\wedge\tau_{r,x}(0)}\e^{-\bar\rho t}\Big(\mu xY_t\partial_x w_1\big(r,xY_t\big)-\bar\rho w_1\big(r,xY_t\big)\Big)\ud t,
\]
we deduce
\begin{equation}\label{eq:w0w1}
\varphi(r,x)\triangleq w_0(r,x)-w_1(r,x)=\sup_\tau\int_0^{\tau\wedge\tau_{r,x}(0)}\e^{-\bar\rho t}h(r,xY_t)\ud t,
\end{equation}
where
\begin{equation}\label{eq:h}
\begin{aligned}
h(r,x)&\triangleq-|\mu| x\partial_x w_1(r,x)-\bar\rho w_1(r,x)+\pi(x)\\
&=-\eta_{\mathrm{max}}\big(\partial_r w_1-\partial_x w_1\big)(r,x).
\end{aligned}
\end{equation}
Thanks to this formulation, it is possible to establish the form of the optimal stopping rule \eqref{eq:tau*}.
\begin{lemma}\label{lem:h}
We have $h\in C^1(\overline{\mathcal I})$. For $x\ge a(r)$ and $r>0$ we have $\partial_r h<0$ and $\partial_x h>0$. Moreover, $h(r,a(r))<0$ and $\lim_{x\to\infty}h(r,x)>0$. \end{lemma}
Before proving the lemma, we obtain one important consequence thereof.
\begin{proposition}\label{Proposition_funcb}
It is optimal to stop at 
\begin{equation}\label{eq:tau*b}
\tau_*=\tau_b\triangleq\inf\{t\ge 0:xY_t\le b(r)\},
\end{equation}
where $b(r)>a(r)$ is the unique solution of $h(r,b(r))=0$ for $r\in[0,\infty)$. Moreover, the function $b$ is increasing and $b\in C([0,\infty))\cap C^1(\mathbb{R}_+)$.   
\end{proposition}
\begin{proof}
By Lemma \ref{lem:h}, $t\mapsto h(r,xY_t)$ is decreasing. Therefore, the maximum on the right-hand side of \eqref{eq:w0w1} is uniquely attained at $\tau_b$. Moreover, by the implicit function theorem we have
\begin{equation}\label{eq:bdot}
\dot b(r)=-\frac{\partial_r h(r,b(r))}{\partial_x h(r,b(r))}>0,\quad r\in\mathbb{R}_+,
\end{equation}
which concludes the proof of $b\in C^1(\mathbb R_+)$. The limit $b(0):=\lim_{r\downarrow 0}b(r)$ exists and it is finite, due to properties of $h(\cdot,\cdot)$. Then, $b\in C([0,\infty))$.
\end{proof}
It remains to prove the lemma. 
\begin{proof}[{\bf Proof of Lemma \ref{lem:h}}] We first observe that, due to $\partial_x w_1(r,a(r))=0$ (cf.\ \eqref{eq:dxw1}), we have $h(r,a(r))=-\bar\rho w_1(r,a(r))+\pi(a(r))$. However, Assumption \ref{ass:profit}-(ii) and \eqref{eq:w1.b} also imply 
\[
w_1(r,a(r))=\int_0^\infty \e^{-\bar\rho t}\pi(a(r+\eta_{\mathrm{max}}t))\ud t>\bar\rho^{-1}\pi\big(a(r)\big),
\]
hence $h(r,a(r))<0$. By the definition of $h$ and \eqref{eq:drdx} we have by monotone convergence 
\begin{equation*}
\begin{aligned}
\lim_{x\to\infty}h(r,x)&=-\eta_{\mathrm{max}}\lim_{x\to\infty}\big[\partial_r w_1(r,x)-\partial_x w_1(r,x)\big]= \eta_{\mathrm{max}}\lim_{x\to\infty}\int_0^\infty \e^{-\bar \rho t}\dot \pi\big(Y_t[x-f_M(t)]\big)Y_t\ud t\\
&=\eta_{\mathrm{max}}\int_0^\infty \e^{-\bar \rho t}\lim_{x\to\infty}\dot \pi\big(Y_t[x-f_M(t)]\big)Y_t\ud t>0,
\end{aligned}
\end{equation*}
where the final inequality holds by \eqref{eq:drdx} and Assumption \ref{ass:profit}-(ii). This concludes the proof of the final statement in the lemma.

Next we calculate the gradient of $h$. First we recall \eqref{eq:drdx} and notice that we can rewrite, by integration by parts,
\begin{equation}\label{heq}
\begin{aligned}
  h(r,x) &= -\eta_{\mathrm{max}}\big(\partial_r w_1-\partial_x w_1\big)(r,x)\\
  &=  -\int_{\tau_M}^\infty \e^{-\bar \rho t} \dot\pi(a(r+\eta_{\max} t))\dot a(r+\eta_{\max} t) \eta_{\max} \ud t + \eta_{\max} \int_0^{\tau_M} \e^{-\bar \rho t} \dot \pi(Y_t[x-f_M(t)]) Y_t \ud t\\
  & = \e^{-\bar\rho \tau_M} \pi(a(r+\eta_{\max}\tau_M))- \bar \rho \int_{\tau_M}^\infty \e^{-\bar \rho t} \pi(a(r+\eta_{\max} t))\ud t\\ 
  &\quad + \eta_{\max} \int_0^{\tau_M} \e^{-\bar \rho t} \dot \pi(Y_t[x-f_M(t)])Y_t\ud t.
\end{aligned}
\end{equation}
Computing the derivative with respect to $r$, arguing as in \eqref{eq:ipp} to justify interchanging the limit and the integral, we get
\begin{equation*}
\begin{aligned}
 \partial_r h(r,x) 
 &= -\bar \rho \e^{-\bar\rho \tau_M}  \pi\big(a(r\!+\!\eta_{\max}\tau_M)\big)\frac{\partial \tau_M}{\partial r}\!+\!\e^{-\bar\rho \tau_M}\dot \pi\big(a(r\!+\!\eta_{\max}\tau_M)\big) \dot a(r\!+\!\eta_{\max}\tau_M)\Big(1\!+\!\eta_{\max}\frac{\partial \tau_M}{\partial r}\Big) \\ 
  &\quad+ \bar \rho \e^{-\bar \rho \tau_M} \pi\big(a(r+\eta_{\max} \tau_M)\big)\frac{\partial\tau_M}{\partial r}- \bar \rho \int_{\tau_M}^\infty \e^{-\bar \rho t} \dot\pi\big(a(r+\eta_{\max} t)\big) \dot a(r+\eta_{\max} t)\ud t \\ 
  &\quad+ \eta_{\max}  \e^{-\bar \rho \tau_M} \dot \pi\big(Y_{\tau_M}[x-f_M(\tau_M)]\big)Y_{\tau_M}\frac{\partial \tau_M}{\partial r}\\
  & =  \e^{-\bar \rho \tau_M} \dot \pi\big(a(r+\eta_{\max}\tau_M)\big)\Big[ \dot a(r+\eta_{\max}\tau_M)\Big(1+\eta_{\max}\frac{\partial \tau_M}{\partial r}\Big)+\eta_{\max} Y_{\tau_M}\frac{\partial \tau_M}{\partial r}\Big] \\ 
  &\quad- \eta_{\max}\bar \rho \int_{\tau_M}^\infty \!\!\e^{-\bar \rho t} \dot\pi\big(a(r+\eta_{\max} t)\big) \dot a(r+\eta_{\max} t)\ud t. 
\end{aligned}
\end{equation*}
The finiteness of the integral can be justified by another integration by parts and Assumption \ref{ass:profit}-(ii). This shows the differentiability with respect to $r$. To show that the derivative is negative, recall that $\frac{\partial \tau_M}{\partial r}<0$ and consider the representation in the second line of \eqref{heq}. The second term depends on $r$ only through $\tau_M$, its derivative with respect to $r$ is therefore strictly negative. To evaluate the first term, consider $r_1<r_2$. Then, using Assumption \ref{ass:profit}-(iii)
\begin{equation*}
\begin{aligned}
&\int_{\tau_M(r_1)}^\infty \e^{-\bar \rho t} \dot\pi\big(a(r_1+\eta_{\max} t)\big)\dot a(r_1+\eta_{\max} t)  \ud t - \int_{\tau_M(r_2)}^\infty \e^{-\bar \rho t} \dot\pi\big(a(r_2+\eta_{\max} t)\big)\dot a(r_2+\eta_{\max} t) \ud t \\
& = \int_{\tau_M(r_1)}^\infty \e^{-\bar \rho t} \Big(\dot\pi\big(a(r_1+\eta_{\max} t)\big)\dot a(r_1+\eta_{\max} t)-\dot\pi\big(a(r_2+\eta_{\max} t)\big)\dot a(r_2+\eta_{\max} t) \Big) \ud t \\
& - \int_{\tau_M(r_2)}^{\tau_M(r_1)} \e^{-\bar \rho t} \dot\pi\big(a(r_2+\eta_{\max} t)\big)\dot a(r_2+\eta_{\max} t) \ud t\leq 0.
\end{aligned}
\end{equation*}
Thus, $\frac{\partial h}{\partial r}<0$. 
Further, to compute the derivative with respect to $x$, we use integration by parts to rewrite $h$ as follows:
\begin{equation*}
\begin{aligned}
h(r,x) &=-\frac{\eta_{\max}}{|\mu| x + \eta_{\max}}\int_0^{\tau_M}\e^{-\bar\rho t}\frac{\ud}{\ud t} \pi\big(Y_t[x-f_M(t)]\big)\ud t-\int_{\tau_M}^\infty\e^{-\bar \rho t}\frac{\ud}{\ud t}\pi\big(a(r+\eta_{\mathrm{max}}t)\big)\ud t\\
 & = \frac{\eta_{\max}}{|\mu| x + \eta_{\max}}\pi(x) -\frac{\bar \rho\eta_{\max}}{|\mu| x + \eta_{\max}}\int_0^{\tau_M}\e^{-\bar\rho t}\pi\big(Y_t[x-f_M(t)]\big)\ud t\\
 &\quad -\frac{\eta_{\max}}{|\mu| x + \eta_{\max}} \e^{-\bar \rho \tau_M} \pi\big(a(r + \eta_{\max}\tau_M)\big)-\int_{\tau_M}^\infty\e^{-\bar \rho t}\frac{\ud}{\ud t}\pi\big(a(r+\eta_{\mathrm{max}}t)\big)\ud t.
\end{aligned}
\end{equation*}
Once again, we can apply monotone convergence theorem to show that $h$ is differentiable with respect to $x$, with
\begin{equation*}
\begin{aligned}
&\partial_x h(r,x)\\
& = \frac{\eta_{\max}}{|\mu| x \!+\! \eta_{\max}}\dot \pi(x)\!-\!\frac{|\mu|\eta_{\max}}{\big(|\mu| x\! +\! \eta_{\max}\big)^2}\Big(\pi(x)\!-\!\bar \rho\int_0^{\tau_M}\!\!\e^{-\bar\rho t}\pi\big(Y_t[x\!-\!f_M(t)]\big)\ud t\Big)\\
&\quad-\frac{\bar \rho\eta_{\max}}{|\mu| x\! +\! \eta_{\max}}\int_0^{\tau_M}\!\!\e^{-\bar\rho t}\dot\pi\big(Y_t[x\!-\!f_M(t)]\big)Y_t\ud t\\
 &\quad +\frac{\eta_{\max}|\mu|}{\big(|\mu| x + \eta_{\max}\big)^2} \e^{-\bar \rho \tau_M} \pi\big(a(r + \eta_{\max}\tau_M)\big)\\
 &\quad+\e^{-\bar \rho \tau_M} \dot\pi\big(a(r + \eta_{\max}\tau_M)\big)\dot a(r + \eta_{\max}\tau_M)\frac{\partial\tau_M}{\partial x}\frac{|\mu| x}{|\mu| x + \eta_{\max}}.
\end{aligned}
\end{equation*}

To show that the derivative is positive, recall that $\frac{\partial \tau_M}{\partial x}>0$ and consider the representation in the second line of \eqref{heq}. The first term depends on $x$ only through $\tau_M$, its derivative with respect to $x$ is therefore strictly positive. To evaluate the second term, consider $x_1<x_2$. Then,
\begin{equation*}
\begin{aligned}
 & \int_0^{\tau_M(x_2)} \e^{-\bar \rho t} \dot \pi\big(Y_t[x_2-f_M(t)]\big) Y_t \ud t - \int_0^{\tau_M(x_1)} \e^{-\bar \rho t} \dot \pi\big(Y_t[x_1-f_M(t)]\big) Y_t \ud t \\
  & = \int_{\tau_M(x_1)}^{\tau_M(x_2)} \e^{-\bar \rho t} \dot \pi\big(Y_t[x_2-f_M(t)]\big) Y_t \ud t + \int_0^{\tau_M(x_1)} \e^{-\bar \rho t} \Big(\dot \pi\big(Y_t[x_2-f_M(t)]\big) - \dot \pi\big(Y_t[x_1-f_M(t)]\big) \Big)Y_t \ud t,
\end{aligned}
\end{equation*}
where both terms are nonnegative  by our assumptions. Thus $\frac{\partial h}{\partial x}>0$. 

Continuity of $\nabla h$ is easily deduced by the explicit expressions for the derivatives.
\end{proof}
\begin{proposition}\label{prop:optfirmmu}
Let $\mu<0$ and $(R_0,X_0)=(r,x)$. Take $w_0$ and $w_1$ as in \eqref{eq:OSP} and \eqref{eq:w1.b}, respectively. Then, the firm's equilibrium payoff reads as
\begin{equation}\label{eq:weq1}
w(r,x)=\left\{
\begin{array}{ll}
w_0(r,x), & x>b(r),\\
w_1(r,x), & a(r)\le x\le b(r).
\end{array}
\right.
\end{equation}
The firm's optimal strategy reads $\eta^*_t=\eta_{\mathrm{max}}1_{\{t\ge \tau_b\}}$, with $\tau_b=\tau_*$ as in \eqref{eq:tau*b}. The optimal control is unique in the sense that if $\eta\in\mathcal{A}_F$ is another optimal control, then $\eta_t=\eta^*_t$ for a.e.\ $t\ge 0$.
\end{proposition}
\begin{proof}
The proof consists of showing that $w$ solves the HJB equation and then applying the verification theorem (in the form of Corollary \ref{cor:verif} with Remark \ref{rem:verif}).

First, we want to show that $w\in C^1(\overline{\mathcal I})$. The result is clear for $a(r)\le x\le b(r)$ by \eqref{eq:drw1} and \eqref{eq:dxw1} and because $w=w_1$ in that set. Indeed, $w_1$ is everywhere continuously differentiable. 
For $x> b(r)$ 
it is convenient to recall that (cf.\ \eqref{eq:tau*b})
\[
\tau^{r,x}_*=\frac{1}{|\mu|}\ln\frac{x}{b(r)}.
\]
Therefore
\[
\frac{\partial\tau_*^{r,x}}{\partial x}=\frac{1}{|\mu| x}\quad\text{and}\quad\frac{\partial\tau^{r,x}_*}{\partial r}=-\frac{1}{|\mu|}\frac{\dot b(r)}{b(r)},
\]
are both continuous (cf.\ \eqref{eq:bdot}). 
Recall that for $x>b(r)$, $w(r,x)=w_0(r.x)=w_1(r,x)+\varphi(r,x)$ with (cf.\ \eqref{eq:w0w1}) 
\begin{equation}\label{eq:vphi}
\varphi(r,x)=\int_0^{\tau_*^{r,x}}\e^{-\bar\rho t}h\big(r,xY_t\big)\ud t,
\end{equation}
where we also used that $\tau_*^{r,x}\le \tau_{r,x}(0)$.
Then, $\nabla w_0$ is continuous in the set $\{(r,x):x\ge b(r)\}$ if and only if $\nabla \varphi$ is such. 
Differentiating in $x$, using the continuity of the derivatives of $h$ shown above to justify exchanging the limit and the integral, and recalling $h(r,b(r))=0$ yields
\begin{equation*}
\begin{aligned}
\partial_x\varphi(r,x)&=\e^{-\bar\rho\tau_*^{r,x}}h(r,b(r))\frac{\partial\tau_*^{r,x}}{\partial x}+\int_0^{\tau_*^{r,x}}\e^{-\bar\rho t}\partial_x h\big(r,xY_t\big)Y_t\ud t\\
&=\int_0^{\tau_*^{r,x}}\e^{-\bar\rho t}\partial_x h\big(r,xY_t\big)Y_t\ud t.
\end{aligned}
\end{equation*}
Differentiating in $r$ yields
\begin{equation*}
\begin{aligned}
\partial_r\varphi(r,x)&=\e^{-\bar\rho\tau_*^{r,x}}h(r,b(r))\partial_r\tau^{r,x}_*+\int_0^{\tau_*^{r,x}}\e^{-\bar\rho t}\partial_r h(r,xY_t)\ud t\\
&=\int_0^{\tau_*^{r,x}}\e^{-\bar\rho t}\partial_r h(r,xY_t)\ud t.
\end{aligned}
\end{equation*}
We immediately deduce $\nabla \varphi$ continuous in the set $\{(r,x):x\ge b(r)\}$. 
Since $w_0(r,b(r))=w_1(r,b(r))$, it remains to verify that $\nabla w_0(r,b(r))=\nabla w_1(r,b(r))$. However,  $\tau_*^{r,b(r)}=0$ implies $\nabla \varphi(r,b(r))=0$ and the claim follows.

Now we show that $w$ solves the HJB equation. Recall that $h(r,x)<0$ for $a(r)\le x<b(r)$ by Lemma \ref{lem:h}. Comparing to \eqref{eq:HJBmu0} yields  $(\partial_r w_1-\partial_x w_1)(r,x)=-\eta_{\mathrm{max}}^{-1}h(r,x)>0$ for $a(r)\le x<b(r)$ and therefore $\eta_{\mathrm{max}}(\partial_r w_1-\partial_x w_1)(r,x)=\mathcal{H}(r,x;w)$. Moreover, it is clear that $w_x(r,a(r))=\partial_x w_1(r,a(r))=0$ for all $r>0$ by \eqref{eq:dxw1}. Thus, $w$ is a solution of the HJB equation in the set $\{(r,x):a(r)\le x<b(r)\}$, i.e., by Proposition \ref{prop:w1}
\begin{equation}\label{eq:HJBw1b}
\begin{aligned}
\mu x\partial_x w(r,x)-\bar\rho w(r,x)+\mathcal H(r,x;w)+\pi(x)&=0,\quad (r,x): a(r)<x<b(r),\\
w_x(r,a(r))&=0,\quad r\in(0,\infty).
\end{aligned}
\end{equation}
Next, we prove that $(\partial_r w_0-\partial_x w_0)(r,x)< 0$ for $x > b(r)$. Fix $b(r)<x_1<x_2$. Recall that $x\mapsto\tau_{r,x}(0)$ is increasing and $\tau_{r,x_2}(0)\ge\tau_{r,x_1}(0)\ge\tau_*^{r,x_1}$.
Then,
\begin{equation*}
\begin{aligned}
\varphi(r,x_2)&\ge \int_0^{\tau_*^{r,x_1}\wedge\tau_{r,x_2}(0)}\e^{-\bar\rho t}h(r,x_2Y_t)\ud t=\int_0^{\tau_*^{r,x_1}}\e^{-\bar\rho t}h(r,x_2Y_t)\ud t\\
&=\int_0^{\tau_*^{r,x_1}}\e^{-\bar\rho t}\big(h(r,x_2Y_t)-h(r,x_1Y_t)\big)\ud t+\varphi(r,x_1).
\end{aligned}
\end{equation*}
Rearranging terms and dividing by $x_2-x_1$, we let $x_2-x_1\to 0$ and deduce
\begin{equation*}
\begin{aligned}
\partial_x\varphi(r,x)\ge \int_0^{\tau_*^{r,x}}\e^{-\bar\rho t}\partial_x h(r,x Y_t)Y_t\ud t>0, 
\end{aligned}
\end{equation*}
where we used continuity of $x\mapsto \tau^{r,x}_*$ and the strict inequality holds for any $x>b(r)$ because $\tau^{r,x}_*>0$ and $\partial_x h(r,x Y_t)>0$ for $t<\tau^{r,x}_*$ (cf.\ Lemma \ref{lem:h}).

Now fix $r_1<r_2$ and $x>b(r_2)$. Since $r\mapsto \tau_{r,x}(0)$ is decreasing and $\tau^{r_2,x}_*\le \tau_{r_2,x}(0)\le \tau_{r_1,x}(0)$, by analogous arguments as the ones above we get
\begin{equation*}
\begin{aligned}
\varphi(r_2,x)-\varphi(r_1,x)\le\int_0^{\tau^{r_2,x}_*}\e^{-\bar \rho t}\big(h(r_2,xY_t)-h(r_1,xY_t)\big)\ud t. 
\end{aligned}
\end{equation*}
Dividing by $r_2-r_1$ and letting $r_2-r_1\to 0$ we deduce
\begin{equation*}
\partial_r\varphi(r,x)\le\int_0^{\tau^{r,x}_*}\e^{-\bar \rho t}\partial_r h(r,xY_t)\ud t<0, 
\end{equation*}
where we used continuity of $r\mapsto \tau^{r,x}_*$ and the strict inequality holds for $x>b(r)$, because $\tau^{r,x}_*>0$ and $\partial_r h(r,xY_t)<0$ for $t<\tau^{r,x}_*$ (cf.\ Lemma \ref{lem:h}).

Thus, for $(r,x)$ such that $x > b(r)$ we have
\begin{equation}\label{eq:gradineq}
\big(\partial_r w_0-\partial_x w_0\big)(r,x)=\big(\partial_r w_1-\partial_x w_1\big)(r,x)+\big(\partial_r \varphi-\partial_x \varphi\big)(r,x) < 0,
\end{equation}
as needed. The latter implies $\mathcal{H}(r,x;w)=0$ for $x>b(r)$. Since $\tau_*$ from \eqref{eq:tau*b} is optimal and using the fact that the mapping in \eqref{eq:Mp} is constant, we deduce that\footnote{In this case also direct differentiation of the function $w_0$ would easily lead to the same result.} 
\[
\mu x\partial_x w_0(r,x)-\bar\rho w_0(r,x)+\pi(x)=0,\quad \text{for $x\ge b(r)$}.
\]
In conclusion we have shown that $w\in C^1(\overline{\mathcal I})$ solves
\begin{equation*}
\begin{aligned}
\mu x \partial_x w(r,x)-\bar\rho w(r,x)+\mathcal{H}(r,x;w)+\pi(x) &=0,\quad (r,x):x > a(r),\\
w_x(r,a(r))&=0,\quad r\in(0,\infty).
\end{aligned}
\end{equation*}
 
To check the transversality condition \eqref{eq:transv.0} for the function $w$, using \eqref{eq:OSP} and \eqref{eq:w1.b} we get 
$$
|w(r,x)|\le \frac{\pi(x)}{\bar \rho}+w_1(r,x) \leq \frac{2\pi(x)}{\bar \rho}+ \e^{r\bar \rho/\eta_{\max}}\int_{r/\eta_{\max}}^\infty \e^{-\bar \rho t} \pi(a( \eta_{\max} t))\ud t
$$
and then proceed as in the proof of Proposition \ref{prop:optfirmmu0}.

By construction, the optimal investment policy $\nu^*$ satisfies \eqref{eq:SC} (cf.\ Theorem \ref{thm:nustar}) and therefore we can apply our verification theorem (Corollary \ref{cor:verif} and Remark \ref{rem:verif}) to deduce optimality of $\eta^*_t\triangleq \eta_{\mathrm{max}}1_{\{t\ge \tau_*\}}$, in the sense that 
$w(r,x)=\sup_{\eta}\mathcal J^F_{r,x}(\eta,\nu^*)=\mathcal J^F_{r,x}(\eta^*,\nu^*)$, for $(r,x)\in\overline{\mathcal{I}}$.

It remains to prove uniqueness of the optimal control. Fix $r\ge 0$ and $a(r)\le x\le b(r)$, so that for any control $\eta\in\mathcal{A}_F$ we have 
\begin{equation}\label{eq:bdddomain}
a(R^\eta_t)\le X^{\nu^*,\eta^*}_t < b(R^\eta_t),\quad\text{for all $t\in(0,\infty)$},
\end{equation}
as a result of the negative drift $\mu<0$ and of $a(r)<b(r)$ for all $r\in[0,\infty)$. Assume that $\eta\in\mathcal{A}_F$ is optimal but $\eta_t<\eta_{\mathrm{max}}$ for $t\in\mathcal I$ where $\int_0^T 1_{\{t\in\mathcal I\}}\ud t>0$ for some $T>0$. Indeed, with no loss of generality we can assume $\mathcal I\subset[0,T]$. From the first equation in the proof of Theorem \ref{thm:verif} and using \eqref{eq:HJBw1b} we obtain 
\begin{equation}\label{eq:itow}
\begin{aligned}
&\e^{-\bar\rho t}w\big(R^\eta_{t},X^{\nu^*,\eta}_{t}\big)\\
&=w(r,x)+\int_0^{t}\e^{-\bar\rho s}\big[\mu X^{\nu^*,\eta}_s\, w_x(\cdot)-\rho w(\cdot)+\eta_s\big(w_r(\cdot)-w_x(\cdot)\big)\big]\big(R^\eta_s,X^{\nu^*,\eta}_s\big)\ud s\\
& < w(r,x)+\int_0^{t}\e^{-\bar\rho s}\big[\mu X^{\nu^*,\eta}_s\, w_x(\cdot)-\bar\rho w(\cdot)+\mathcal H(\cdot;w)\big]\big(R^\eta_s,X^{\nu^*,\eta}_s\big)\ud s\\
& = w(r,x)-\int_0^{t}\e^{-\bar\rho s}\pi\big(X^{\nu^*,\eta}_s\big)\ud s,
\end{aligned}
\end{equation}
where the strict inequality holds for any $t>T$ because 
\[
\eta_s\big(w_r-w_x\big)\big(R^\eta_s,X^{\nu^*,\eta}_s\big)<\mathcal H\big(R^\eta_s,X^{\nu^*,\eta}_s;w\big),
\]
for all $s\in\mathcal{I}$. The strict inequality is preserved when letting $t\to \infty$ because $\mathcal I\subset[0,T]$. Since $w\ge 0$, rearranging terms we reach a contradiction because 
\[
w(r,x)=w_1(r,x)>\int_0^{\infty}\e^{-\bar\rho s}\pi\big(X^{\nu^*,\eta}_s\big)\ud s=\mathcal{J}^F_{r,x}(\eta,\nu^*),
\]
and $\eta$ cannot be optimal for $a(r)\le x\le b(r)$, $r\in[0,\infty)$. 

Next we consider $x>b(r)$ for some $r\in[0,\infty)$. Assume that $\eta\in\mathcal{A}_F$ is an optimal control and let $\tau_b^\eta:=\inf\{t\ge 0:X^{\nu^*,\eta}_t\le b(R^\eta_t)\}$. From the previous paragraph we know that $\eta_t=\eta_{\mathrm{max}}$ for a.e.\ $t\ge \tau^\eta_b$ and it only remains to show that $\eta_t=0$ for a.e.\ $t\in[0,\tau^\eta_b]$. Assume otherwise and let $\eta_t>0$ for $t\in\mathcal{I}\cap[0,\tau^\eta_b]$ where $\mathcal I$ is a set of positive Lebesgue measure. Repeating the same arguments as in the derivation of \eqref{eq:itow} we come to the same conclusion, because \eqref{eq:gradineq} implies
\[
\eta_s\big(w_r-w_x\big)\big(R^\eta_s,X^{\nu^*,\eta}_s\big)<0=\mathcal H\big(R^\eta_s,X^{\nu^*,\eta}_s;w\big),
\]
for all $s\in\mathcal{I}\cap[0,\tau^\eta_b]$. Then, letting $t\to \infty$ we arrive again at a contradiction, showing that $\eta_t=0$ for a.e.\ $t\in[0,\tau^\eta_b]$.

In conclusion, we have obtained $\eta_t=\eta^*_t$ for a.e.\ $t\in[0,\infty)$, as claimed.
\end{proof}

Thanks to Theorem \ref{thm:nustar} and Proposition \ref{prop:optfirmmu}, the form of the Nash equilibrium follows by a direct application of Corollary \ref{cor:verif} and Remark \ref{rem:verif}. The proof of uniqueness is omitted because it repeats verbatim the argument in the proof of Corollary \ref{cor:NE2} combining uniqueness of $\nu^*$ from Theorem \ref{thm:nustar} and of $\eta^*$ from Proposition \ref{prop:optfirmmu}.
\begin{corollary}\label{optstrat.cor}
Let the controlled dynamics be as in \eqref{eqdet} with $\mu<0$. Then, the unique equilibrium pair is given by
\begin{equation}\label{eq.optimalpair_det}
(\eta^*,\nu^*)=(\eta_{\mathrm{max}}1_{\{t\ge \tau_b\}},\nu^a),
\end{equation}
with $\nu^a$ as in \eqref{eq:nua} and $\tau_b$ as in \eqref{eq:tau*b}.    
\end{corollary}

\subsection{Strong solution of the HJB system}\label{sec:strong}

A posteriori we can show that the investor's equilibrium payoff $v$ solves the variational problem \eqref{eqvariationalv1}. Then, the pair $(v,w)$ solves the HJB system \eqref{eqw1}--\eqref{eqvariationalv1} because we have shown that $w$ solves \eqref{eqw1} in the proof of Proposition \ref{prop:optfirmmu}. However, since the function $v$ may not be smooth, we need a precise notion of solution of the HJB system \eqref{eqw1}--\eqref{eqvariationalv1} (we call it {\em strong solution}). In particular, we will adapt the definition of the set $\mathcal{M}$ in order to account for the fact that $v_x$ is only defined almost everywhere.
\begin{definition}[{\bf Strong solution}]\label{def:strong}
For $\sigma=0$, recall that $\mathcal{L}=\mu x\partial_x$. The pair $(v,w)$ is a {\em strong solution} of the HJB system if:
\begin{itemize}
\item[(i)] $v$ is locally Lipschitz on $(0,\infty)^2$ with $v_x\le \alpha$ a.e.

\item[(ii)] Letting
$\mathcal M\triangleq\{(r,x): v_x(r,x)\ \text{exists and}\ v_x(r,x) = \alpha\}$, $\mathcal{I}\triangleq\mathcal{M}^c$
and $\partial\mathcal{M}=\overline{\mathcal M}\cap\overline{{\mathcal I}}$,
it holds $w\in C^1(\overline{\mathcal I})$ and
\begin{equation*}
\left\{
\begin{array}{ll}
(\mathcal{L} w - \bar\rho w) (r,x) + \mathcal{H}(r,x;w) + \pi(x) =0, &\qquad (r,x)\in\mathcal I,\\[+4pt]
w_x(r,x) = 0, &\qquad (r,x)\in\partial\mathcal M. 
\end{array}
\right.
\end{equation*}
\item[(iii)] Letting $\eta^*(r,x)\triangleq \eta_{\mathrm{max}} 1_{\{w_r>w_x\}}(r,x)$,
\begin{equation*}
\max \left\{ (\mathcal{L} v- \rho v) (r,x) + \left( v_{r}(r,x) - v_{x}(r,x) \right) \eta^*(r,x)  + \Pi(r,x), v_x(r,x) - \alpha \right\} = 0,
\end{equation*}
for a.e.\ $(r,x)\in[0,\infty)\times\mathbb R_+$.
\end{itemize}
\end{definition}
Next we show that $v$ satisfies (i) and (iii) of the above definition. For the statement and proof of the next proposition it is convenient to introduce sets
\begin{equation*}
\mathcal{O}_a\triangleq\{(r,x):0<x<a(r)\},\quad
\mathcal    {O}_{a,b}\triangleq\{(r,x):a(r)< x< b(r)\}\quad\text{and}\quad\mathcal{O}_{b}\triangleq\{(r,x):b(r)<x\}.
\end{equation*}
Moreover, we are going to use the following notation: given a set $S$ that can be partitioned as $S=A\cup B$, we say that a function $\varphi:S\to\mathbb R$ belongs to $C(\overline A)\cap C(\overline B)$ if $\varphi$ is continuous separately on $A$ and $B$ with continuous extensions to the closure of both sets; this allows the function $\varphi$ to be discontinuous across the boundary $\overline A\cap\overline B$, thus $C(S)\subsetneq C(\overline A)\cap C(\overline B)$. 
\begin{proposition}\label{prop:strong}
The investor's equilibrium payoff $v$ satisfies (i) and (iii) in Definition \ref{def:strong}. The set $\mathcal{M}$ is explicitly given by $\mathcal{M}=\{(r,x): 0< x\le a(r)\}$ and its boundary reads $\partial\mathcal M=\{(r,x):x=a(r)\}$. Finally,  $v_x\in C(\overline{\mathcal{O}}_{a})\cap C(\overline{\mathcal{O}}_{a,b})\cap C(\overline{\mathcal{O}}_{b})$ (possibly discontinuous across $r\mapsto b(r)$) and $v_r \in C(\overline{\mathcal{O}}_{a,b})$.
\end{proposition}

\begin{proof}
First of all we obtain an analytical expression for $v$ using that $v(r,x)=\mathcal{J}^I_{r,x}(\eta^*,\nu^*)$. 
From the explicit formulae for $\eta^*$ and $\nu^*$ we obtain the following expressions: for $(r,x)\in\overline{\mathcal{O}}_b$ we have
\begin{equation}\label{eq:v1}
\begin{aligned}
v(r,x)&=\int_0^{\tau^{r,x}_b}\e^{-\rho t}\Pi\big(r,xY_t\big)\ud t+\e^{-\rho\tau^{r,x}_b}\int_0^{\tau^{r,b(r)}_M}\e^{-\rho t}\Pi\big(r+\eta_{\mathrm{max}}t,Y_t[b(r)-f_M(t)]\big)\ud t\\
&\quad +\e^{-\rho\tau^{r,x}_b}\int_{\tau^{r,b(r)}_M}^\infty\e^{-\rho t}\Pi\big(r+\eta_{\mathrm{max}}t,a(r+\eta_{\mathrm{max}}t)\big)\ud t \\
&\quad-\alpha \e^{-\rho\tau^{r,x}_b}\int_{\tau^{r,b(r)}_M}^\infty\e^{-\rho t}\Big(\eta_{\mathrm{max}}\dot a(r+\eta_{\mathrm{max}}t)+\eta_{\mathrm{max}}+|\mu|a(r+\eta_{\mathrm{max}}t)\Big)\ud t;
\end{aligned}
\end{equation}
for $(r,x)\in\overline{\mathcal{O}}_{a,b}$ we have
\begin{equation}\label{eq:v2}
\begin{aligned}
v(r,x)&=\int_0^{\tau^{r,x}_M}\e^{-\rho t}\Pi\big(r+\eta_{\mathrm{max}}t,Y_t[x-f_M(t)]\big)\ud t\\
&\quad +\int_{\tau^{r,x}_M}^\infty\e^{-\rho t}\Pi\big(r+\eta_{\mathrm{max}}t,a(r+\eta_{\mathrm{max}}t)\big)\ud t \\
&\quad-\alpha \int_{\tau^{r,x}_M}^\infty\e^{-\rho t}\Big(\eta_{\mathrm{max}}\dot a(r+\eta_{\mathrm{max}}t)+\eta_{\mathrm{max}}+|\mu|a(r+\eta_{\mathrm{max}}t)\Big)\ud t;
\end{aligned}
\end{equation}
for $(r,x)\in\mathcal{O}_a$ we have
\begin{equation}\label{eq:dstar}
\begin{aligned}
v(r,x)&=-\alpha\big(a(r)-x\big)+
\int_0^\infty\e^{-\rho t}\Pi\big(r+\eta_{\mathrm{max}}t,a(r+\eta_{\mathrm{max}}t)\big)\ud t\\
&\quad-\alpha \int_0^\infty\e^{-\rho t}\Big(\eta_{\mathrm{max}}\dot a(r+\eta_{\mathrm{max}}t)+\eta_{\mathrm{max}}+|\mu|a(r+\eta_{\mathrm{max}}t)\Big)\ud t.
\end{aligned}
\end{equation}
Notice that for any $\tau\ge 0$
\begin{equation*}
\begin{aligned}
\int_\tau^\infty\e^{-\rho t}\dot a(r+\eta_{\mathrm{max}}t)\ud t&=\frac{1}{\eta_{\mathrm{max}}}\int_\tau^\infty\e^{-\rho t}\ud a(r+\eta_{\mathrm{max}} t)\\
&=\frac{1}{\eta_{\mathrm{max}}}\Big(-\e^{-r\tau}a(r+\eta_{\mathrm{max}}\tau)
+\rho\int_\tau^\infty\e^{-\rho t}a(r+\eta_{\mathrm{max}}t)\ud t\Big).
\end{aligned}
\end{equation*}
It is now easy to check that $(r,x)\mapsto v(r,x)$ is locally Lipschitz on $[0,\infty)\times\mathbb{R}_+$ using also the explicit formulae for $\nabla \tau^{r,x}_b$ and $\nabla\tau^{r,x}_M$. 

Since
\[
v(r,x)=\sup_{\nu\in\mathcal{A}_I}\mathcal{J}^I_{r,x}(\eta^*,\nu),
\]
then by dynamic programming arguments for singular control, justified by continuity of $v$ (cf.\ \cite{de2023dynamic}), we know that for any $\nu\in\mathcal{A}_I$
\begin{equation}\label{eq:supermg}
t\mapsto \e^{-\rho t}v\big(R^{\eta^*}_t, X^{\nu,\eta^*}_t\big) +\int_0^t\e^{-\rho s}\Pi(R^{\eta^*}_s,X^{\nu,\eta^*}_s)\ud s-\alpha\int_{[0,t]}\e^{-\rho s}\ud \nu_s
\end{equation}
is a nonincreasing function (a deterministic supermartingale) and 
\begin{equation}\label{eq:mart}
t\mapsto \e^{-\rho t}v\big(R^{\eta^*}_t, X^{\nu^*,\eta^*}_t\big) +\int_0^t\e^{-\rho s}\Pi(R^{\eta^*}_s,X^{\nu^*,\eta^*}_s)\ud s-\alpha\int_{[0,t]}\e^{-\rho s}\ud \nu^*_s
\end{equation}
is constant (a deterministic martingale). In particular, from the supermartingale property we deduce, by choosing $\nu_0=\delta>0$ and $t= 0$, that $v(r,x)\ge v(r,x+\delta)-\alpha\delta$. Hence, 
\begin{equation}\label{eq:ubound}
\partial_x v(r,x)\le \alpha,\quad\text{for a.e.\ $(r,x)\in[0,\infty)\times\mathbb{R}_+$}.
\end{equation}
Then (i) in Definition \ref{def:strong} holds.

Thanks to the Lipschitz regularity of $v$ we can apply a change of variable formula to \eqref{eq:supermg} with $\nu\equiv 0$ because the dynamics $(R^{\eta^*},X^{0,\eta^*})$ is absolutely continuous in time. That yields
\begin{equation*}
\begin{aligned}
v(r,x)&\ge \e^{-\rho t}v\big(R^{\eta^*}_t, X^{0,\eta^*}_t\big) +\int_0^t\e^{-\rho s}\Pi(R^{\eta^*}_s,X^{0,\eta^*}_s)\ud s\\
&=v(r,x)+\int_0^t \e^{-\rho s}\big(\mathcal{L}v+\big(v_r-v_x\big)\eta^*-\rho v+\Pi\big)(R^{\eta^*}_s,X^{0,\eta^*}_s)\ud s.
\end{aligned}
\end{equation*}
Dividing by $t$ and letting $t\downarrow 0$ we deduce
\begin{equation*}
\mathcal{L}v(r,x)+\big(v_r(r,x)-v_x(r,x)\big)\eta^*(r,x)-\rho v(r,x)+\Pi(r,x)\le 0,\quad a.e.\ (r,x)\in[0,\infty)\times\mathbb{R}_+.
\end{equation*}
Since $t\mapsto\nu^*_t$ is also absolutely continuous when $(R_0,X_0)=(r,x)$ is such that $x\ge a(r)$ (cf.\ Remark \ref{rem:nustar}), then an analogous change of variable argument, combined with \eqref{eq:mart} yields
\begin{equation}\label{eq:HJBv}
\mathcal{L}v(r,x)+\big(v_r(r,x)-v_x(r,x)\big)\eta^*(r,x)-\rho v(r,x)+\Pi(r,x)= 0,
\end{equation}
for a.e.\ $(r,x)\in[0,\infty)\times\mathbb{R}_+$ such that $x> a(r)$. Combining the results above we obtain that $v$ satisfies (iii) in Definition \ref{def:strong}. Actually, upon closer inspection of the formulae in \eqref{eq:v1} and \eqref{eq:v2} we notice that $v$ and $v_x$ are continuous separately in the set $\mathcal{O}_{a,b}$ and $\mathcal{O}_b$, with continuous extensions to the boundary of the two domains (we are not claiming continuity of derivatives across the boundaries). Since $\eta^*(r,x)$ is also constant in those two sets, we deduce that $v_r$ is continuous in $\overline{\mathcal{O}}_{a,b}$ and indeed \eqref{eq:HJBv} holds in the classical sense at all points $(r,x)$ with $x>a(r)$ and $x\neq b(r)$.

It remains to show that $\mathcal{M}=\{(r,x):0<x\le a(r)\}$. From the explicit formulae for $v$ it is immediate to deduce $v_x(r,x)=\alpha$ for all $(r,x)$ such that $0<x\le a(r)$. Then we must show that $v_x(r,x)<\alpha$ at all points in $x>a(r)$ where the derivative exists. 

Let us fix $(r,x)\in\mathcal{O}_{a,b}$ and let $\widehat\nu^*=\widehat \nu^{*;r,x}$ be optimal for $v(r,x)$ and given by \eqref{eq:nua2}. Then, from \eqref{eq:barJ} we get
\[
v(r,x)=\int_0^\infty\e^{-\rho t}\Big(\Pi\big(r+\eta_{\mathrm{max}}t,Y_t[x+\widehat\nu^*_t-f_M(t)]\big)-\alpha\delta Y_t\Big)\ud t.
\]
Notice that in this case $\eta^*_t=\eta_{\mathrm{max}}$ for all $t\ge 0$ because the dynamics $(R^{\eta^*},X^{\nu^*,\eta^*})$ is bound to evolve in $\overline{\mathcal{O}}_{a,b}$. Moreover, the firm's optimal control remains the same also when $(R_0,X_0)=(r,x-\varepsilon)$
for any small $\varepsilon>0$, because $X_0\le b(R_0)$ and $\widehat\nu^{*;r,x}$ is admissible but suboptimal for the payoff $\bar{\mathcal{J}}^I_{r,x-\varepsilon}(\nu,\eta^*)$ in \eqref{eq:barJ}. We then have 
\[
v(r,x-\varepsilon)\ge\int_0^\infty\e^{-\rho t}\Big(\Pi\big(r+\eta_{\mathrm{max}}t,Y_t[x-\varepsilon+\widehat\nu^*_t-f_M(t)]\big)-\alpha\delta Y_t\Big)\ud t.
\]
Subtracting the two expressions, yields
\begin{equation*}
\begin{aligned}
&v(r,x)-v(r,x-\varepsilon)\\
&\le \int_0^\infty\e^{-\rho t}\Big(\Pi\big(r+\eta_{\mathrm{max}}t,Y_t[x+\widehat\nu^*_t-f_M(t)]\big)-\Pi\big(r+\eta_{\mathrm{max}}t,Y_t[x-\varepsilon+\widehat\nu^*_t-f_M(t)]\big)\Big)\ud t.
\end{aligned}
\end{equation*}
If $(r,x)$ is a point where $v_x(r,x)$ exists, we can divide by $\varepsilon$ and let $\varepsilon\downarrow 0$ to obtain
\begin{equation*}
\begin{aligned}
v_x(r,x)&\le \int_0^\infty\e^{-\rho t}Y_t\Pi_x\big(r+\eta_{\mathrm{max}}t,Y_t[x+\widehat\nu^*_t-f_M(t)]\big)\ud t\\
&=\int_0^{\tau^{r,x}_M}\e^{-\rho t}Y_t\Pi_x\big(r+\eta_{\mathrm{max}}t,Y_t[x-f_M(t)]\big)\ud t\\
&\quad+\int_{\tau^{r,x}_M}^\infty\e^{-\rho t}Y_t\Pi_x\big(r+\eta_{\mathrm{max}}t,a(r+\eta_{\mathrm{max}}t)\big)\ud t.
\end{aligned}
\end{equation*}
By definition of $a(r)$ we have, for $t\ge \tau^{r,x}_M$ 
\[
Y_t\Pi_x\big(r+\eta_{\mathrm{max}}t,a(r+\eta_{\mathrm{max}}t)\big)=\alpha\delta Y_t
\]
and therefore, using also $\e^{-\rho t}Y_t=\e^{-\delta t}$,
\[
\int_{\tau^{r,x}_M}^\infty\e^{-\rho t}Y_t\Pi_x\big(r+\eta_{\mathrm{max}}t,a(r+\eta_{\mathrm{max}}t)\big)\ud t=\alpha\int_{\tau^{r,x}_M}^\infty\delta\e^{-\delta t}\ud t=\alpha Y_{\tau^{r,x}_M}\e^{-\rho\tau^{r,x}_M}.
\]
Then
\[
v_x(r,x)\le \alpha Y_{\tau^{r,x}_M}\e^{-\rho\tau^{r,x}_M}+\int_0^{\tau^{r,x}_M}\e^{-\rho t}Y_t\Pi_x\big(r+\eta_{\mathrm{max}}t,Y_t[x-f_M(t)]\big)\ud t\triangleq \Gamma(r,x).
\]

Since $a(r)<x< b(r)$, for sufficiently small $\varepsilon>0$ we can repeat analogous arguments to estimate 
\begin{equation*}
\begin{aligned}
&v(r,x+\varepsilon)-v(r,x)\\
&\ge \int_0^\infty\e^{-\rho t}\Big(\Pi\big(r+\eta_{\mathrm{max}}t,Y_t[x+\varepsilon+\widehat\nu^*_t-f_M(t)]\big)-\Pi\big(r+\eta_{\mathrm{max}}t,Y_t[x+\widehat\nu^*_t-f_M(t)]\big)\Big)\ud t,
\end{aligned}
\end{equation*}
where, again, $\widehat\nu^*_t=\widehat\nu^{*;r,x}_t$ is independent of $\varepsilon$. Thus, dividing by $\varepsilon$ and passing to the limit we conclude $v_x(r,x)\ge \Gamma(r,x)$. Combining with the previous bound we get $v_x(r,x)=\Gamma(r,x)$ for $a(r)<x<b(r)$. From this representation we immediately deduce continuity of $v_x$ in the set $\mathcal{O}_{a,b}$. Moreover, $v_x$ is extended continuously to $\overline{\mathcal{O}}_{a,b}$ thus lifting the regularity of $v_x$ from $L^\infty(\mathcal{O}_{a,b})$ to $C(\overline{\mathcal{O}}_{a,b})$. It is also clear that $\Gamma(r,a(r))=\alpha$ because $\tau^{r,a(r)}_M=0$ and therefore $v_x$ is continuous across the boundary $r\mapsto a(r)$. Next we are going to show that $\Gamma(r,x)<\alpha$ in $\mathcal{O}_{a,b}$. 

Taking a derivative in $x$ of $\Gamma$ we obtain
\begin{equation*}
\begin{aligned}
\Gamma_x(r,x)&=\Big(\Pi_x\big(r+\eta_{\mathrm{max}}\tau^{r,x}_M,a(r+\eta_{\mathrm{max}}\tau^{r,x}_M)\big)-\delta\alpha\Big)\e^{-\rho\tau^{r,x}_M}Y_{\tau^{r,x}_M}\frac{\partial\tau^{r,x}_M}{\partial x}\\
&\quad +\int_0^{\tau^{r,x}_M}\e^{-\rho t}(Y_t)^2\Pi_{xx}\big(r+\eta_{\mathrm{max}}t,Y_t[x-f_M(t)]\big)\ud t\\
&=\int_0^{\tau^{r,x}_M}\e^{-\rho t}(Y_t)^2\Pi_{xx}\big(r+\eta_{\mathrm{max}}t,Y_t[x-f_M(t)]\big)\ud t<0,
\end{aligned}
\end{equation*}
where the second equality holds by definition of the boundary $a(r)$ and the strict inequality is by strict concavity of $\Pi(r,\cdot)$. Then, $\Gamma(r,x)<\alpha$ for all $x>a(r)$, which implies $v_x(r,x)<\alpha$ for $a(r)< x\le b(r)$. 

Now we look at $x>b(r)$. It is clear from the form of $v$ in \eqref{eq:v1} that for each $r\in[0,\infty)$, $v(r,\cdot)$ is twice continuously differentiable for $x>b(r)$. Then, for fixed $r\in[0,\infty)$ the HJB equation reads 
\[
\mu x v_x (r,x)-\rho v(r,x)+\Pi(r,x)=0,\quad\text{for all $x>b(r)$.}
\]
However, using that $\Pi(r,\cdot)$ is twice continuously differentiable (cf.\ Assumption \ref{ass:profit}-(i)) we deduce from the equation above that actually $v(r,\cdot)$ is three times continuously differentiable.

Setting $u\triangleq v_x-\alpha$, we differentiate the equation above with respect to $x$ and obtain
\[
\mu x u_x (r,x)-\delta u(r,x)+\big(\Pi_x(r,x)-\alpha\delta)=0,\quad\text{for all $x>b(r)$.}
\]
Let us start by noticing that $\Pi_x(r,x)-\alpha\delta<0$ for $x\ge b(r)>a(r)$ by strict concavity of $\Pi(r,\cdot)$ and the fact that $\Pi_x(r,a(r))-\alpha\delta=0$. We also know from \eqref{eq:ubound} that $u(r,x)\le 0$ for $x>b(r)$. Then, by the maximum principle we deduce $u(r,x)<0$ for $x>b(r)$. This is directly seen by the representation
\[
u(r,x)=\e^{-\delta (t\wedge\tau^{r,x}_b)}u(r,xY_{t\wedge\tau^{r,x}_b})+\int_0^{t\wedge\tau^{r,x}_b}\e^{-\delta s}\big(\Pi_x(r,xY_s)-\alpha\delta\big)\ud s,
\]
where we recall $\tau^{r,x}_b=\inf\{s\ge 0:xY_s\le b(r)\}$. 

Although $v_x(r,\cdot)$ can be extended continuously to $b(r)$ from above and from below, we are unable to establish the relationship between $v_x(r,b(r)-)$ and $v_x(r,b(r)+)$. In particular, it may occur that $v_x(r,\cdot)$ does not exists at $b(r)$. However, if $v_x(r,b(r))$ exists, then it must be strictly smaller than $\alpha$ because $v_x(r,b(r)-)<\alpha$. Hence, $(r,b(r))\notin\mathcal M$. Otherwise $v_x(r,b(r))$ does not exists and  $(r,b(r))\notin\mathcal M$. So in all cases $(r,b(r))\notin\mathcal M$.
Then we have proven that $\mathcal{M}=\{(r,x):0<x\le a(r)\}$ as claimed, which also implies $\partial\mathcal{M}=\{(r,x):x=a(r)\}$. 

Finally, the set $\{(r,x):x=b(r)\}\subset\mathcal{I}$ is of zero measure and it can be neglected in the variational inequality for $v$. 
\end{proof}
Combining the above proposition with the fact that $w$ satisfies (ii) in Definition \ref{def:strong} we deduce the following corollary (cf.\ the proof of Proposition \ref{prop:optfirmmu} and notice that the definition of $\mathcal M$ in that proof is given by \eqref{eq:MM}, which turns out to agree with the result in Proposition \ref{prop:strong}).
\begin{corollary}\label{cor:strong}
The pair of equilibrium payoffs $(v,w)$ is a strong solution of the HJB system.
\end{corollary}

\section{An algorithm for the construction of an equilibrium in the general case }\label{sec:equilibrium}

In view of Theorem \ref{thm:verif}, finding an equilibrium in our model boils down to finding a solution of \eqref{eqw1}--\eqref{eqvariationalv1}. In full generality we are not able to obtain an analytical solution to the problem.
Therefore, we proceed by developing a numerical method that combines finite differences for both \eqref{eqw1} and \eqref{eqvariationalv1} with a penalization method that reduces the nonlinear problem in \eqref{eqvariationalv1} to an easier semilinear one. 

The penalization method follows a well-trodden path in PDE theory (cf., e.g., \cite[Chapter 9]{Quarteroni2008}) which approximates \eqref{eqvariationalv1} by relaxing the hard constraint $v_x\le \alpha$ into a soft constraint. More precisely, given (small) $\epsilon > 0$, we want to find $v^{\epsilon}$ that satisfies 
\begin{equation}\label{eqvepsilon}
\mathcal G^\rho[v^{\epsilon},\eta^{*}](r,x) = - \Pi(r,x) - \chi^\epsilon[v^\epsilon](r,x),
\end{equation}
with $\mathcal G^\rho[\cdot,\cdot](r,x)$ defined in \eqref{eqG}, and
\begin{equation}\label{eqChi}
    \chi^{\epsilon}[v^\epsilon](r,x) \triangleq \frac1\epsilon(v_{x}^{\epsilon}(r,x) - \alpha)^+.
\end{equation}
Under suitable assumptions, it is often possible to show that as $\epsilon\to 0$ the solution $v^\epsilon$ of the penalized problem converges to a solution of the original problem \eqref{eqvariationalv1}. In our case, the proof appears very complicated due to the (expected) low regularity of the function $\eta^*$ and, more in general, due to the coupling between \eqref{eqvepsilon} and \eqref{eqw1}. However, we observe such convergence numerically.

Since the simultaneous solution of \eqref{eqw1} and \eqref{eqvariationalv1} (or \eqref{eqvepsilon}) requires knowledge of the function $\eta^*$ and of the sets $\mathcal M$ and $\mathcal I$, we need to argue in a sort of iterative way (with the number of iterations denoted by $\ell$).
We initialize our algorithm by taking $\eta^*\equiv 0$ in \eqref{eqvariationalv1}, and $\ell = 0$.
It is shown below that the resulting variational inequality admits an explicit solution, which we denote by $\hat v(r,x)$. The sets $\hat{\mathcal{M}}\triangleq\{\hat v_x=\alpha\}$ and $\hat {\mathcal I}\triangleq(\hat{\mathcal{M}})^c$ can be calculated explicitly (cf.\ \eqref{eqv0}) with $\hat {\mathcal I}=\{(r,x):x>\hat a(r)\}$ and the function $r\mapsto \hat a(r)$ is found in \eqref{eqva0}. In this iteration, the boundary of $\hat{\mathcal M}$ is given by $\partial\hat{\mathcal M}=\{(r,\hat a(r)),r\ge 0\}$,
and it can be used to solve the zero-order iteration of the problem for the firm.

The next step is to calculate the solution $w^{(\ell)}=w^{(0)}$ of \eqref{eqw1} with $\hat {\mathcal M}$ and $\hat {\mathcal I}$ instead of ${\mathcal M}$ and $\mathcal I$. This is done by finite-difference scheme as detailed in \eqref{eqdiscretizations} of Section \ref{sec:mainresults}. Once we have obtained the function $w^{(0)}$ we can define the function $\eta^{*(\ell)}=\eta^{*(0)}$ as a {\em proxy} for the firm's optimal control: 
\begin{equation}\label{eqetavarphi}
    \eta^{*(0)}(r,x)= \eta_{\mathrm{max}} 1_{\{\partial_ r w^{(0)}>\partial_x w^{(0)}\}}(r,x).
\end{equation}
That concludes the initialization of the algorithm. 

Next, for the first iteration of our scheme, we set $\ell=1$ and we approximate \eqref{eqvariationalv1} by \eqref{eqvepsilon}. 
Then, we want to find $v^{\epsilon (\ell)}=v^{\epsilon (1)}$ that satisfies \eqref{eqvepsilon} in the form
\begin{equation}\label{eqvepsilonphi}
\mathcal G^\rho[v^{\epsilon(\ell)},\eta^{*(\ell-1)}](r,x) = - \Pi(r,x) - \chi^\epsilon[v^{\epsilon(\ell)}](r,x),
\end{equation}
for $(r,x)\in[0,\infty)\times(0,\infty)$ with boundary conditions $v^{\epsilon (\ell)}(0,x) = 0$ and $v^{\epsilon(\ell)}(r,0) = 0$. The boundary conditions are motivated by the form of the investor's equilibrium payoff in Section \ref{sec:deterministic}.
The solution of \eqref{eqvepsilonphi} is obtained again by finite differences as described in Section \ref{sec:mainresults}.
Once we have obtained a solution $v^{\epsilon(\ell)}$ of \eqref{eqvepsilonphi} we can determine numerically the sets 
$\mathcal{M}^{(\ell)}_\epsilon\triangleq \{v^{\epsilon(\ell)}_x=\alpha\}$ and $\mathcal{I}_\epsilon^{(\ell)}\triangleq ([0,\infty)\times\mathbb{R}_+)\setminus\mathcal{M}^{(\ell)}_\epsilon$. It turns out that 
\begin{equation*}
\mathcal{I}_\epsilon^{(\ell)}=\big\{(r,x): x> a^{\epsilon(\ell)}(r)\big\},
\end{equation*}
where $r\mapsto a^{\epsilon(\ell)}(r)$ is a continuous function on $[0,\infty)$. 
In order to conclude the first iteration we calculate a solution $w^{(\ell)}=w^{(1)}$ of \eqref{eqw1} with $\mathcal{M}_\epsilon^{(\ell)}$, $\mathcal{I}_\epsilon^{(\ell)}$ instead of $\mathcal M$, $\mathcal I$. That also yields a new proxy for the firm's optimal control:
\begin{equation}\label{eq:etaphi}
    \eta^{*(\ell)}(r,x)= \eta_{\mathrm{max}} 1_{\{\partial_ r w^{(\ell)}>\partial_x w^{(\ell)}\}}(r,x).
\end{equation}
We must notice that also $w^{(\ell)}$ and $\eta^{*(\ell)}$ depend on $\epsilon$ via the sets $\mathcal M^{(\ell)}_\epsilon$ and $\mathcal I^{(\ell)}_\epsilon$. However, we suppress such dependence in our notation for ease of exposition.

The procedure continues as follows: Given $w^{(\ell)}$, $\eta^{*(\ell)}$, $v^{\epsilon(\ell)}$, $\mathcal M^{(\ell)}_\epsilon$, $\mathcal I^{(\ell)}_\epsilon$ we find $v^{\epsilon(\ell+1)}$ by solving \eqref{eqvepsilonphi} and then we determine the sets $\mathcal M^{(\ell+1)}_\epsilon$, $\mathcal I^{(\ell+1)}_\epsilon$ with boundary $r\mapsto a^{\epsilon(\ell+1)}(r)$; subsequently we find $w^{(\ell+1)}$ by solving \eqref{eqw1} with $\mathcal M^{(\ell+1)}_\epsilon$, $\mathcal I^{(\ell+1)}_\epsilon$ instead of $\mathcal M$, $\mathcal I$ and we obtain $\eta^{*(\ell+1)}$ as in \eqref{eq:etaphi}.
This iteration continues until a stopping criteria prescribed in Algorithm \ref{algo} (step 9, Section \ref{sec:mainresults}) is reached.

\begin{remark}
The regularization parameter $\epsilon$ plays a crucial role in finding an approximation for the solution of \eqref{eqvariationalv1}, enabling us to make the problem more amenable to numerical techniques.
As $\ell\to \infty$ we observe numerically that $w^{(\ell)}$, $\eta^{*(\ell)}$, $v^{\epsilon(\ell)}$, $a^{\epsilon(\ell)}$ converge to limits that we denote $w^{\epsilon}$, $\eta^{*\epsilon}$, $v^{\epsilon}$, $a^{\epsilon}$. Then, letting $\epsilon$ go to zero, we also observe numerically that the functions $w^{\epsilon}$, $\eta^{*\epsilon}$, $v^{\epsilon}$, $a^{\epsilon}$ have a well-defined limit, which we denote $w$, $\eta^{*}$, $v$, $a$. In practice, in our numerical implementation, we fix a small $\epsilon$ and take the resulting solutions of the iterative procedure described above as our proxy for the true solution of the system \eqref{eqw1}--\eqref{eqvariationalv1}.
\end{remark}

\subsection{The investor's problem in isolation\label{sec:invest_isolation}}

From now on we work under the assumption: 
\begin{assumption}\label{ass:profit2}
The profit functions $\Pi$ and $\pi$ are given by \eqref{fucfu}.
\end{assumption}

In order to initialize our algorithm we need to start by considering an investor who acts in isolation, i.e., with no emission reduction ever performed by the firm. The investor's expected payoff then reads
\[
\mathcal{J}^I_{r,x}(\nu)\triangleq \mathbb{E}_{r,x} \left[ \int_0^\infty \e^{- \rho t} \Pi(r,X_t^{\nu,0})\ud t- \alpha \int_0^\infty \e^{- \rho t} \ud \nu_t \right],
\]
and the corresponding value function reads
\[
\hat v(r,x)=\sup_{\nu\in\mathcal A^I}\mathcal{J}^I_{r,x}(\nu).
\]
Setting $\hat{\mathcal{M}}=\{\hat{v}_x=\alpha\}$, the analytical expression of $\hat v$ can be determined by the direct solution of 
\begin{equation}\label{eqv1}
\left\{
\begin{array}{ll}
(\mathcal{L} \hat v - \rho \hat v) (r,x) = - \Pi(r,x) , & (r,x)\in \hat{\mathcal{I}}=(\hat{\mathcal M})^c,\\
\hat v_x(r,x) = \alpha, & (r,x)\in\partial \hat{\mathcal{M}}, \\
\hat v_{xx}(r,x) = 0, & (r,x)\in\partial \hat{\mathcal{M}},
\end{array}
\right.
\end{equation}
with $|\hat v(r,x)|\le c(r)(1+x)$ for some $c(r)>0$ and using the ansatz $\partial \hat{\mathcal M}=\{(r,\hat a(r)),r\ge 0\}$, for some $\hat a(r)$ to be determined by imposing the third condition in the system above.  

Under Assumption \ref{ass:profit2},
lengthy by straightforward calculations yield:
\begin{equation}\label{eqv0}
	\hat v(r,x) = B(r)x^{-m} + \lambda x^\beta r^\gamma,
\end{equation}
with constants
\begin{equation}\label{eq:parameters}
\begin{aligned}
\lambda &= \frac{1}{\sigma^2/2(m+\beta)(n-\beta)}, \\
	m &= \frac{\mu-\sigma^2/2 + \sqrt{(\mu-\sigma^2/2)^2 + 2\sigma^2 \rho} }{\sigma^2}, \\
	n &= \frac{-(\mu-\sigma^2/2) + \sqrt{(\mu-\sigma^2/2)^2 + 2\sigma^2 \rho} }{\sigma^2},
\end{aligned}
\end{equation}
and where 
\[
B(r) = \frac{\kappa \lambda(1-\beta)\beta }{m(m+1)} r^{\frac{\gamma(m+1)}{1-\beta}}
\]
with 
\[
\kappa \triangleq \left(\frac{\lambda \beta}{\alpha} \left( \frac{m-\beta+2}{m+1} \right) \right)^{\frac{\beta+m}{1-\beta}}.
\]
Additionally,
\begin{equation}\label{eqva0}
	\hat a(r) = \kappa^{\frac{1}{\beta+m}} r^{\frac{\gamma}{1-\beta}}.
\end{equation}
The optimal investment in this setting is given by
\begin{equation}\label{eq:hatnu}
\hat\nu_t=\int_0^t X^0_s\ud \hat \lambda_s,
\end{equation}
where $\hat \lambda_t=\sup_{0\le s\le t}\Big(\hat a(r)/X^0_s-x\Big)^+$, $t\ge 0$.

\begin{remark}\label{rem:euil2}
It is worth noticing that if $r=0$, then the investor never invests in this setup. That provides us with a boundary condition for the firm's value function. Indeed, for $\nu\equiv 0$ we have $w(0,x) =x/(\bar\rho - \mu)$,
whenever $\bar\rho - \mu > 0$.
\end{remark}

\begin{remark}\label{rem:equilinfty}
When $x\uparrow \infty$ the firm is not going to mitigate its emissions because $\lim_{x\to\infty}\mathcal J^F_{r,x}(\eta,\nu)=\infty$ for any pair $(\eta,\nu)\in\mathcal A_F\times\mathcal A_I$ and all $r\in [0,\infty)$. 
For $\eta\equiv 0$ the investor is again faced with a problem with value $\hat v$ and optimal boundary $\hat a$. Based on this heuristics we postulate that for large values of $x$ the firm's payoff should be given by $\mathcal J^F_{r,x}(0,\hat \nu)$, where $\hat \nu$ is given in \eqref{eq:hatnu}, whereas the investor's payoff is again $\hat v$. Analogous calculations to the ones above yield
\begin{equation*}
\mathcal J^F_{r,x}(0,\hat \nu)=C(r)x^{-m}+\lambda x
\end{equation*}
with $m$ and $\lambda$ as in \eqref{eq:parameters} with $\beta=1$, and
\[
C(r)\triangleq \frac{\lambda}{m}\hat a(r)^{m+1},
\]
where $\hat a$ is given in \eqref{eqva0}.
\end{remark}

\subsection{A numerical scheme \label{sec:mainresults}}

Our approach to solve the problem described in Sections \ref{sec:model}--\ref{sec:equilibrium} will rely on the algorithm explained below.
We will employ a finite-difference scheme to solve both \eqref{eqw1} and \eqref{eqvepsilon}.
More precisely, we adopt the first-order backward difference for first-order derivatives with respect to $r$, followed by fourth-order central discretizations for the first and second-order derivatives with respect to $x$ (cf., e.g., \cite[Chapter 2]{LeVeque20007}).

Given a sufficiently smooth function $\varphi$, let $\varphi_{i,j}=\varphi (r_i,x_j)$ at points on uniform grid partitions $\{ r_0,\ldots,r_M \}$ and $\{ x_0,\ldots,x_N \}$ of $[0,\infty)$ with $x_0=r_0=0$ and large but fixed $r_N$ and $x_N$. We approximate first and second order derivatives as
\begin{equation}\label{eqdiscretizations}
    \begin{aligned}
&\varphi_r(r_i,x_j)  \approx \frac{{\varphi_{i,j} - \varphi_{i-1,j}}}{\Delta_r}, \\
&\varphi_x(r_i,x_j)  \approx \frac{\varphi_{i,j-2} - 8 \varphi_{i,j-1} + 8 \varphi_{i,j+1}- \varphi_{i,j+2}}{12 \Delta_x},\\
& \varphi_{xx}(r_i,x_j) \approx \frac{-\varphi_{i,j-2} + 16\varphi_{i,j-1} -30 \varphi_{i,j}+16\varphi_{i,j+1}-\varphi_{i,j+2}}{12 \Delta_x^2},
    \end{aligned}
\end{equation}
for $i \in \{1,\dots, M\}$ and $j \in \{2,\dots, N-2\}$, with $\Delta_r = r_i - r_{i-1}$ and $\Delta_x = x_j - x_{j-1}$, for any pair $(i,j)$. 
We assume the following initial conditions: 
\begin{equation*}
\left\{
\begin{array}{l}
v^\epsilon(r_0,x_j) = 0, \\
v^\epsilon(r_i,x_0) = 0,\\
w(r_0,x_j)=(\bar\rho-\mu)^{-1} x_j, \\
w(r_i,x_0) = 0,\\
\end{array}
\right.
\end{equation*}
for every $i \in \{ 1,\dots, M\}$ and $j \in \{ 1,\dots, N \}$. The conditions at $r_0=0$ are in keeping with Remark \ref{rem:euil2}. For the condition at $x_0=0$ we should notice that the geometric Brownian motion cannot start from zero and we intuitively assign zero value to a firm with zero profitability. However, this condition is somewhat superfluous because, already starting from the first iteration of our algorithm, the controlled dynamics for $X^{\nu,\eta}$ is not allowed to visit $x=0$.

Although in principle our problem is set on $[0,\infty)^2$, in practice we must select (large) maximum elements $r_N$ and $x_N$ of our state space in order to compute the solution. However, this requires us to specify at least one more boundary condition for the PDEs at either points $(r_N,x_j)$ or $(r_i,x_N)$ for $(i,j)$. We choose to specify the values of $w(r_i,x_N)$ and $v^\epsilon(r_i,x_N)$, for which we have natural candidates, thanks to Remark \ref{rem:equilinfty}. Indeed, we assume
\begin{equation}\label{eqv2}
\left\{
\begin{array}{l}
v^\epsilon(r_i,x_N) = \hat v(r_i,x_N), \\
w(r_i,x_N)= C(r_i)x_N^{-m}+\lambda x_N, 
\end{array}
\right.
\end{equation}
for every $i \in \{ 1,\dots, M\}$.

Next, Algorithm \ref{algo} describes our strategy to derive optimal numerical solutions for $w(r,x)$, $v^\epsilon(r,x)$, $\eta^{\ast}(r,x)$, and $a^\epsilon(r)$. All PDEs in the algorithm are solved using finite-difference scheme with the approximation of derivatives as described above.

\begin{algorithm}\label{algo}
Given $\mu \in \mathbb R$, $\sigma \in \mathbb R_+$, $\rho,\bar\rho,\alpha,\epsilon > 0$, $\beta, \gamma \in (0,1)$, a threshold $\eta_{\mathrm{max}}>0$ and desired precision levels $\varpi$ and $\varpi'$, consider the following steps:
\begin{enumerate}
\item Compute $\hat v(r,x)$ and $\hat a(r)$ via \eqref{eqv0} and \eqref{eqva0}, respectively. Store $a^{\epsilon(0)}(r) \leftarrow \hat a(r)$, $v^{\epsilon(0)}(r,x) \leftarrow \hat v(r,x)$ and compute $\chi^\epsilon[v^{\epsilon(0)}](r,x)$ as in \eqref{eqChi}, for all $(r,x)$.
 
\item Solve numerically \eqref{eqw1} with $\hat{\mathcal{M}}$ and $\hat{\mathcal{I}}$ instead of $\mathcal{M}$ and $\mathcal I$ and with $\eta^\ast \equiv 0$. Denote the solution  $w^{(0)}(r,x)$ and then compute $\eta^{{\ast}(0)}(r,x)$ via \eqref{eqetavarphi}, for all $(r,x)$. 
	
\item Increase $\ell \leftarrow \ell +1$. \label{volta}
 
 \item From $\eta^{\ast(\ell -1)}$ solve \eqref{eqvepsilon} and find $v^{\epsilon(\ell)}(r,x)$, for all $(r,x)$ as follows: for $k\ge 1$, solve 
 \[
\mathcal G^\rho[v^{\epsilon(\ell)}_k, \eta^{*(\ell-1)}](r,x)=-\Pi(r,x)-\chi^\epsilon[v^{\epsilon(\ell)}_{k-1}](r,x),
\]
with $v^{\epsilon(\ell)}_{0}=v^{\epsilon(\ell-1)}$. Iterate until $\|v^{\epsilon(\ell)}_{k}-v^{\epsilon(\ell)}_{k-1}\|\le \varpi$.

\item Store $v^{\epsilon(\ell)}(r,x)\leftarrow v^{\epsilon(\ell)}_k(r,x)$.

 \item Compute $a^{\epsilon(\ell)}(r)$ as 
    \begin{equation}\label{eqaepsilon}
    a^{\epsilon(\ell)} (r) =\partial \left\lbrace  (r,x): v^{\epsilon(\ell)}_x (r,x) = \alpha \right\rbrace.
    \end{equation}

    \item From $a^{\epsilon(\ell)}(r)$, solve \eqref{eqw1} and find $w^{(\ell)}(r,x)$, for all $(r,x)$, by considering the following new iteration: for  fixed $\ell$ and each $k\ge 1$, solve
    \begin{equation*}
\left\{
\begin{array}{ll}
(\mathcal{L} w^{(\ell)}_k - \bar\rho w^{(\ell)}_k) (r,x)+ \mathcal{P}(w^{(\ell)}_k,\eta^{(\ell)}_{k-1})(r,x) = - \pi(x), &(r,x):x> a^{\epsilon(\ell)}(r),\\
\partial_x w^{(\ell)}_k (r,x)= 0, &(r,x): x=a^{\epsilon(\ell)}(r), \\
| w^{(\ell)}_k(r,x) | \leq c(1+x), & (r,x) \in [0,\infty)\times\mathbb{R}_+,
\end{array}
\right.
\end{equation*}
where $\eta^{(\ell)}_0=0$, $\mathcal{P}(\varphi,\eta) \triangleq   \big(\varphi_r - \varphi_x\big)\eta$, and $\eta^{(\ell)}_k(r,x)= \eta_{\mathrm{max}} 1_{\{\partial_ r w^{(\ell)}_k>\partial_x w^{(\ell)}_k\}}(r,x)$, until $\lVert w^{(\ell)}_k(r,x) - w^{(\ell)}_{k-1}(r,x) \lVert \leq \varpi$.
    
\item Store $w^{(\ell)}(r,x) \leftarrow w^{(\ell)}_k(r,x)$ and compute $\eta^{{\ast}(\ell)}_k(r,x)$ via \eqref{eq:etaphi}. \label{volta2}

\item Check: If
 \begin{equation*}
   \max(\lVert w^{(\ell)}(r,x) - w^{(\ell-1)}(r,x) \lVert, \\ \lVert v^{\epsilon(\ell)}(r,x) - v^{\epsilon(\ell-1)}(r,x) \lVert) > \varpi', 
 \end{equation*}
 go back to step \ref{volta}.
 
	Otherwise, proceed to the next step.
 
	\item Return $w(r,x) \equiv w^{(\ell)}(r,x)$, $v^\epsilon(r,x) \equiv v^{\epsilon(\ell)}(r,x)$, $\eta^\ast(r,x) \equiv \eta^{\ast (\ell)}(r,x)$ and $a(r) \equiv a^{\epsilon(\ell)}(r)$, for all $(r,x)$.
\end{enumerate}
\end{algorithm}

In summary, Algorithm \ref{algo} has been designed to approximate optimized solutions for both equilibrium payoffs $w(r,x)$ and $v^\epsilon(r,x)$, along with the firm's optimal strategy $\eta^{\ast}$ and the boundary function $a^\epsilon(r)$ that triggers investor's actions.
The algorithm initiates with the assumption that the firm takes no initial action to reduce pollution ($\eta^\ast = 0$).
It then proceeds by calculating the investor's response given by $\hat a(r)$ and $\hat v(r,x)$ (step 1), together with $w^{(0)}(r,x)$ and $\eta^{\ast(0)}$ via step 2.
Then, for each $\ell \geq 1$, step 4 computes $v^{\epsilon(\ell)}(r,x)$ via a sub loop with iterations $k \geq 1$; step 6 constructs $a^{\epsilon(\ell)}(r)$ and step 7 obtains $w^{(\ell)}(r,x)$ by implementing a sub loop with iterations $k \geq 0$.
It's important to note that the $k$-dependent sub-loops contributing to the construction of both $w(r,x)$ and $v^\epsilon(r,x)$ are run independently. 
The iterative $\ell$-dependent loop continues until the predefined stopping criterion at step 9 is achieved, refining our variables of interest, and ultimately converging towards the optimal equilibrium solutions of our problem.

\begin{remark}
Algorithm \ref{algo} draws inspiration from Howard's algorithm (or policy iteration), widely used in dynamic programming and optimization.
Seminal works on this methodology are attributed to Bellman and can be found in \cite{Bellman1955,Bellman1957}.
Howard extended Bellman’s approach to stationary infinite-horizon Markovian dynamic programming problems in \cite{Howard1960}.
Howard's algorithm is celebrated for its effectiveness in solving sequential decision-making problems and has been widely applied in diverse fields such as economics, engineering, and finance.
Our approach incorporates the core principles of Howard's algorithm while tailoring them to the specific requirements and features of our problem.
\end{remark}

\section{Numerical Results}\label{sec:numerics}

In this section, we perform a detailed numerical analysis of the equilibria discussed in the previous sections {\bf under Assumption \ref{ass:profit2}}. We first look at the form of the optimal strategies and of the equilibrium payoffs in the deterministic setting from Section \ref{sec:deterministic}, i.e., $\sigma=0$, with decreasing revenues $\mu\le 0$. 
Then we will implement the algorithm described in Section \ref{sec:equilibrium} in order 
to derive equilibrium payoffs and optimal strategies in the full stochastic problem.
Numerical results are obtained with MATLAB (R2022b). Section \ref{simul:det} addresses the deterministic problem and Section \ref{simul:sthoc} the stochastic one.

Our choice of parameters in the numerical examples below aims primarily to illustrate interesting behaviors arising in our model. Full empirical calibration of our model is out of scope of this theoretical paper, however, in Appendix \ref{appendix.sec} we provide practical
guidance on how each parameter could be mapped to observable quantities. In all the numerical examples, the cost of investment in \eqref{eqvaluev} is set to $\alpha=1$ and the firm's maximum investment rate is set to $\eta_{\mathrm{max}}=1$, unless otherwise stated. For the solution of  \eqref{eqvepsilonphi} we set the regularization parameter to $\epsilon=10^{-4}$.

Unless otherwise specified, in the fully stochastic case the values of $\mu$ and $\sigma$ are borrowed from \cite{Reddy2016} and they are equal to $0.0741$ and $0.3703$, respectively. In \cite{Reddy2016} the authors study the profit dynamics of an Australian company in the Metals and Mining sector. 
Finally, we set the precision levels required for the numerical algorithm from Section \ref{sec:equilibrium} to $\varpi=\varpi'=10^{-3}$.

\subsection{Deterministic setting\label{simul:det}}

In Section \ref{sec:detmunegative}, we have presented the solution for the deterministic setting with $\mu<0$. This solution includes an explicit formula for the boundary $a(r)$ given by \eqref{eq:nua2}, and $b(r)$ specified in Proposition \ref{Proposition_funcb}.
Note that the construction of $b(r)$ depends on the solution of two coupled non-linear equations: $h(r,b(r))=0$ and \eqref{eq:tauM2}, where $h(\cdot,\cdot)$ is described by \eqref{eq:h}.

In order to construct numerical solutions of $b(r)$, we have implemented the implicit Euler method with Newton-Raphson method as follows.
Consider a finite partition $\{r_0, \dots, r_M \}$ of $r \in [0,\infty)$. From \eqref{eq:bdot} we can see that
\begin{equation}\label{eq:b_numerical}
b_{i+1}= b_{i} + \Delta_r\cdot g(r_{i+1},b_{i+1}),   
\end{equation}
for $i \in \{0,1, \dots, M -1\}$, where $\Delta_r \triangleq r_{i+1}-r_{i}$ and the function $g(\cdot)$ is the right-hand side of \eqref{eq:bdot} with explicit expression obtained using the formulae for $\partial_r h$ and $\partial_x h$ from the proof of Lemma \ref{lem:h}. Notice that derivatives of $\tau_M$ appearing in $\partial_r h$ and $\partial_x h$ are explicit thanks to \eqref{eq:dtMdx} and \eqref{eq:tauMdr}, whereas $\tau_M$ is calculated from \eqref{eq:tauM}.

By fixing an $i \in \{0,1, \dots, M-1 \}$, let $\tilde{b}=b_{i+1}$ in \eqref{eq:b_numerical}. We want to find $\tilde{b}$ that solves $\tilde{b} - b_{i} - \Delta_r\cdot g (r_{i+1}, \tilde{b}) = 0$, which is equivalent to finding the zero of a function
\begin{equation}\label{eq.sbtilde}
s(\tilde{b}) \triangleq \tilde{b} - b_{i} - \Delta_r \cdot g (r_{i+1}, \tilde{b}).  \end{equation}
To solve \eqref{eq.sbtilde}, it turns out that the Newton's iteration is given by (see \cite[Chapter 8]{GriffithsandHigham})
\begin{equation}\label{eq:iteNewton}
    \tilde{b}_{k+1} = \tilde{b}_{k} - \frac{s(\tilde{b}_k)}{\dot{s}(\tilde{b}_k)}, \quad \text{with} \quad \dot{s}(\tilde{b}_k) = 1 - \Delta_r\cdot \frac{\partial}{\partial b} g (r_{i+1}, \tilde{b}_k).
\end{equation}
We iterate \eqref{eq:iteNewton} until $\lVert \tilde{b}_{k+1} - \tilde{b}_{k} \lVert < \tilde{\varpi}$, where $\tilde{\varpi}=10^{-3}$ is a prescribed precision level. Hence, $b_{i+1}$ is set to be equal to the resulting  $\tilde{b}$.
This procedure is repeated for all $i \in \{0,1, \dots, M-1 \}$, obtaining an approximation of $b(r)$ on $\{r_0, \dots, r_M \}$.
In each $k$-subloop described in \eqref{eq:iteNewton}, the values of $\tau_M(r_{i+1}, \tilde{b}_k)$ for every fixed $i\in\{0,1, \dots, M-1 \}$ are obtained by solving \eqref{eq:tauM2} with the MATLAB function FZERO.
\begin{figure}[htbp]
\centering  
\begin{subfigure}{.48\textwidth}
  \centering
  \includegraphics[width=\linewidth]{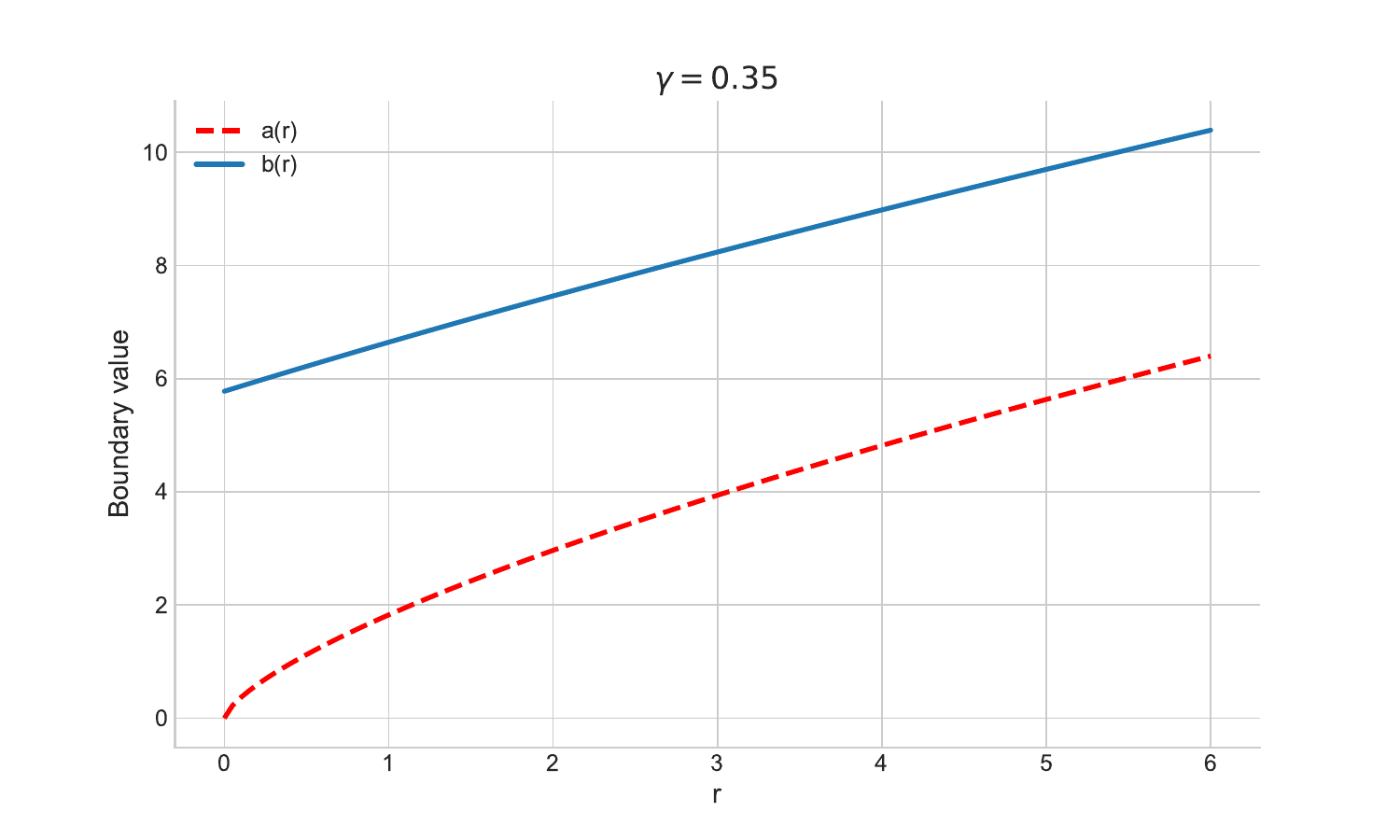}
  \caption{}
  \label{figd001}
\end{subfigure}\hfill
\begin{subfigure}{.48\textwidth}
  \centering
  \includegraphics[width=\linewidth]{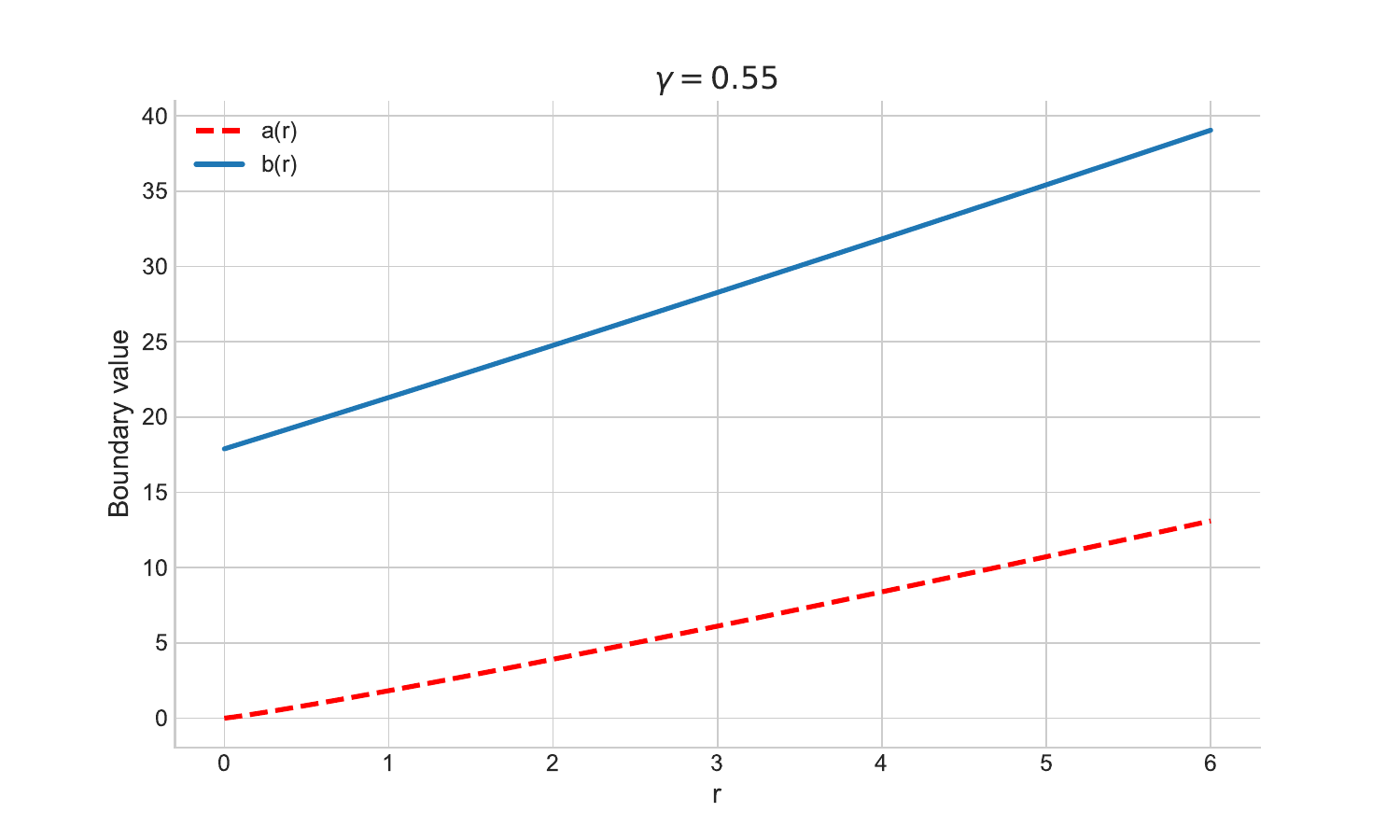}
  \caption{}
  \label{figd002}
\end{subfigure}
\caption{Comparison between functions $a(r)$ and $b(r)$ for two different values of the parameter $\gamma$ (which determines the sensitivity of investor to the environmental impact of the company).  The parameters are: $\rho = \bar \rho = 0.3$, $\mu=-0.0741$, $\beta=0.5$, $\gamma=0.35$ (Figure \ref{figd001}) and $\gamma=0.55$ (Figure \ref{figd002}). The function $a(r)$ is given by \eqref{eq:nua2}, while the boundary $b(r)$ is described in Proposition \ref{Proposition_funcb}. For higher value of $\gamma$, the investors are more concerned about the environmental impact of the company, and both abatement and investment occur earlier (for higher values of the environmental performance.)}\label{figd0}
\end{figure}

\begin{figure}[htbp]
\centering
\begin{subfigure}{0.48\textwidth}
  \centering
  \includegraphics[width=\linewidth]{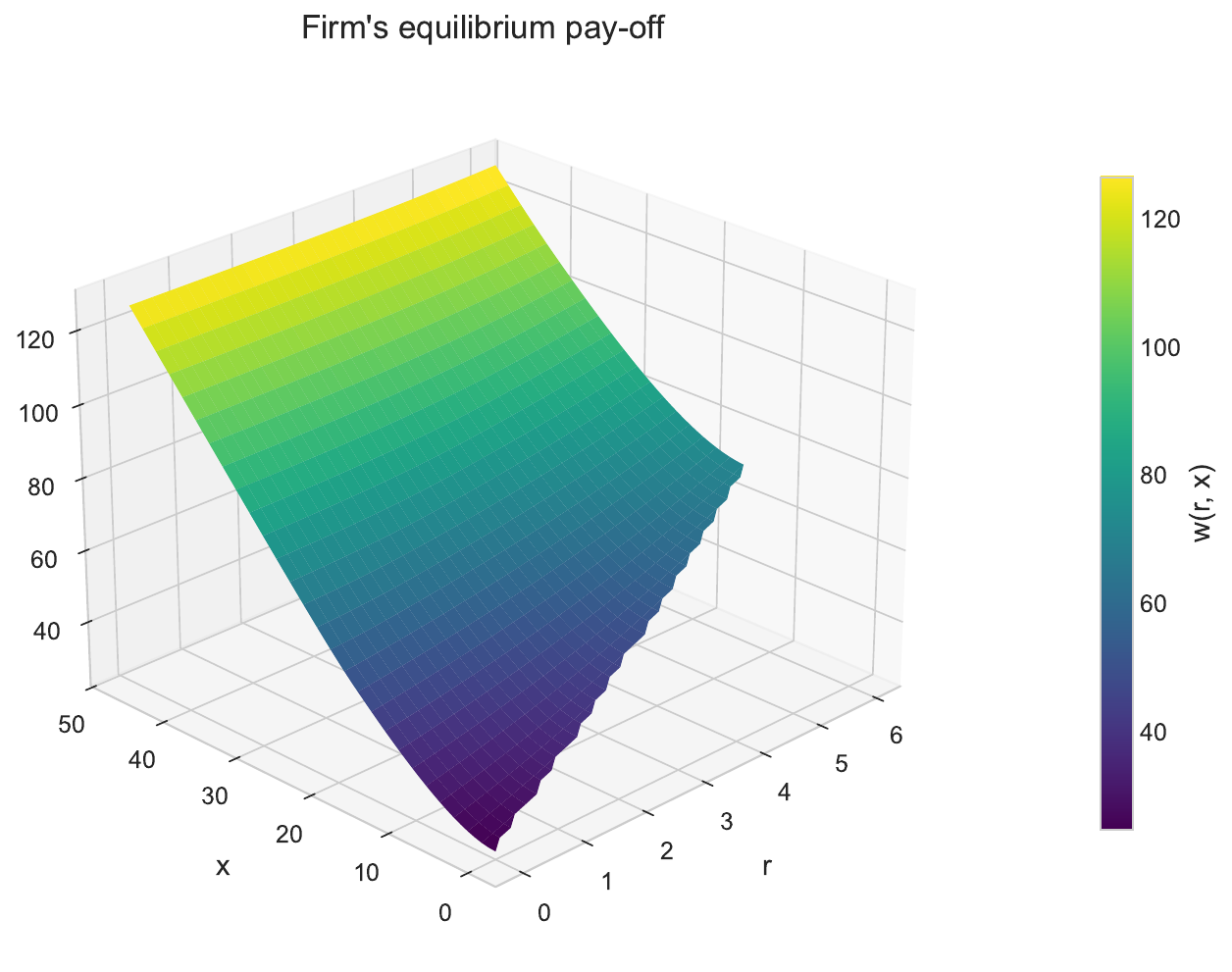}
  \caption{}
  \label{figd01}
\end{subfigure}\hfill
\begin{subfigure}{0.48\textwidth}
  \centering
  \includegraphics[width=\linewidth]{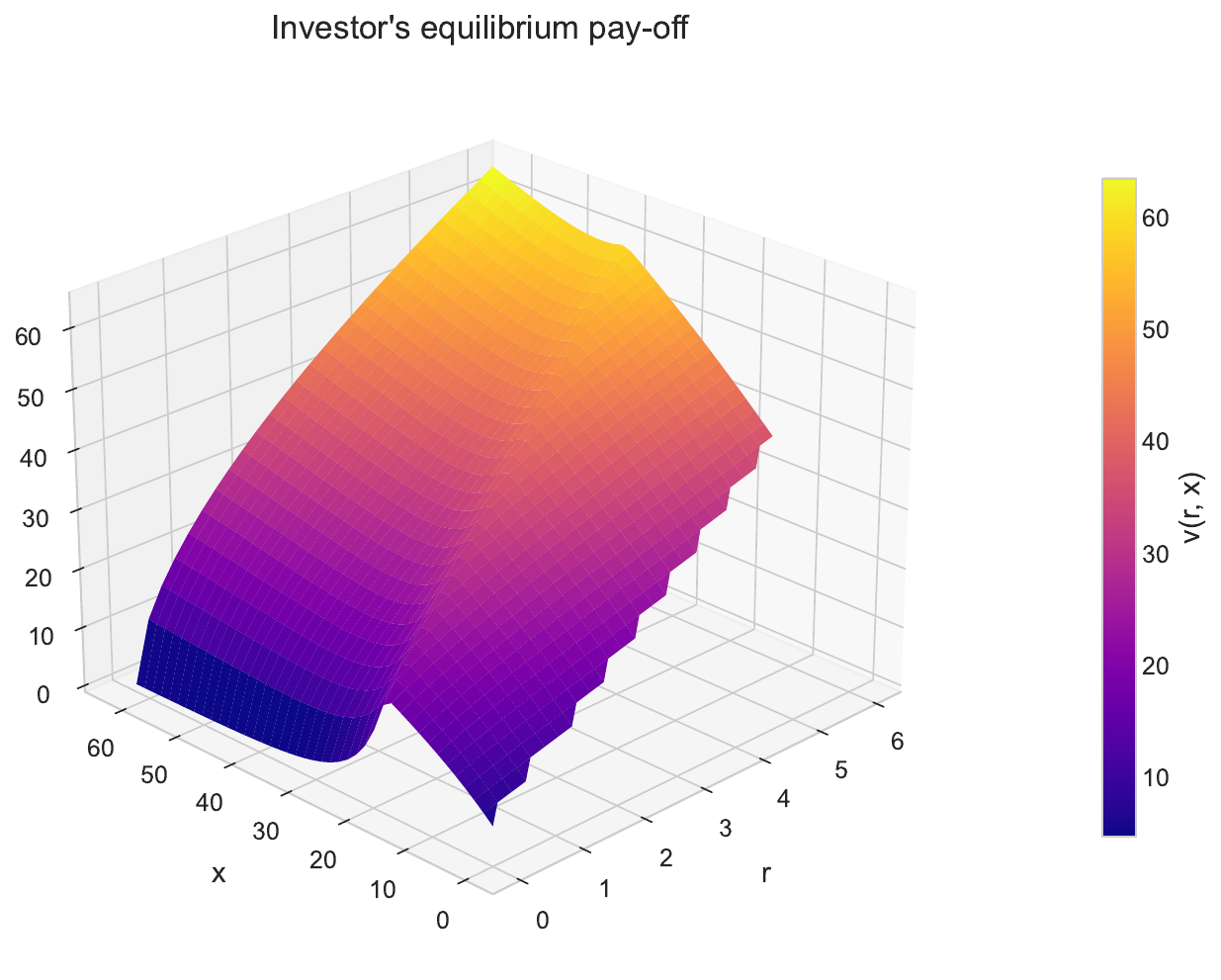}
  \caption{}
  \label{figd02}
\end{subfigure}
\caption{Firm and investor equilibrium expected payoff functions for the deterministic setting discussed in Section \ref{sec:detmunegative}, with $\mu=-0.0741$, $\rho=\bar\rho = 0.3$, $\beta=0.5$ and $\gamma=0.55$. Figure \ref{figd01} shows $w(r,x)$ and Figure \ref{figd02} shows $v(r,x)$, for $x > a(r)$.}
\label{figd1}
\end{figure}

Figure \ref{figd0} presents numerical simulations for the functions $a(r)$ and $b(r)$. 
The equilibrium payoffs for both firm and investor are illustrated in Figure \ref{figd1}. Figure \ref{ab.fig} shows sensitivity of the boundary $b(r)$ to $\eta_{\max}\to \infty$. The investor's boundary $a(r)$ does not change as it does not depend on $\eta_{\max}$, but the firm's boundary $b(r)$ tends to $+\infty$.

Notably, solving the HJB equations is not necessary in this deterministic setting because we have derived explicit solutions to the proposed optimal control problem.
Consequently, after obtaining the functions $a(r)$ and $b(r)$ using the parameters mentioned above, we can directly compute the value functions for both players, which are given by: $w(r,x)$ as in \eqref{eq:weq1} and $v(r,x)$ as in \eqref{eq:v1}, \eqref{eq:v2} and \eqref{eq:dstar}, 
with equilibrium pair $(\eta^*,\nu^*)$ as in \eqref{eq.optimalpair_det}. Notice that the investor's value function is not monotonic in $x$ because of the 'game' features of our problem. 

\begin{figure}
\centerline{\includegraphics[width=0.9\textwidth]{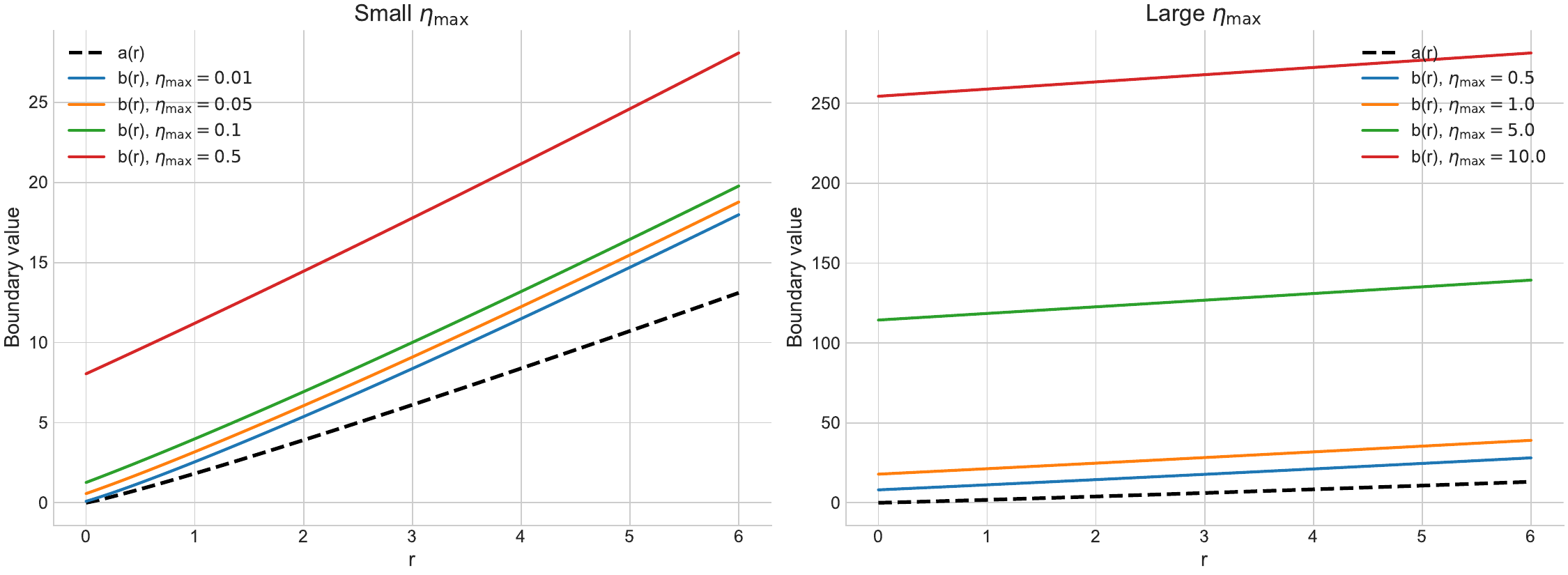}}
\caption{Investor boundary $a(r)$ and firm boundary $b(r)$ for different values of $\eta_{\max}$. The other parameters are $\bar\rho = 0.3, \mu = -0.07, \beta = 0.5, \gamma = 0.55, \alpha = 1$.}
\label{ab.fig}
\end{figure}

\subsection{Stochastic setting\label{simul:sthoc}}

An important step to obtain numerically $w$ and $v^\epsilon$ is the construction of the boundary $a^\epsilon$. It turns out that the shape of $a^\epsilon$ is qualitatively similar to the one of the initial condition $a^{\epsilon(0)} \equiv \hat a$ (i.e., the solution to the investor's problem in isolation discussed in Section \ref{sec:invest_isolation}). It is clear by its explicit expression \eqref{eqva0} that $\hat a$ is convex if $\gamma>1-\beta$ and concave if $\gamma<1-\beta$. Plots of $\hat a$ are provided in Figure \ref{figa0} for parameter values:
$\rho=\bar\rho = 0.3$, $\mu=0.0741$, $\sigma=0.3703$, $\beta=0.55$ and $\gamma = \{ 0.1, 0.25,  0.3, 0.45, 0.5, 0.6, 0.75, 0.9 \}$.

\begin{figure}[htbp]
\centering
\includegraphics[scale=.5]{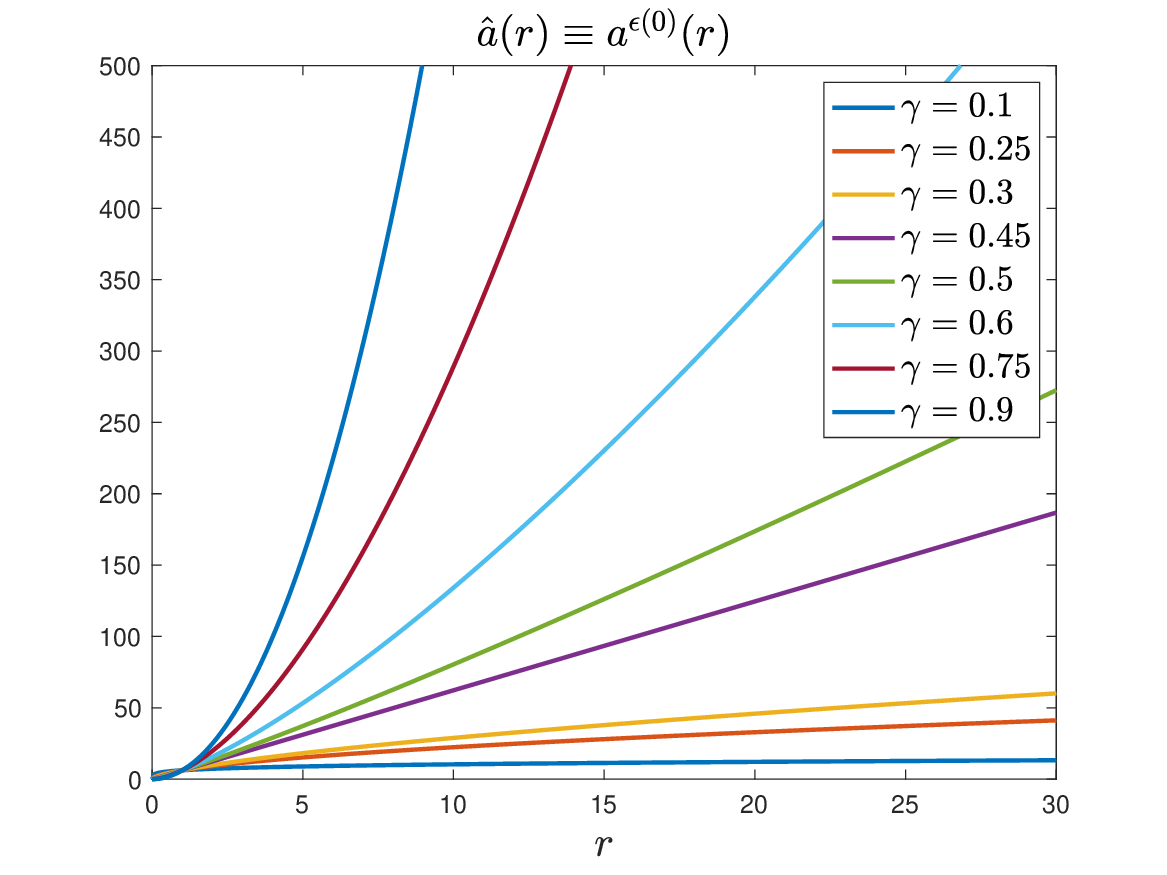}
\caption{Function $\hat a(r)$, as in \eqref{eqva0}, evaluated for $\rho=\bar\rho = 0.3$, $\mu=0.0741$, $\sigma=0.3703$, $\beta=0.55$, and $\gamma = \{ 0.1, 0.25,  0.3, 0.45, 0.6, 0.75, 0.9 \}$.}\label{figa0}
\end{figure}

Figure \ref{fig1} displays the equilibrium payoffs $w$ (Figure \ref{fig1a}) and $v^\epsilon$ (Figure \ref{fig1b}) obtained by solving \eqref{eqw1} and \eqref{eqvepsilon}.
The performance of Algorithm \ref{algo} is illustrated by Figures \ref{fig1c} and \ref{fig1d}.
We observe that the overall algorithm converges after just $4$ iterations.
In the top panel of Figure \ref{fig1c}, we show the error at the end of each sub-loop in the construction of $w(r,x)$ and $v(r,x)$. As expected we are always within the precision bound $\varpi$. In the bottom panel of Figure \ref{fig1c} we show the number of iterations required to construct the equilibrium payoff for both the firm and the investor (within the desired precision level $\varpi$).
Finally, Figure \ref{fig1d} shows the algorithm's overall convergence, with the final error achieved at $\max(\lVert w^{(\ell)}(r,x) - w^{(\ell-1)}(r,x) \rVert, \lVert v^{\epsilon(\ell)}(r,x) - v^{\epsilon(\ell-1)}(r,x) \rVert) = 6.2523 \times 10^{-4}$.

\begin{figure}[htbp]
\centering

\begin{subfigure}{0.48\textwidth}
\centering
\includegraphics[width=\linewidth]{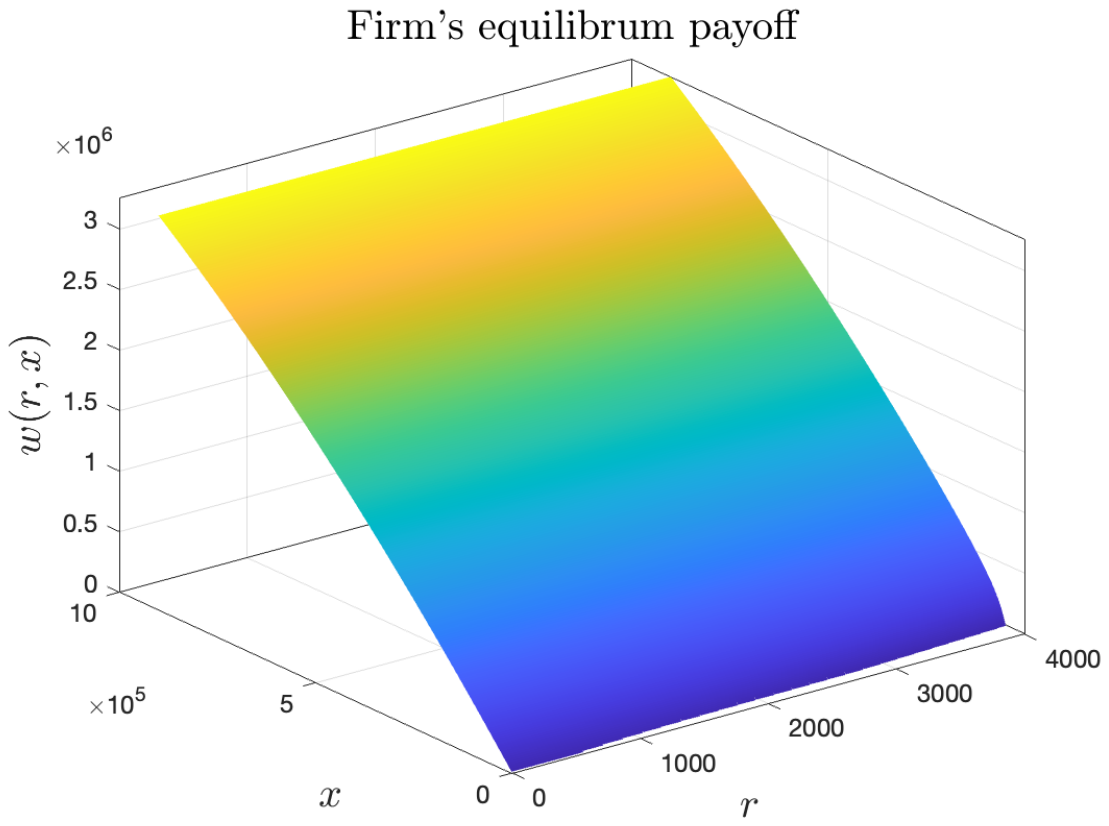}
\caption{Solution of \eqref{eqw1} with $\pi(x)=x$.}
\label{fig1a}
\end{subfigure}\hfill
\begin{subfigure}{0.48\textwidth}
\centering
\includegraphics[width=\linewidth]{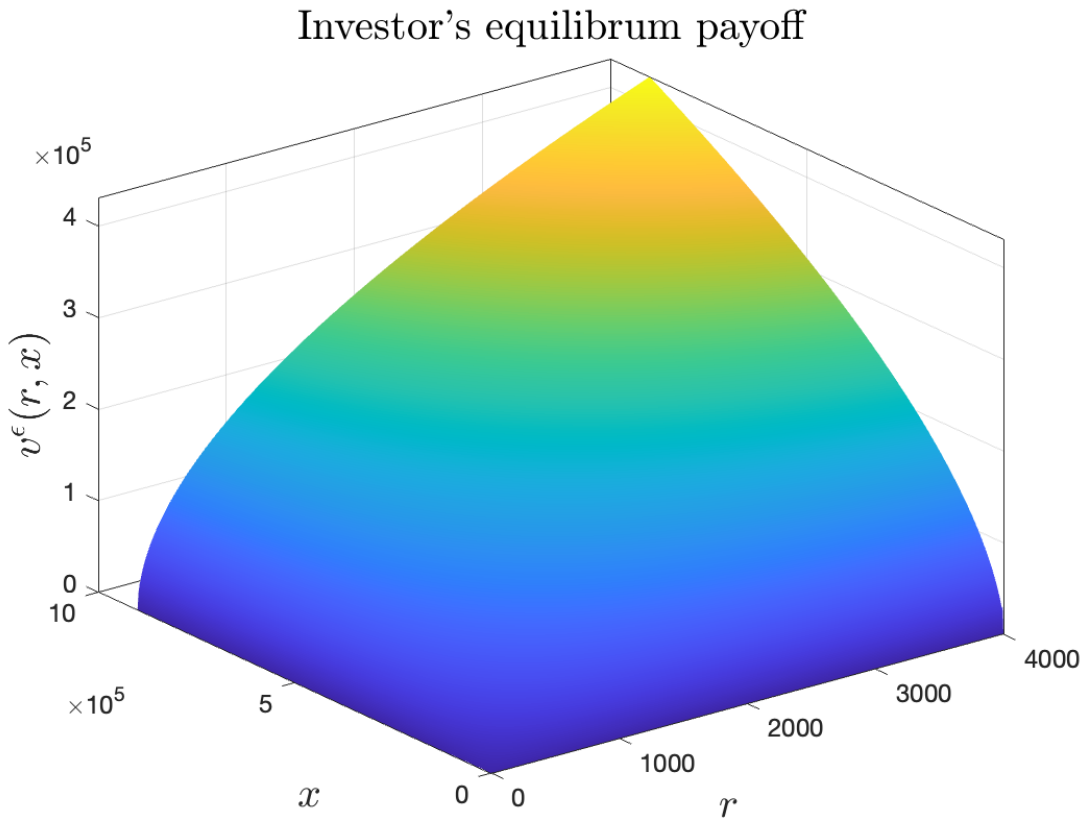}
\caption{Solution of \eqref{eqvepsilon} with $\Pi(r,x)=x^{0.55}r^{0.5}$.}
\label{fig1b}
\end{subfigure}

\vspace{0.3cm}

\begin{subfigure}{0.48\textwidth}
\centering
\includegraphics[width=\linewidth]{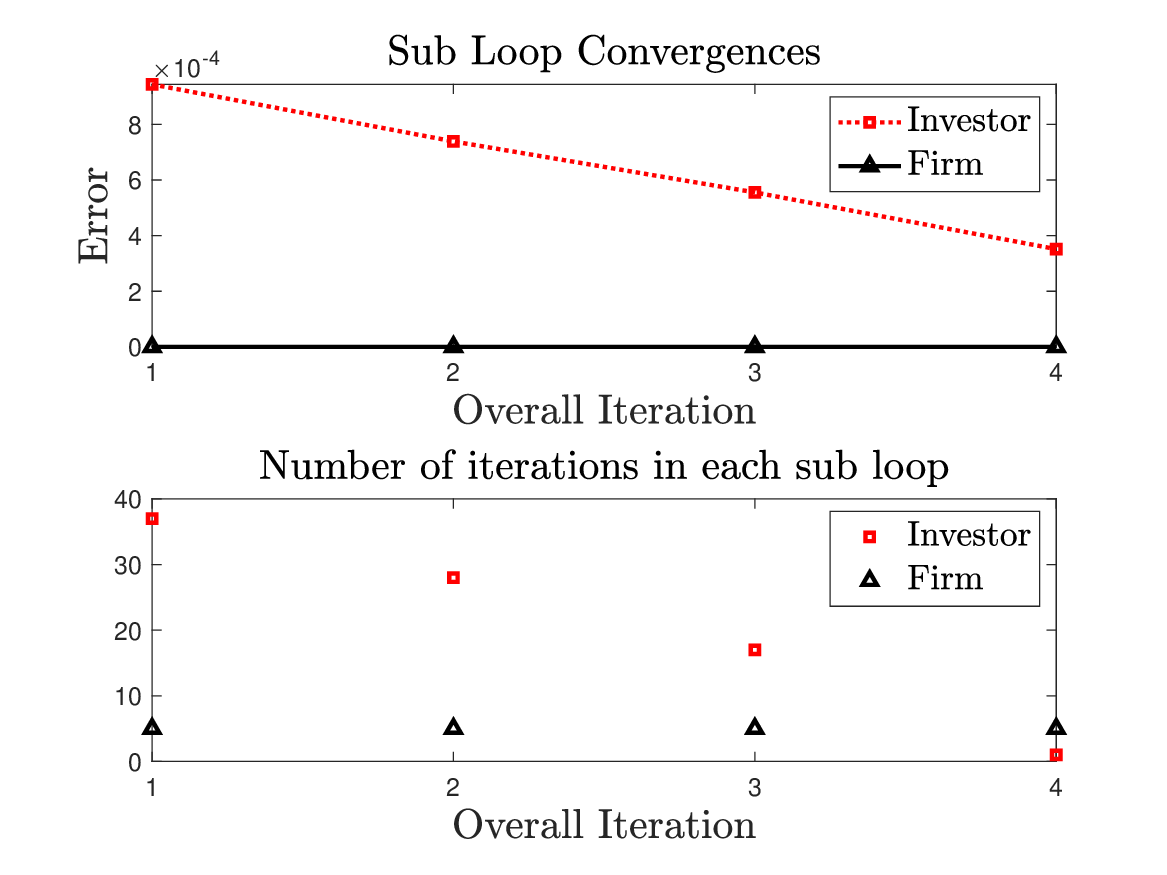}
\caption{}
\label{fig1c}
\end{subfigure}\hfill
\begin{subfigure}{0.48\textwidth}
\centering
\includegraphics[width=\linewidth]{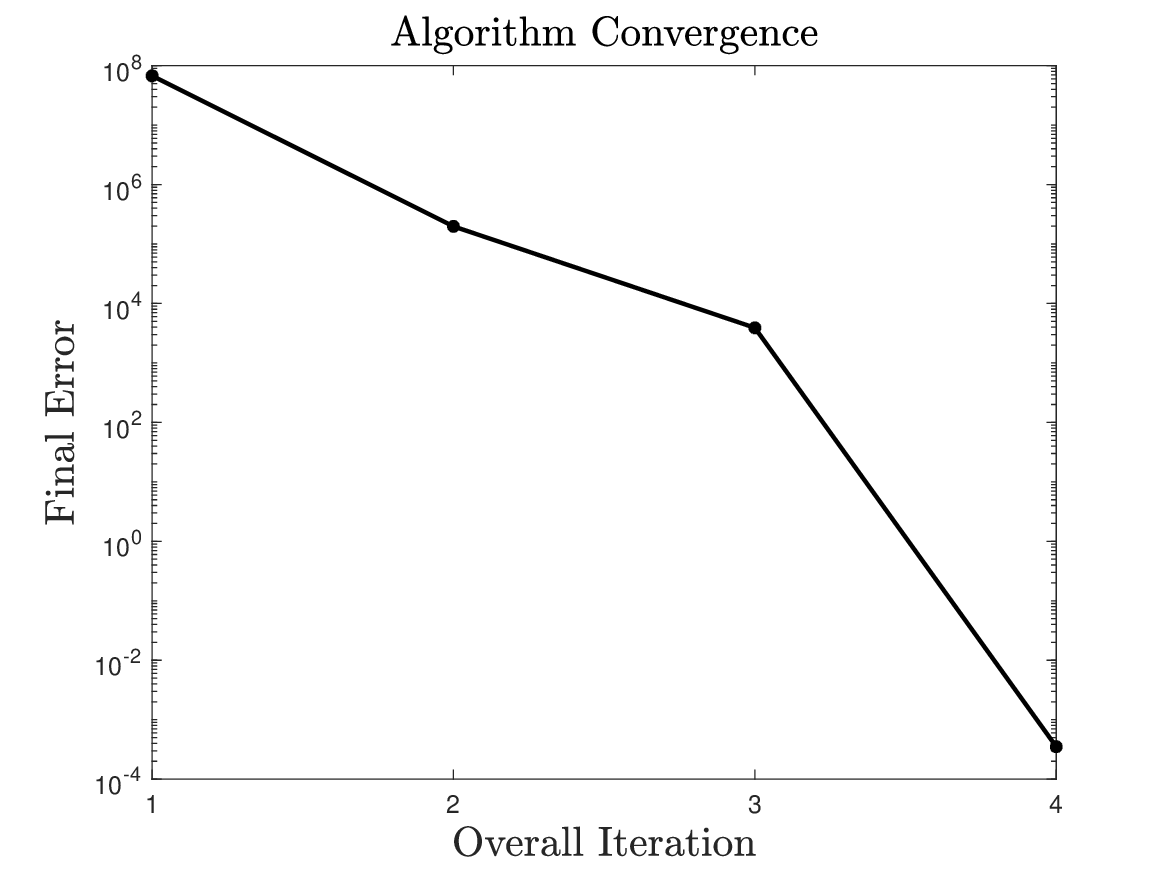}
\caption{}
\label{fig1d}
\end{subfigure}

\caption{Firm and investor equilibrium expected payoff functions obtained via Algorithm \ref{algo}, where $\mu=0.0741$, $\sigma=0.3703$, $\rho=\bar\rho=0.3$ and $\eta_{\mathrm{max}}=1$. In Figure \ref{fig1a}, $w(r,x)$ has been determined by solving \eqref{eqw1} with $\pi(x)=x$. Figure \ref{fig1b} displays the final $v^\epsilon(r,x)$ obtained by solving \eqref{eqvepsilon} with $\Pi(r,x)=x^{0.55}r^{0.5}$, $\alpha=1$ and $\epsilon=10^{-4}$. Figure \ref{fig1c} shows the performance of Algorithm \ref{algo} in each global iteration. The plots show the convergence and the number of iterations in each sub-loop for both the firm and the investor. The overall convergence of the algorithm can be seen in Figure \ref{fig1d}.}
\label{fig1}

\end{figure}

In Figure \ref{fig3} we illustrate the geometry of the regions in the state space where the firm and the investor act. More precisely, the white region corresponds to $\eta^\ast=0$, while the green region represents $\eta^\ast=\eta_{\mathrm{max}}$ (cf.\ \eqref{eq:etaphi}). A boundary $r\mapsto b^\epsilon(r)$ of the set $\{(r,x): \partial_r w(r,x) > \partial_x w(r,x)\}$, separates the firm's action region from the inaction one. The figure also displays the investor's optimal boundary  $r\mapsto a^\epsilon(r)$ derived from \eqref{eqaepsilon}. Taken together, functions $a^\epsilon$ and $b^\epsilon$ summarize the optimal strategies for both the firm and the investor. The firm mitigates emissions when $X_t^{\nu^\ast,\eta^\ast} \leq b^\epsilon(R_t^{\eta^\ast})$, while the investor provides capital when $X_t^{\nu^\ast,\eta^\ast} \leq a^\epsilon(R_t^{\eta^\ast})$. It is worth noticing that the shapes of both $r\mapsto b^\epsilon(r)$ and $r\mapsto a^\epsilon(r)$ are qualitatively similar to the shapes of the optimal boundaries $r\mapsto b(r)$ and $r\mapsto a(r)$ that we obtained in the deterministic setup of Section \ref{sec:deterministic}. Therefore, the form of the equilibrium we constructed theoretically in the deterministic framework coveys the same economic message as the one obtained numerically in the fully stochastic framework.

\begin{figure}[htbp]
\centering

\begin{subfigure}{0.48\textwidth}
\centering
\includegraphics[width=\linewidth]{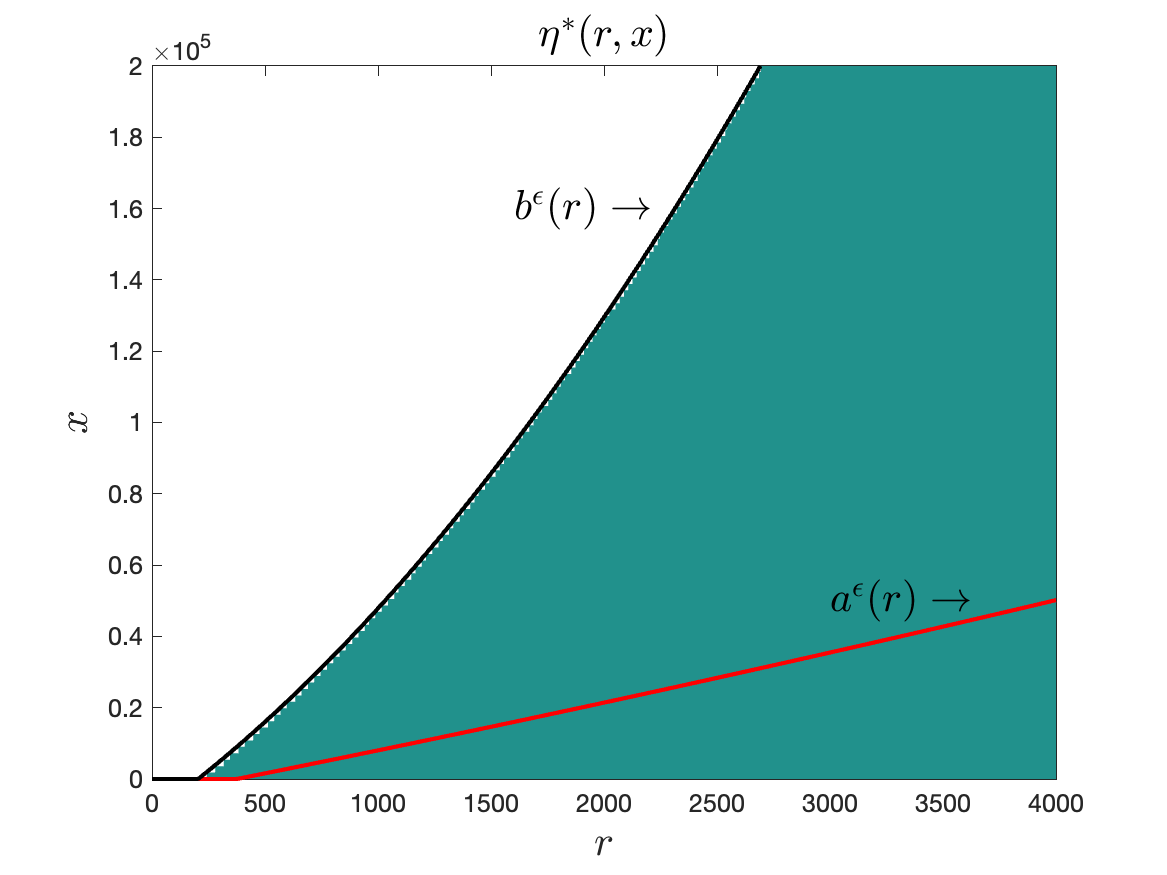}
\caption{}
\label{fig3a}
\end{subfigure}\hfill
\begin{subfigure}{0.48\textwidth}
\centering
\includegraphics[width=\linewidth]{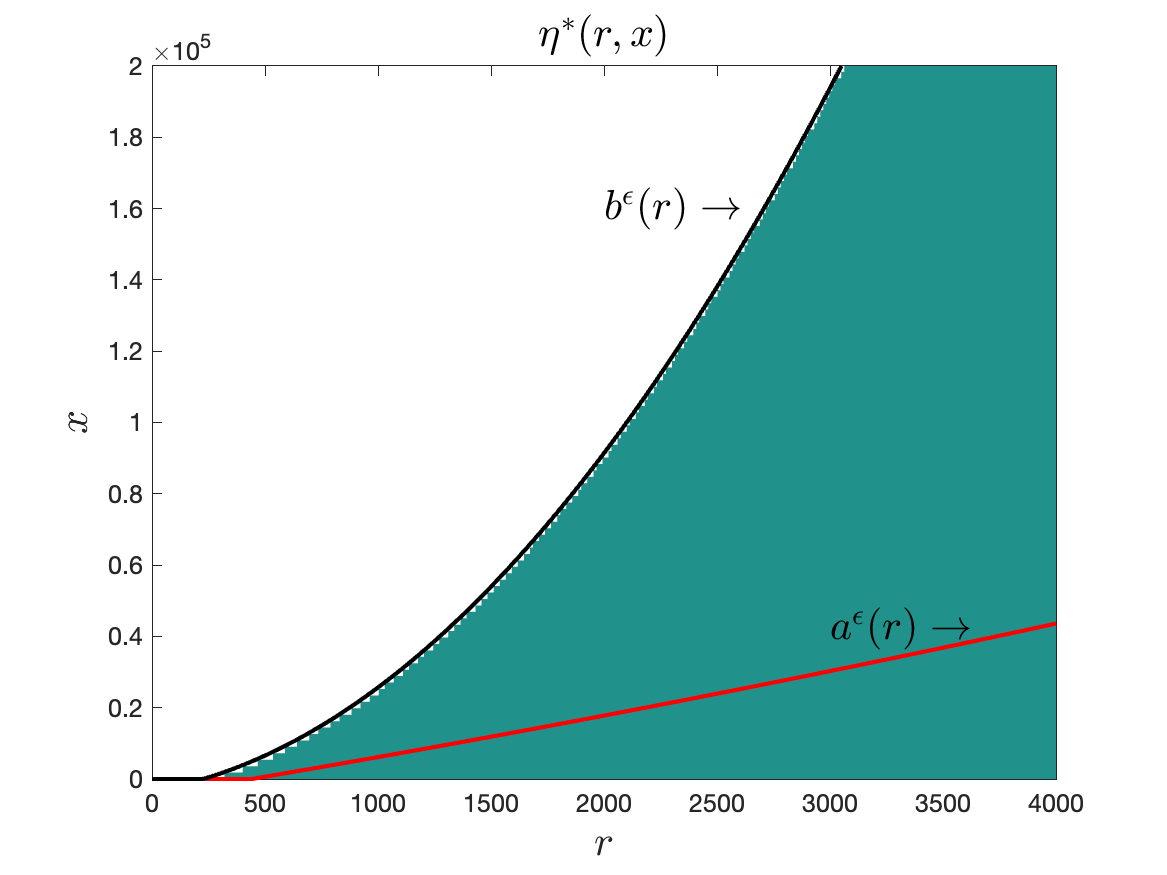}
\caption{}
\label{fig3b}
\end{subfigure}

\vspace{0.2cm}

\begin{subfigure}{0.48\textwidth}
\centering
\includegraphics[width=\linewidth]{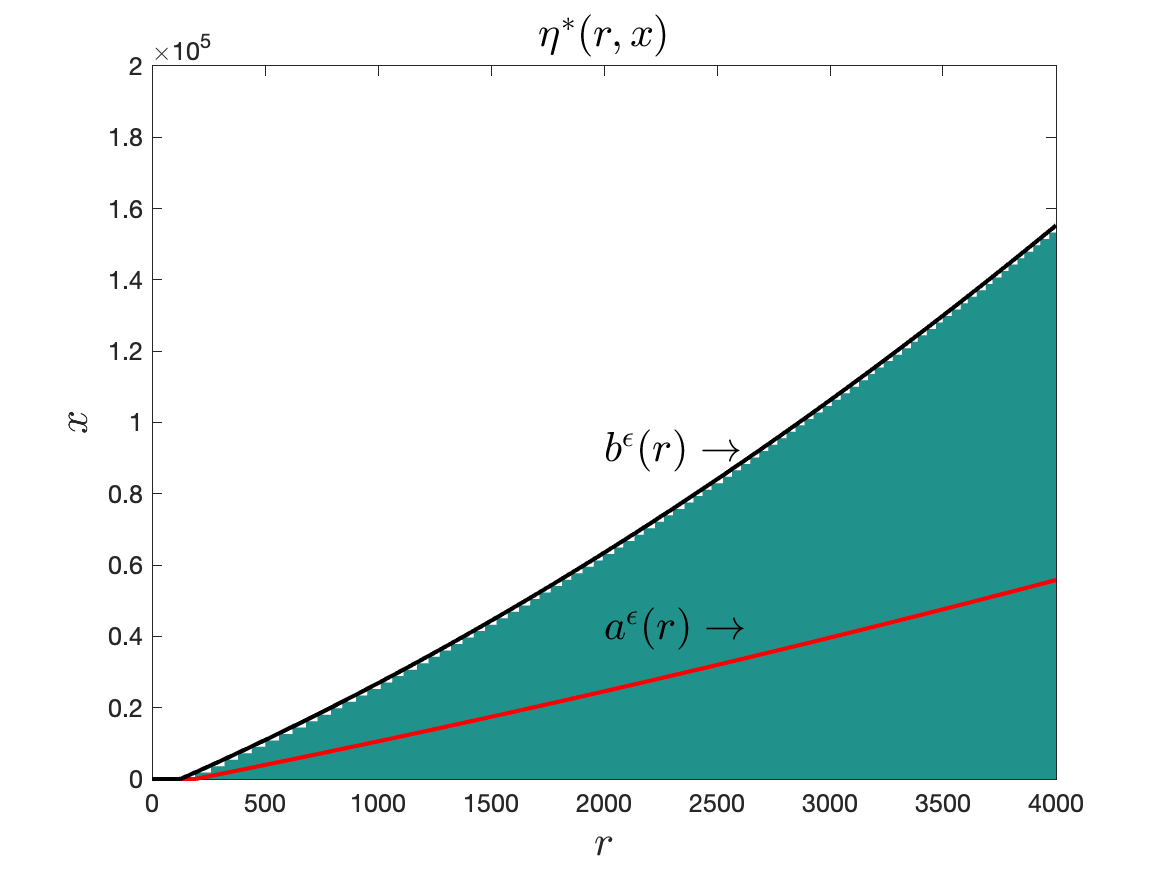}
\caption{}
\label{fig3c}
\end{subfigure}\hfill
\begin{subfigure}{0.48\textwidth}
\centering
\includegraphics[width=\linewidth]{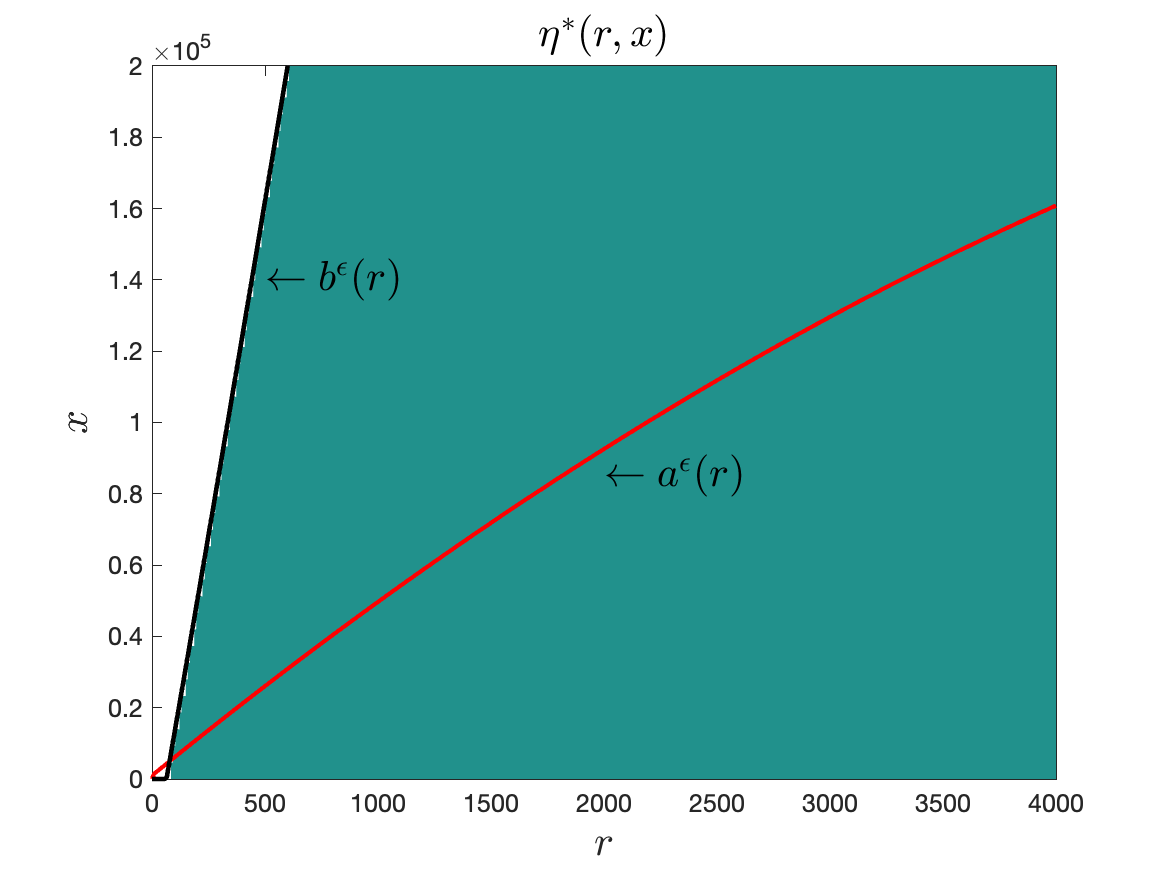}
\caption{}
\label{fig3d}
\end{subfigure}

\caption{The curve $r\mapsto b^\epsilon(r)$ is the boundary of the firm's action set. The function $r\mapsto a^\epsilon(r)$ is the investor's optimal investment threshold. The white region is the inaction region for both players. In the green region the firm implements emission abatement at the maximum rate. Optimal investment keeps the firm's production capacity above $a^\epsilon$.\\
\textbf{Parameter values}. In Figure \ref{fig3a}: $\mu = 0.0741$, $\sigma = 0.3703$, $\rho=\bar\rho = 0.2850$, $\beta=0.55$ and $\gamma=0.5$. In Figure \ref{fig3b}: $\mu = 0.0445$, $\sigma = 0.3703$, $\rho=\bar\rho = 0.2850$, $\beta=0.55$ and $\gamma=0.5$. In Figure \ref{fig3c}: $\mu = 0.0741$, $\sigma = 0.2222$, $\rho=\bar\rho = 0.2850$, $\beta=0.55$ and $\gamma=0.5$. In Figure \ref{fig3d}: $\mu = 0.16$, $\sigma = 0.4$, $\rho=\bar\rho = 0.3$, $\beta=0.75$ and $\gamma=0.2$.}
\label{fig3}

\end{figure}

Figures \ref{fig3a}--\ref{fig3c} also illustrate the sensitivity of $a^\epsilon(r)$ and $b^\epsilon(r)$ when we fix $\rho=\bar\rho = 0.2850$, $\beta=0.55$, $\gamma=0.5$, and perturb the values of $\mu$ and $\sigma$ as shown in Table \ref{tab1} below.
We observe that the investment boundary is higher for higher values of $\mu$
 and for lower values of $\sigma$: since the investor is profit-seeking and risk-averse, they are more inclined to invest into a more profitable and less risky company. The company is risk-neutral, so when the volatility $\sigma$ is lower but the drift $\mu$ is the same (Figure \ref{fig3c}), the company's mitigation actions are mostly determined by those of the investors: since the investment boundary is higher, the company is able to attract the same level of investment with less mitigation. When the volatility is the same but the drift is lower (Figure \ref{fig3b}), there are two competing effects: on the one hand, lower investment boundary motivates the company to mitigate more, on the other hand, lower profitability motivates it to mitigate less because the returns from mitigation are lower. Overall, we see that the mitigation boundary is lower but not by so much as in Figure \ref{fig3c}.
 Additionally,  Figure \ref{fig3d} shows the values of $\eta^\ast(r,x)$ for concave functions $a^\epsilon(r)$ and $b^\epsilon(r)$, where $\mu = 0.16$, $\sigma = 0.4$, $\rho=\bar\rho = 0.3$, $\beta=0.75$ and $\gamma=0.2$.
Similarly, numerical simulations show that the two boundaries become closer when $\rho$ increases.

\begin{table}[ht]
    \centering
    \begin{tabular}{lcc} \toprule
     & {$\mu$} & {$\sigma$}\\ \midrule
    {Figure \ref{fig3a}}   & 0.0741  & 0.3703  \\
    {Figure \ref{fig3b}}  & 0.0445  & 0.3703 \\
    {Figure \ref{fig3c}}  & 0.0741  & 0.2222 \\ \bottomrule
    \end{tabular}
    \caption{Comparison between the parameters that show the sensitivity of $\eta^\ast(r,x)$.}
    \label{tab1}
\end{table}

Figure \ref{figsigma} shows the sensitivity of the investment and abatement thresholds to changes in volatility. We observe that when the volatility coefficient increases, the investor delays investment ($a^\epsilon$ decreases) and the firm expedites abatement ($b^\epsilon$ increases). The decrease in the investment boundary is in line with results from the irreversible investment literature: higher volatility delays the exercise of the investment option. The firm's behaviour instead suggests that a riskier environment incentivizes the firm to seek external investment and therefore to abate emissions earlier.
\begin{figure}[htbp]
\centering
\includegraphics[width=0.9\linewidth]{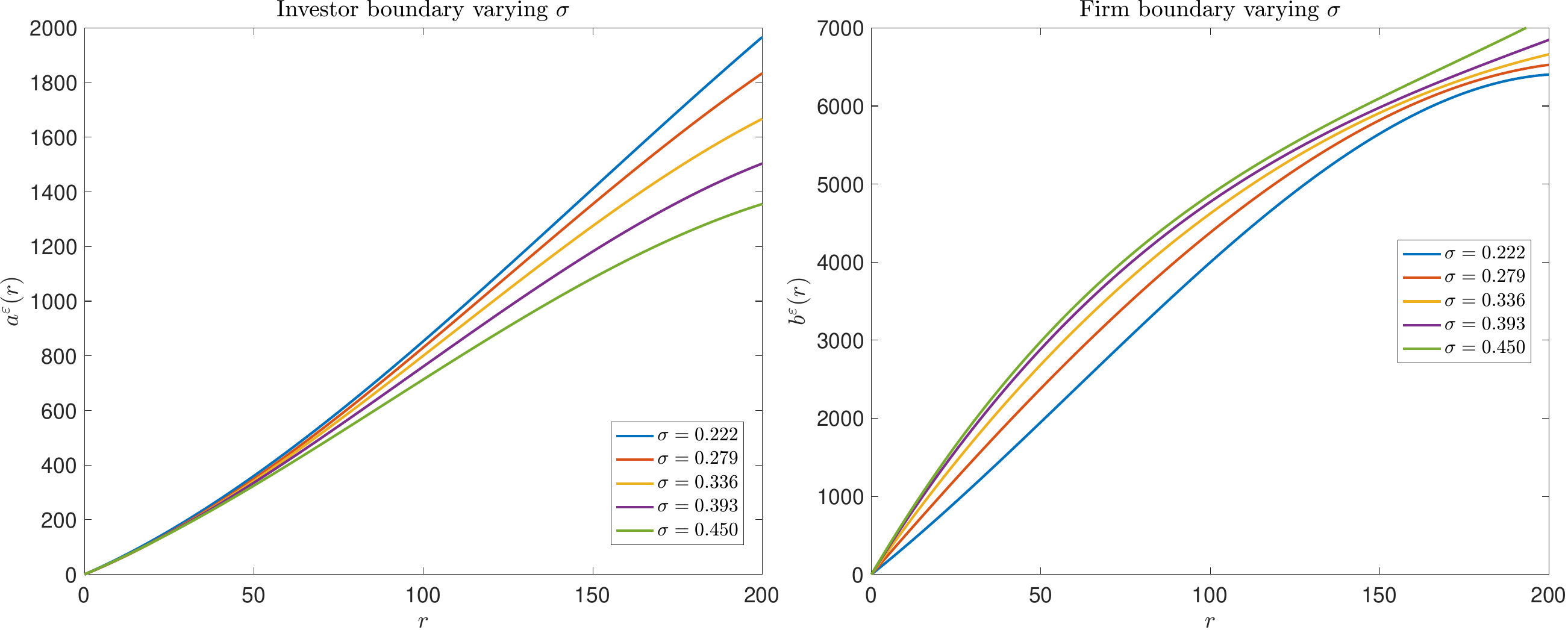}
\caption{Comparison of the investment ({\em left panel}) and abatement ({\em right panel}) thresholds for different values of the volatility coefficient.}
\label{figsigma}
\end{figure}

Figure \ref{figgamma} shows that increasing the risk aversion for the investor expedites both the firm's abatement decision and the investor's intervention (both $b^\epsilon$ and $a^\epsilon$ increase with $\gamma$). This is also in line with irreversible investment theory, where a more risk averse investor tends to invest earlier to support production capacity. As a result, the firm has a bigger incentive to abate emissions. The result is also in line with our observations in the deterministic case (cf.\ Figure \ref{ab.fig}).
\begin{figure}[htbp]
\centering
\includegraphics[width=0.9\linewidth]{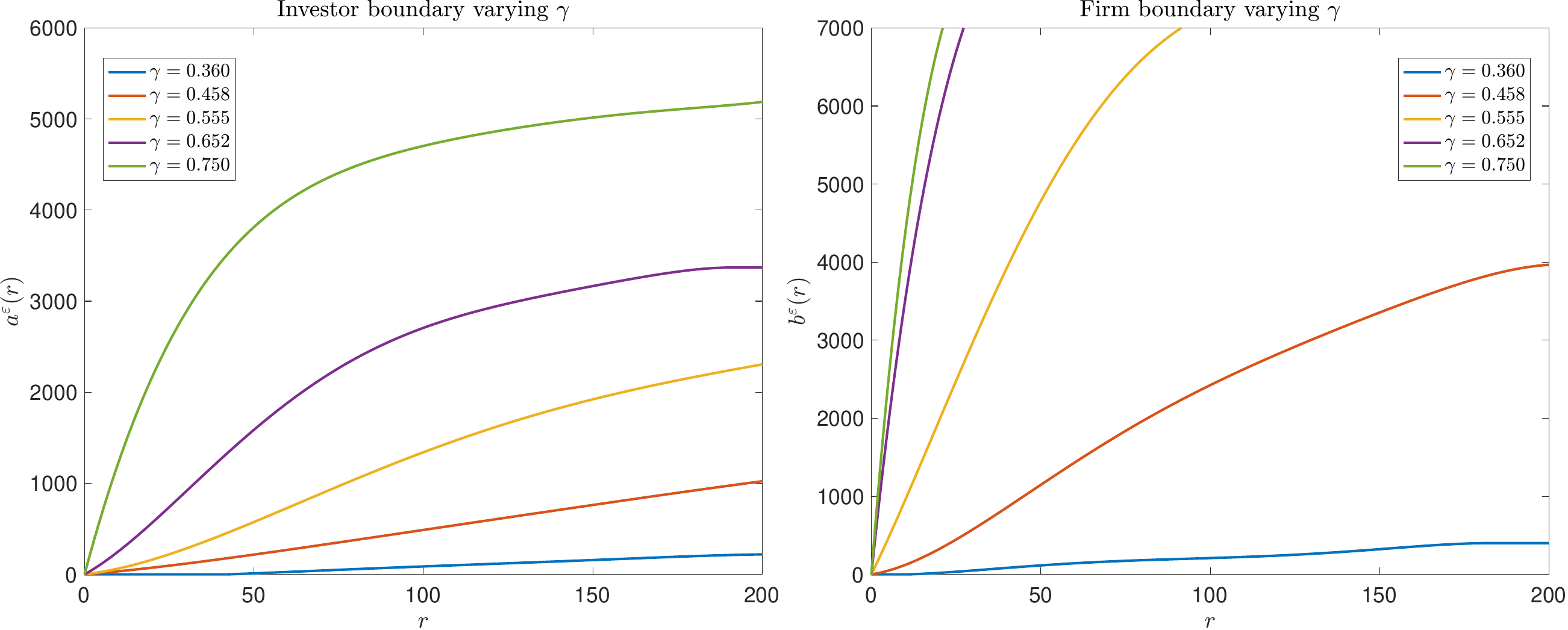}
\caption{Comparison of the investment ({\em left panel}) and abatement ({\em right panel}) thresholds for different values of the risk aversion coefficient.}
\label{figgamma}
\end{figure}

\begin{figure}[htbp]
\centering

\begin{subfigure}{0.48\textwidth}
\centering
\includegraphics[width=\linewidth]{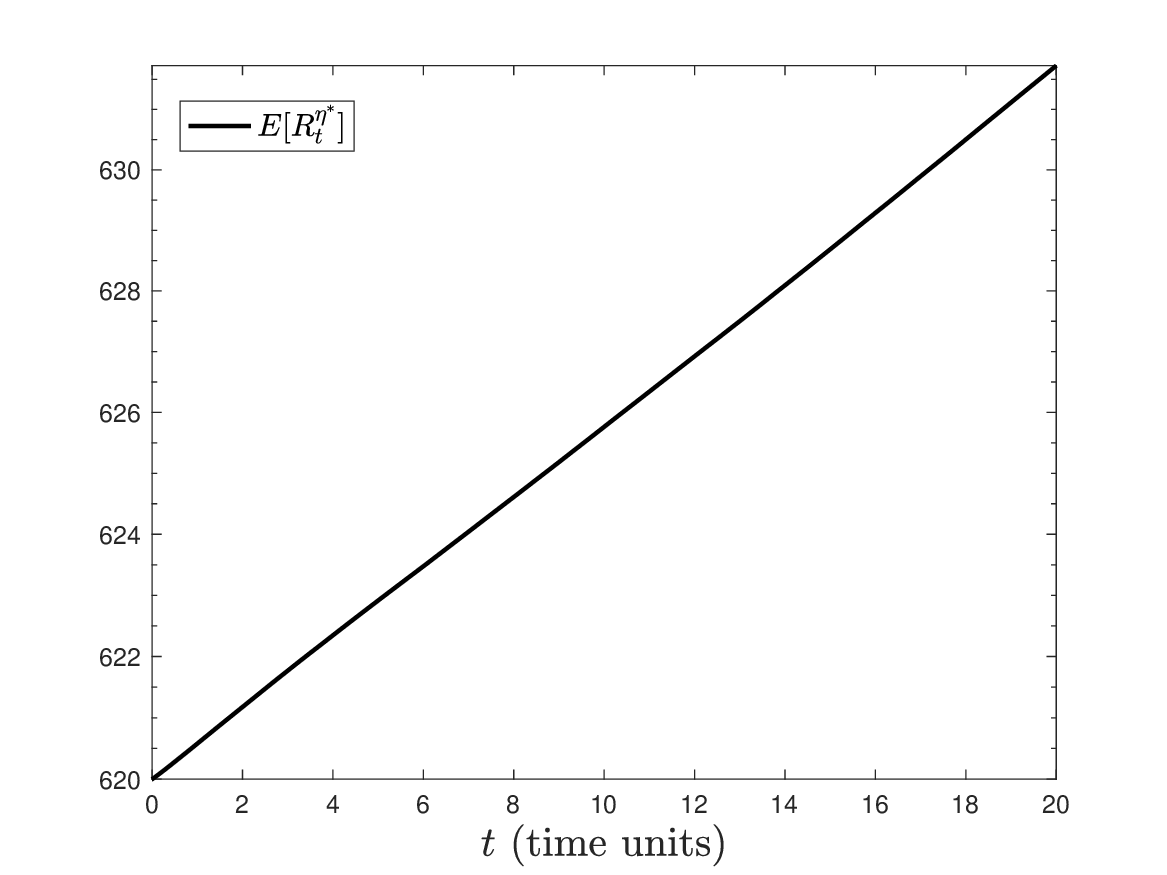}
\caption{}
\label{fig4c}
\end{subfigure}\hfill
\begin{subfigure}{0.48\textwidth}
\centering
\includegraphics[width=\linewidth]{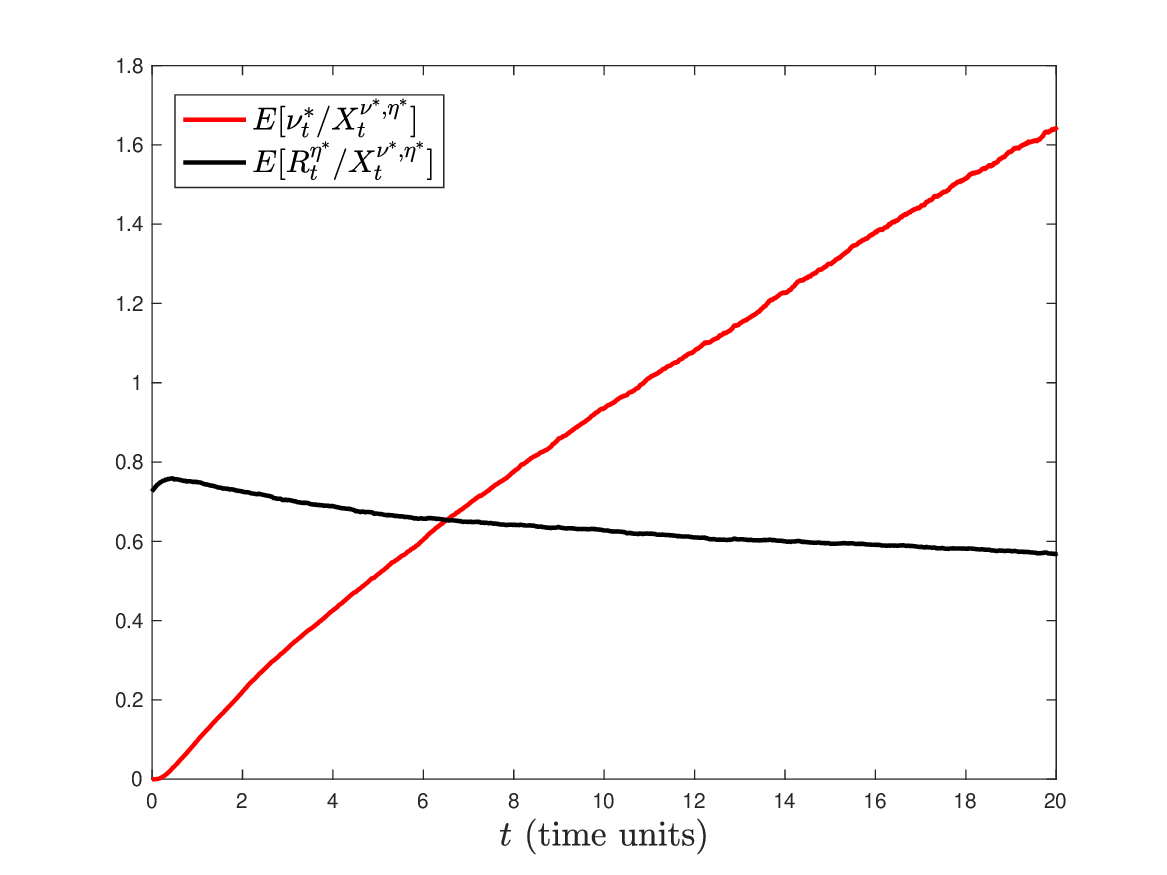}
\caption{}
\label{fig4d}
\end{subfigure}

\caption{Figures \ref{fig4c} and \ref{fig4d} illustrate: (i) the dynamics of the average amount allocated by the firm towards emission reduction over time, $\mathbb{E}[R_t^{\eta^\ast}]$, (ii) the expected ratio of the firm's abatement spending to total profits, $\mathbb{E}[R_t^{\eta^\ast} / X_t^{\nu^\ast,\eta^\ast}]$, and (iii) the expected ratio of the investor's cumulative financial contribution to total profits, $\mathbb{E}[\nu_t^\ast / X_t^{\nu^\ast,\eta^\ast}]$. These quantities were approximated via Monte Carlo simulations with $10^5$ sample paths.}
\label{fig4}

\end{figure}

Figures \ref{illustration2.fig} and \ref{fig4} present the numerical simulations derived from Algorithm \ref{algo} for a negative value of $\mu$ ($\mu = -0.0445$), with other parameters  $\sigma = 0.3703$, $\rho = \bar\rho = 0.2850$, $\beta = 0.65$ and $\gamma=0.36$.   
Figure \ref{illustration2.fig} shows optimal control actions through the dynamics of $(R_t^{\eta^*},X_t^{\nu^*,\eta^*})$, $a^\epsilon(R_t^{\eta^*})$ and $b^\epsilon(R_t^{\eta^*})$, for initial conditions given by $R_0 = 620$ and $X_0 = 854$.
Here, the dynamic of the profit stream  \eqref{eqX} was obtained through the classical Euler--Maruyama method. This figure presents one sample path of the processes $X_t^{\nu^*,\eta^*}$, $a(R_t^{\eta^*})$ and $b(R_t^{\eta^*})$.
Firm's profits $X_t^{\nu^\ast,\eta^\ast}$ are maintained by the investor above the moving investment boundary $a(R_t^{\eta^\ast})$ (via Skorokhod reflection). 
We also observe that when the firm's profits are sufficiently high (i.e., $X_t^{\nu^*,\eta^*} > b(R_t^{\eta^*})$) there is no emission abatement and $R_t^{\eta^\ast}$ remains constant. 

In order to analyze average quantities, we conducted Monte Carlo simulations with $10^5$ sample paths.
Figures \ref{fig4c} and \ref{fig4d} illustrates the resulting average investment strategies for both the firm and the investor.
These figures illustrate three key dynamics: the evolution of average cumulative spending in emission abatement, denoted as $t\mapsto \mathbb E[R^{\eta^\ast}_t]$ (Figure \ref{fig4c}); the evolution of the ratio between total abatement spending and current profits, represented by $t\mapsto \mathbb E[R_t^{\eta^\ast}/ X_t^{\nu^\ast,\eta^\ast}]$; and the evolution of the ratio between the investor's cumulative financial investment and the firm's current profits, given by $t\mapsto \mathbb E[\nu^\ast_t/ X_t^{\nu^\ast,\eta^\ast}]$ (Figure \ref{fig4d}).
We see that in the long-run the average investment-to-profit ratio increases faster than the average abatement-to-profit ratio, whereas in the short-run the abatement increases faster than the investment, relatively to the profit levels. This seems to indicate that emission reduction brings a long-term financial benefit to the firm by allowing it to attract larger investment levels (relatively to the production capacity) compared to a situation where abatement actions are not taken.

\begin{figure}[htbp]
\centering

\begin{subfigure}{0.48\textwidth}
\centering
\includegraphics[width=\linewidth]{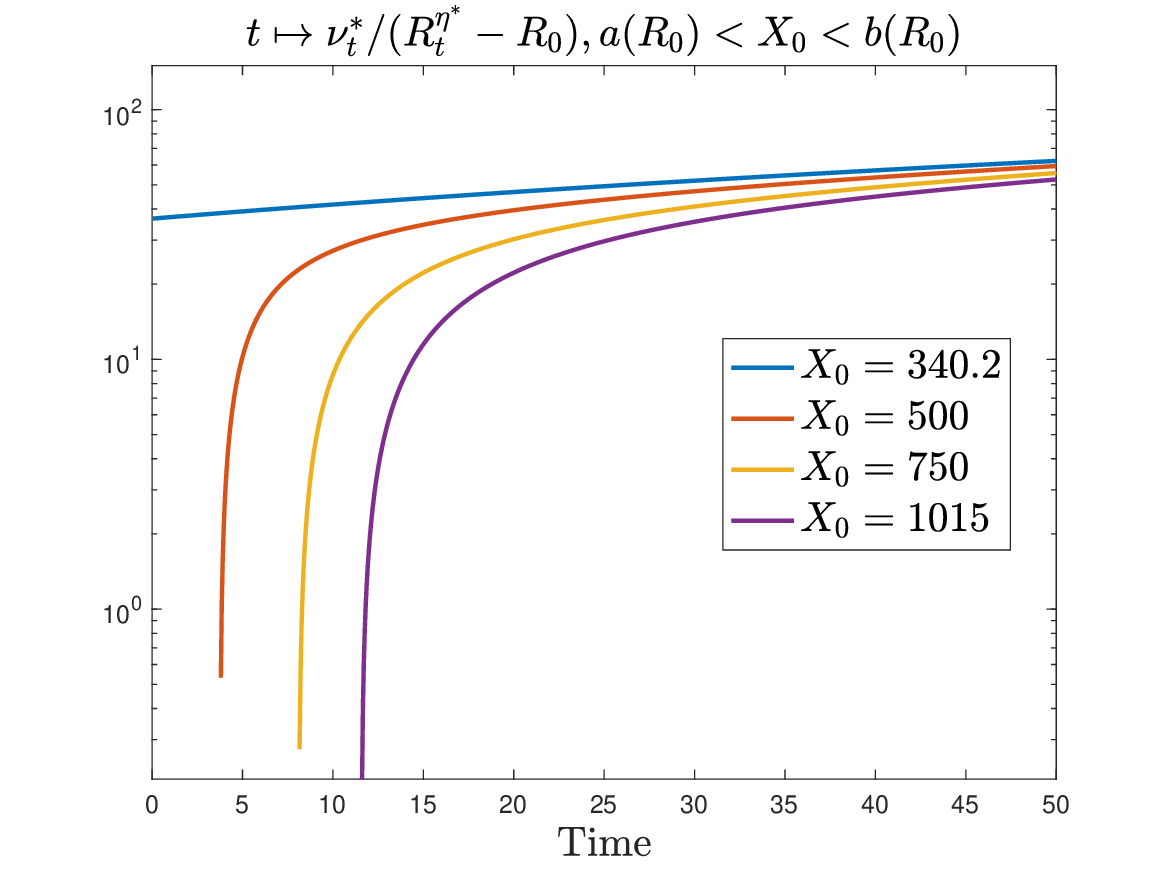}
\caption{}
\label{fig6a}
\end{subfigure}\hfill
\begin{subfigure}{0.48\textwidth}
\centering
\includegraphics[width=\linewidth]{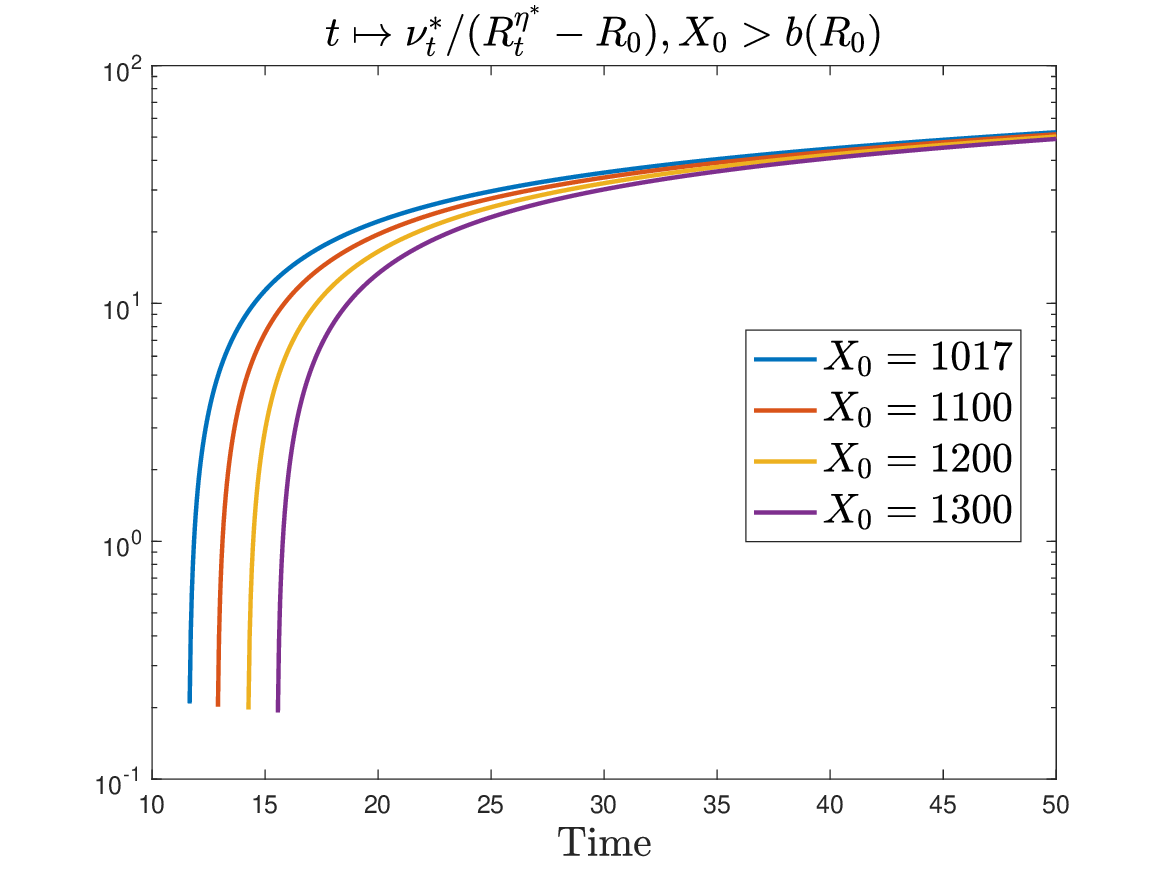}
\caption{}
\label{fig6b}
\end{subfigure}

\vspace{0.2cm}

\begin{subfigure}{0.48\textwidth}
\centering
\includegraphics[width=\linewidth]{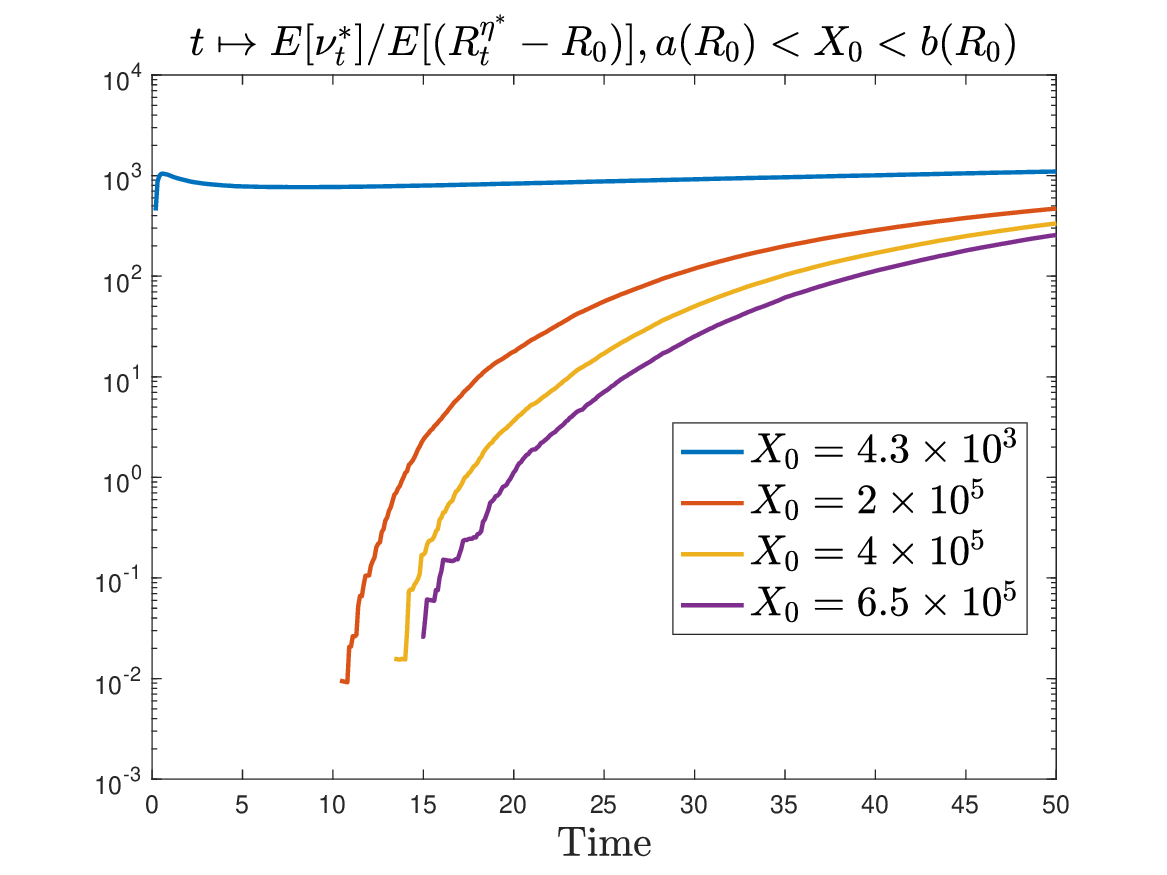}
\caption{}
\label{fig6c}
\end{subfigure}\hfill
\begin{subfigure}{0.48\textwidth}
\centering
\includegraphics[width=\linewidth]{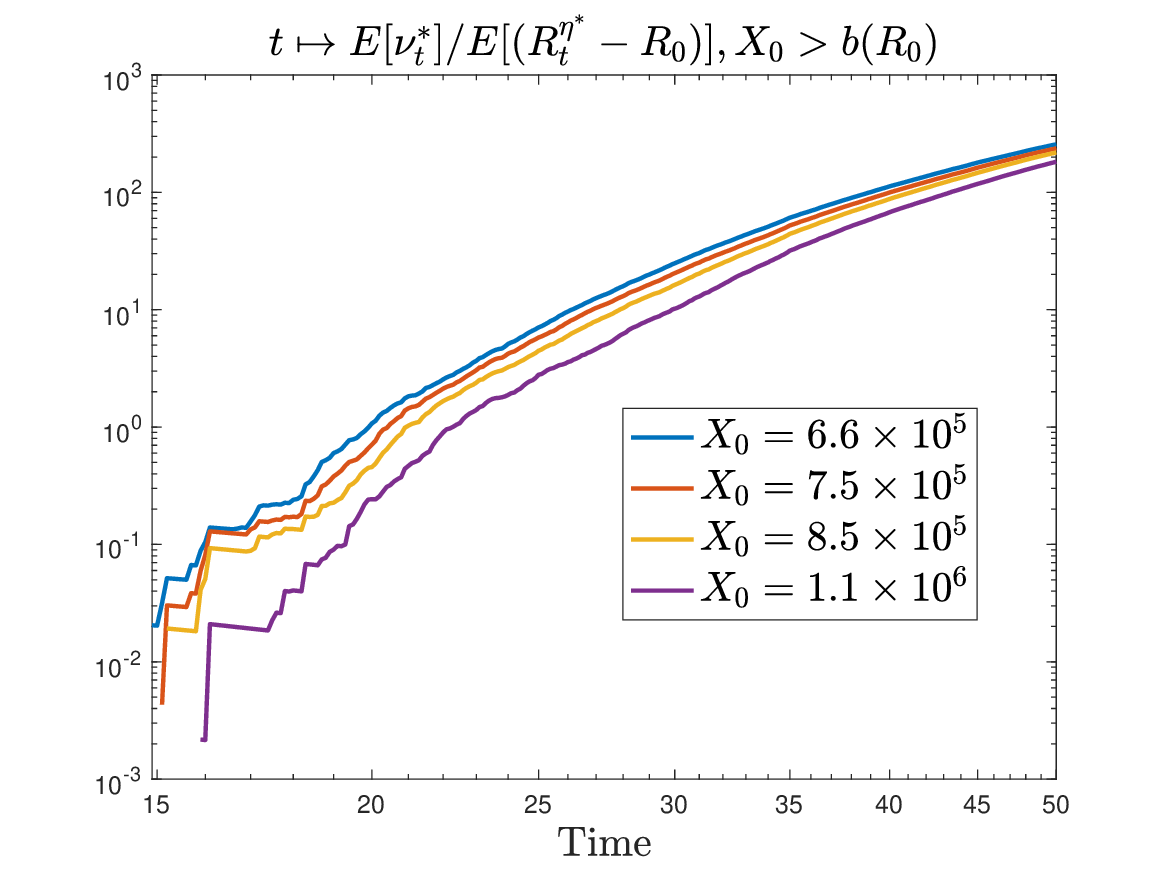}
\caption{}
\label{fig6d}
\end{subfigure}

\caption{Figures \ref{fig6a}--\ref{fig6b} show the dynamics of the ratio between total investment and total emission abatement as a function of time, in the zero-noise limit, for different values of the initial firm size, with $\rho = 0.06$ and $\mu = -0.05$. In the left panel the initial firm size takes values in the range $a(R_0) < X_0 < b(R_0)$, while in the right panel it takes values in the range $X_0 > b(R_0)$. Figures \ref{fig6c}--\ref{fig6d} show the analogous dynamics for the stochastic case, with $\rho = 0.06$, $\mu = -0.05$ and $\sigma = 0.2234$. Here we consider the ratio between the average total investment and the average total emission abatement as functions of time.}
\label{fig6}

\end{figure}

{Finally, Figure \ref{fig6} illustrates the dynamics over time of the ratio between the total investment and the total emission abatement. This ratio can be seen as a measure of the impact of the emission abatement on the financial performance of the company: the higher the ratio, the higher the incentive for the firm to abate. We remark that this ratio is decreasing in the initial production capacity $X_0$:  the lower $X_0$, the less abatement is necessary before the investor starts adding money to the firm. Similarly, for short time horizons, it is highly unlikely that the firm will reach the investment region, so that the ratio is zero for small $t$ and then increases as the probability of attracting investment becomes higher. }

\section{Conclusions}\label{sec:concl}

In this paper, we presented a novel tractable continuous-time model for strategic sustainable investment, framed as as a nonzero-sum stochastic differential game between a privately owned company and a sustainable impact-minded investor. Our model has several implications for impact investing. First, although the objective function of the company only depends on its financial performance, in equilibrium, the company may strategically decide to abate emissions to attract capital from impact investors. Second, this happens when the financial performance of the company drops below the ``abatement boundary": the company whose financial performance is strong enough does not require external capital and is therefore not incentivised to abate emissions. Third, the investors start injecting capital when the financial performance of the company reaches the ``investment boundary", which is below the abatement boundary: this means that the company must perform abatement for some time before the capital from impact investors starts to flow. Fourth, in the long term, when the company's financial performance is at or around the investment boundary, the company may improve its financial performance by increasing its environmental impact, as this will lead to additional capital inflows from impact investors.

The added complexity of the continuous-time formulation is justified by the novel features of our solution, which could not have been discovered in a one-period model: the optimal strategies of the company and the investor are characterized through abatement and investment boundaries, and have a temporal ordering: in equilibrium, the company acts first by starting abatement, which subsequently allows it to attract investor capital. 

Multiple extensions of our basic model could be considered in subsequent research. First, one could include environmental performance metrics into the company's objective function, to account for executive compensation plans indexed on ESG performance. This would realign the interests of the company with those of the investors, and increase the company's environmental impact. Second, one could allow the company to attract the capital of both environmentally-minded and commercial investors. This would reduce the incentive for the company to abate emissions and, consequently, reduce its environmental impact. Finally, one could allow for blended financing where the government provides additional funding or capital guarantees to facilitate the flow of private capital to impactful projects. 



\section*{Acknowledgement}

Caio César Graciani Rodrigues gratefully acknowledges funding from the Scuola Superiore Meridionale, MERC department.
Tiziano De Angelis was partly funded by Next Generation EU -- PRIN2022 (2022BEMMLZ) and PRIN-PNRR2022 (P20224TM7Z). Peter Tankov gratefully acknowledges financial support from the ANR  under the program France 2030 (grant ANR-23-EXMA-0011), from
the FIME Research Initiative of the Europlace Institute of Finance and from the sponsoring program ``Impact Investing Chair'' led by Groupe ENSAE-ENSAI and the Fondation ENSAE-ENSAI, sponsored by
Mirova.
Part of this work was carried out when Peter Tankov was a senior visiting fellow of LTI@UniTo program of Collegio Carlo Alberto, Turin, Italy. The support of this program is gratefully acknowledged. 

The authors would like to thank Prof.\ R.\ Frey, Prof.\ E.\ Barucci, Dr.\ A.\ Calvia and participants to the Green Finance Workshop 2025 (Rome TorVergata University) for useful comments.

\section*{Data Availability Statement}

Data sharing not applicable to this article as no datasets were generated or analysed during the current study.

\bibliographystyle{siam}
\bibliography{ref}

\appendix

\section{Calibration roadmap}\label{appendix.sec}

\begin{table}[htbp]
\centering
\caption{Mapping of model parameters to observable quantities.}
\label{tab:param_mapping}

{\small
\renewcommand{\arraystretch}{1.5}
\begin{tabular}{p{0.12\textwidth} p{0.25\textwidth} p{0.34\textwidth} p{0.22\textwidth}}
\hline
\textbf{Parameter} & \textbf{Model role (interpretation)} & \textbf{Practical proxy / calibration target} & \textbf{Typical data source} \\
\hline
$\mu,\sigma$
& Drift and volatility of the firm’s baseline profit/capacity process $X_t$  in the absence of controls.
& Estimate from time series of operating profit / EBITDA / cash flow. In the numerical illustration, $\mu,\sigma$ are taken from an empirical estimate in \cite{Reddy2016}.
& Firm financial statements; Compustat/ORBIS; sector studies. \\[0.4em]

$\rho,\bar{\rho}$
& Investor and firm discount rates.
& Map to (i) investor required return / hurdle rate and (ii) firm WACC (or project discount rate). Differences can reflect heterogeneous horizons.
& Fund documentation; market data for discount rates; firm WACC estimates. \\

$\eta_{\max}$
& Upper bound on the abatement spending \emph{rate} $\eta_t$ (constraint $0\le \eta_t\le \eta_{\max}$); determines maximal feasible speed of mitigation spending.
& Calibrate as an implementation/capacity constraint: maximum annual decarbonization capex+opex the firm can deploy (engineering, procurement, regulatory, organizational constraints).
& Firm climate capex plans; sustainability reports; sector benchmarks. \\[0.4em]

$\alpha$
& Investor’s marginal cost of increasing production capacity by one unit; can capture investment frictions (fees/intermediation costs).
& Calibrate to effective cost of capital injection (fees, monitoring costs, intermediation wedges), or normalize ($\alpha=1$) under a unit choice.
& Fund fee schedules; observed financing spreads / issuance costs; calibration normalization. \\[0.4em]

$\beta,\gamma$
& Elasticities in the investor’s Cobb--Douglas payoff $\Pi(r,x)=x^\beta r^\gamma$: relative weight on financial performance vs.\ environmental performance.
& Choose to match ``impact--return trade-off'': e.g., how investor contributions vary with profitability vs.\ environmental spending; can be pinned down by matching investment/abatement ratios across firms or time.
& Impact fund stated mandates; empirical sensitivities of flows/valuations to ESG/impact metrics. \\[0.4em]

$R_0$
& Initial cumulative abatement expenditure.
& Cumulative ``green'' capex/opex already spent/committed by the calibration date.
& Sustainability reporting / transition plans; internal accounting of climate-related spending. \\[0.4em]

$X_0$
& Initial profit / production capacity level.
& Set to current operating profit (e.g., EBITDA or operating cash flow) in chosen units; initial conditions used in simulations can be taken from firm data at the calibration date.
& Firm financial statements; accounting values. \\
\hline
\end{tabular}}
\label{calibration.tab}
\end{table}

In our model, the state variable $(X_t)$ is interpreted as  production capacity, assumed to be proportional to instantaneous profit with a one-to-one normalization and measured in monetary units.  For simplicity, investment ($\nu$) and abatement ($\eta$) are measured in the same ``production-capacity" units.   As a result, a convenient empirical normalization is to measure $(X_t)$ in, say, million EUR/year of EBITDA (or operating profit), $(\eta_t)$ in million EUR/year of decarbonization expenditure (opex+capex devoted to abatement), $(R_t)$ in cumulative million EUR of such spending, and $(\nu_t)$ in cumulative million EUR of investor capital injections. The variable $(R_t)$ is explicitly chosen as abatement expenditure (rather than direct emissions reductions) because it is comparatively easy to define and measure. Table \ref{calibration.tab} below summarizes the main parameters of the model and provides guidance on the calibratiobn of these parameters from observable quantities. 

\end{document}